\documentclass[a4paper]{scrartcl}

\usepackage[utf8]{inputenc}
\usepackage[T1]{fontenc}
\usepackage{microtype}
\usepackage{lmodern}

\usepackage{tikz}
\usetikzlibrary{arrows,calc,fit,shapes,automata,backgrounds,positioning}

\usepackage{graphicx}
\usepackage{xcolor}
\usepackage[fleqn]{amsmath}
\usepackage{amsfonts}
\usepackage{amssymb}
\usepackage{amsthm}
\usepackage{mathtools}
\usepackage{stmaryrd}
\usepackage{bm}
\usepackage{colortbl}
\usepackage{multirow}
\usepackage{import}

\usepackage{tabularx}
\usepackage[strings]{underscore}

\usepackage{authblk}
\usepackage{thm-restate}
\usepackage{bbm,booktabs,fontawesome}

\theoremstyle{plain}
\newtheorem{theorem}   {Theorem}[section]
\newtheorem{proposition} [theorem]{Proposition}
\newtheorem{lemma}       [theorem]{Lemma}
\newtheorem{corollary}   [theorem]{Corollary}

\theoremstyle{definition}
\newtheorem{example}     [theorem]{Example}
\newtheorem{definition}[theorem]{Definition}

\definecolor{niceredbright}{HTML}{bd0310}
\definecolor{nicebluebright}{HTML}{197b9b}
\definecolor{nicered}{HTML}{7f0a13}
\definecolor{niceblue}{HTML}{104354}
\definecolor{nicegreen}{HTML}{217516}
\definecolor{nicepurple}{HTML}{884bab}
\definecolor{nicebg}{HTML}{f6f0e4}
\definecolor{niceredlight}{HTML}{c9888d}
\definecolor{nicebluelight}{HTML}{78a4b8}
\definecolor{nicegreenlight}{HTML}{76de68}
\definecolor{nicepurplelight}{HTML}{bc87db}

\makeatletter
\RequirePackage[bookmarks,unicode,colorlinks=true]{hyperref}%
   \def\@citecolor{niceblue}%
   \def\@urlcolor{niceblue}%
   \def\@linkcolor{nicered}%

\def\orcidID#1{\smash{\href{http://orcid.org/#1}{\protect\raisebox{-1.25pt}{\protect\includegraphics{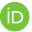}}}}}
\makeatother

\renewcommand{\subsubsection}{\paragraph}

\usepackage{enumitem}
\setlist[enumerate,1]{itemsep=0pt,topsep=1ex,before={\pagebreak[1]}}
\setlist[itemize,1]{itemsep=0pt,topsep=1ex}

\usepackage{etoolbox}
\makeatletter
\pretocmd\start@gather{%
    \if@minipage\kern-\topskip\kern-\baselineskip\kern+7pt\fi
}{}{}
\makeatother

\colorlet{darkred}{red!45!black}
\colorlet{darkgreen}{green!50!black}
\colorlet{darkblue}{blue!70!black}

\definecolor{detection} {HTML}{000000} 
\definecolor{acceptance}{HTML}{000000} 
\definecolor{selection} {HTML}{000000} 
\definecolor{fairness}  {HTML}{000000} 

\definecolor{fillD}  {HTML}{ddffdd}
\definecolor{fillA}  {HTML}{ddddff}
\definecolor{fillF}  {HTML}{ffdddd}
\definecolor{fillDA} {HTML}{ddffff}
\definecolor{fillDF} {HTML}{ffffdd}
\definecolor{fillAF} {HTML}{ffddff}
\definecolor{topDAF} {HTML}{ffffff}
\definecolor{botDAF} {HTML}{eeeeee}
\definecolor{drawD}  {HTML}{007700} 
\definecolor{drawA}  {HTML}{0000ff} 
\definecolor{drawF}  {HTML}{cc0000} 
\definecolor{drawDA} {HTML}{0077ff}
\definecolor{drawDF} {HTML}{cc7700}
\definecolor{drawAF} {HTML}{cc00ff}
\definecolor{drawDAF}{HTML}{444444}

\definecolor{colDAF}{HTML}{000000}
\definecolor{colDaF}{HTML}{3a3b04}
\definecolor{colDAf}{HTML}{093f3c}
\definecolor{coldAF}{HTML}{4a054b}
\definecolor{colDaf}{HTML}{023f0b}
\definecolor{coldAf}{HTML}{031455}
\definecolor{coldaf}{HTML}{470501}

\definecolor{niceredbright}{HTML}{bd0310}
\definecolor{nicebluebright}{HTML}{197b9b}
\definecolor{nicered}{HTML}{7f0a13}
\definecolor{niceblue}{HTML}{104354}
\definecolor{nicegreen}{HTML}{217516}
\definecolor{nicepurple}{HTML}{884bab}
\definecolor{nicebg}{HTML}{f6f0e4}
\definecolor{niceredlight}{HTML}{c9888d}
\definecolor{nicebluelight}{HTML}{78a4b8}
\definecolor{nicegreenlight}{HTML}{76de68}
\definecolor{nicepurplelight}{HTML}{bc87db}

\input{tex/utilities.tex}   
\NewCommand{\Naturals}{\mathbb{N}}
\NewCommand{\N}{\mathbb{N}}
\NewCommand{\Positives}{\mathbb{N}_{+}}
\NewCommand{\Integers}{\mathbb{Z}}
\NewCommand{\Z}{\mathbb{Z}}
\NewCommand{\PowerSet}[m]{2^{#1}}
\NewCommand{\Q}{\mathbb{Q}}
\NewCommand{\R}{\mathbb{R}}
\NewCommand{\NL}{\mathsf{NL}}
\NewCommand{\Tuple}[m]{(#1)}
\NewCommand{\bigTuple}[m]{\bigl(#1\bigr)}
\NewCommand{\BigTuple}[m]{\Bigl(#1\Bigr)}
\NewCommand{\Set}[m]{\{#1\}}
\NewCommand{\bigSet}[m]{\bigl\{#1\bigr\}}
\NewCommand{\SetBuilder}[mm]{\{#1:#2\}}
\NewCommand{\bigSetBuilder}[mm]{\bigl\{#1\bigm|#2\bigr\}}
\NewCommand{\BigSetBuilder}[mm]{\Bigl\{#1\Bigm|#2\Bigr\}}
\NewCommand{\Card}[m]{\operatorname{card}(#1)}
\NewCommand{\UpwardClosure}[m]{\lceil#1\rceil}
\NewCommand{\Range}[soms]{
  \IfBooleanTF{#1}{\langle}{[}%
  \IfValueT{#2}{#2{\,:\,}}%
  #3%
  \IfBooleanTF{#4}{\rangle}{]}%
}
\NewCommand{\DefEq}{\coloneqq} 

\NewCommand{\Alphabet}{\varLambda}
\NewCommand{\Automaton}{P'}
\NewCommand{\ExtendedAutomaton}{P}
\newcommand{\AdTrans}{A}
\NewCommand{\Configuration}[o]{\Switch{#1}{C}{C}{D}{E}}
\NewCommand{\ConfigurationSet}[o]{\Switch{#1}{\mathcal{C}}{\mathcal{C}}{\mathcal{D}}{\mathcal{E}}}
\NewCommand{\Run}{\rho}
\NewCommand{\Selection}[o]{\Switch{#1}{S}{S}{T}{}}
\NewCommand{\Schedule}{\sigma}
\NewCommand{\Scheduler}{\varSigma}
\NewCommand{\SelectionConstr}{s}
\NewCommand{\FairnessConstr}{f}
\NewCommand{\LanguageClass}{\mathcal{G}}
\NewCommand{\StarGraph}{\mathit{ST}}
\NewCommand{\TypeFont}[m]{\textup{\texttt{#1}}}
\NewCommand{\DetectionType}[m]{%
  \TypeFont{%
    \ifthenelse{\equal{#1}{set}}{d}{%
      \ifthenelse{\equal{#1}{multiset}}{\color{detection}D}{%
        \ifthenelse{\equal{#1}{*}}{*}{\Warning}%
      }%
    }%
  }%
}
\NewCommand{\AcceptanceType}[m]{%
  \TypeFont{%
    \ifthenelse{\equal{#1}{halting}}{a}{%
      \ifthenelse{\equal{#1}{stabilizing}}{\color{acceptance}A}{%
        \ifthenelse{\equal{#1}{*}}{*}{\Warning}%
      }%
    }%
  }%
}
\NewCommand{\SelectionType}[m]{%
  \TypeFont{%
    \ifthenelse{\equal{#1}{liberal}}{}{%
      \ifthenelse{\equal{#1}{synchronous}}{\color{selection}\$}{%
        \ifthenelse{\equal{#1}{exclusive}}{\color{selection}S}{%
          \ifthenelse{\equal{#1}{*}}{*}{\Warning}%
        }%
      }%
    }%
  }%
}
\NewCommand{\FairnessType}[m]{%
  \TypeFont{%
    \ifthenelse{\equal{#1}{weak}}{f}{%
      \ifthenelse{\equal{#1}{strong}}{\color{fairness}F}{%
        \ifthenelse{\equal{#1}{*}}{*}{\Warning}%
      }%
    }%
  }%
}
\NewCommand{\Type}[mmmm]{%
  \DetectionType{#1}%
  \AcceptanceType{#2}%
  \SelectionType{#3}%
  \FairnessType{#4}%
}

\NewCommand{\ModelName}[m]{\textcolor{col#1}{\texttt{#1}}}
\NewCommand{\DAsF}{\ModelName{DAF}}
\NewCommand{\DasF}{\ModelName{DaF}}
\NewCommand{\DAsf}{\ModelName{DAf}}
\NewCommand{\dAsF}{\ModelName{dAF}}
\NewCommand{\Dasf}{\ModelName{Daf}}
\NewCommand{\dAsf}{\ModelName{dAf}}
\NewCommand{\dasf}{\ModelName{daf}}
\NewCommand{\DAZf}{\textcolor{colDAf}{\texttt{DA\$}}}

\NewCommand{\DAF}{\ModelName{DAF}}
\NewCommand{\DaF}{\ModelName{DaF}}
\NewCommand{\DAf}{\ModelName{DAf}}
\NewCommand{\dAF}{\ModelName{dAF}}
\NewCommand{\Daf}{\ModelName{Daf}}
\NewCommand{\dAf}{\ModelName{dAf}}
\NewCommand{\daf}{\ModelName{daf}}
\NewCommand{\daF}{\texttt{daF}}

\NewCommand{\Cutoff}[mm]{\lceil#1\rceil_{#2}}
\NewCommand{\Supp}[m]{\Cutoff{#1}{1}}

\NewCommand{\trans}{\mapsto}
\NewCommand{\RVwait}{\text{\footnotesize\faHourglassStart}}
\NewCommand{\RVsearch}{\text{\small\faSearch}}
\NewCommand{\RVanswer}{\text{\small\faHandOUp}}
\NewCommand{\RVconfirm}{\text{\small\faCheck}}

\NewCommand{\clsTrivial}{\mathsf{Trivial}}
\NewCommand{\clsCutoff}{\mathsf{Cutoff}}
\NewCommand{\clsCutoffOne}{\mathsf{Cutoff}(1)}
\NewCommand{\NLinSpace}{\mathsf{NSPACE}(n)}
\NewCommand{\clsMajority}{\mathsf{Maj}}
\NewCommand{\clsInvSM}{\mathsf{ISM}}
\NewCommand{\clsL}{\mathsf{L}}
\renewcommand{\ldots}{...}
\renewcommand{\dots}{...}

\newcommand{\Proofsketch}{Proof (sketch)}
\newcommand{\parag}[1]{\paragraph*{#1.}}

\begin{document}

\title{Decision Power of Weak Asynchronous Models of Distributed Computing\thanks{This work was supported by an ERC Advanced Grant (787367: PaVeS) and by the Research Training Network of the Deutsche Forschungsgemeinschaft (DFG) (378803395: ConVeY).}}
\author{Philipp Czerner \orcidID{0000-0002-1786-9592}, Roland Guttenberg \orcidID{0000-0001-6140-6707},\\ Martin Helfrich \orcidID{0000-0002-3191-8098}, Javier Esparza \orcidID{0000-0001-9862-4919}}
\affil{\{czerner, guttenbe, helfrich, esparza\}@in.tum.de\\
Department of Informatics, TU München, Germany}
\maketitle

\begin{abstract}
Esparza and Reiter have recently conducted a systematic comparative study of models of distributed computing 
consisting of a network of identical finite-state automata that cooperate to decide if the underlying graph of the network satisfies a given property. The study classifies models according to four criteria, and shows that twenty-four initially possible combinations collapse into seven equivalence classes with respect to their decision power, i.e.\ the properties that the automata of each class can decide. However, Esparza and Reiter only show (proper) inclusions between the classes, and so do not characterise their decision power. In this paper we do so for \emph{labelling} properties, i.e.\ properties that depend only on the labels of the nodes, but not on the structure of the graph. In particular, majority (whether more nodes carry label $a$ than $b$) is a labelling property. Our results show that only one of the seven equivalence classes identified by Esparza and Reiter can decide majority for arbitrary networks. We then study the expressive power of the classes on bounded-degree networks, and show that three classes can. In particular, we present an algorithm for majority that works for all bounded-degree networks under adversarial schedulers, i.e.\  even if the scheduler must only satisfy that every node makes a move infinitely often, and prove that no such algorithm can work for arbitrary networks.

\end{abstract}


\section{Introduction} \label{sec:intro}

A common feature of networks of natural or artificial devices, like molecules, cells, microorganisms, or nano-robots, is that agents have very limited computational power and no identities. Traditional distributed computing models are often inadequate to study the power and efficiency of these networks, which has led to a large variety of new models, including population protocols \cite{AADFP06,angluin2005stably}, chemical reaction networks \cite{SoloveichikCWB08}, networked finite state machines \cite{EmekW13}, the weak models of distributed computing of \cite{HJKLLLSV15}, and the beeping model \cite{CK10,AfekABCHK13} (see e.g.\ \cite{FeinermanK13,NB15} for surveys and other models). 

These new models share several characteristics \cite{EmekW13}: the network can have an arbitrary topology; all nodes run the same protocol; each node has a finite number of states, independent of the size of the network or its topology; state changes only depend on the states of a bounded number of neighbours; nodes do not know their neighbours, in the sense of \cite{Angluin80}. Unfortunately, despite such substantial common ground, the models still exhibit much variability. 
In \cite{ER20} Esparza and Reiter have recently identified four fundamental criteria according to which they diverge:
\begin{itemize}
\item \textit{Detection}. In some models, nodes can only detect the \emph{existence} of neighbours in a certain state, e.g., \cite{AfekABCHK13,HJKLLLSV15}, while in others they can \emph{count} their number up to a fixed threshold, e.g., \cite{EmekW13,HJKLLLSV15}. 
\item \textit{Acceptance}. Some models compute by \emph{stable consensus}, requiring all nodes to eventually agree on the outcome of the computation, e.g. \cite{AADFP06,angluin2005stably,SoloveichikCWB08};  others require the nodes to produce an output and halt, e.g. \cite{HJKLLLSV15,KuhnLO10}. 
\item \textit{Selection}. Some models allow for
\emph{liberal selection}: at each moment, an arbitrary  subset of nodes is  selected to take a step~\cite{EmekW13,Reiter17}.  \emph{Exclusive} models (also called \emph{interleaving} models) select exactly one node  (or one pair of neighbouring nodes) \cite{AADFP06,angluin2005stably,SoloveichikCWB08}. \emph{Synchronous} models select all nodes at each step e.g.,\cite{HJKLLLSV15} or classical synchronous networks \cite{Lynch96}. 
\item \textit{Fairness}. Some models assume that selections are \emph{adversarial},
only satisfying the minimal requirement that each node is selected infinitely often \cite{Francez,LehmannPS81}. Others assume \emph{stochastic} or \emph{pseudo-stochastic} selection (meaning that selections satisfy a fairness assumption capturing the main features of a stochastic selection) \cite{AADFP06,angluin2005stably,SoloveichikCWB08}. In this case, the selection scheduler is a source of randomness that can be tapped by the nodes to ensure  e.g. that eventually all neighbours of a node will be in different states. 
\end{itemize}

In \cite{ER20}, Esparza and Reiter initiated a comparative study of the computational power of these models. They introduced 
 {\em distributed automata}, a generic formalism able to capture all combinations of the features above. A distributed automaton consists of a set of rules that tell the nodes of a labelled graph how to change their state depending on the states of their neighbours. Intuitively, the automaton describes an algorithm that allows the nodes to decide whether the graph satisfies a given property. The decision power of a class of automata is the set of graph properties they can decide, for example whether the graph contains more red nodes than blue nodes (the \emph{majority} property), or whether the graph is a cycle. The main result of \cite{ER20} was that the twenty-four classes obtained by combining the features above collapse into only seven equivalence classes w.r.t. their decision power. 
The collapse is a consequence of a fundamental result: the selection criterion does not affect the decision power. That is, the liberal, exclusive, or synchronous versions of a class with the same choices in the detection, acceptance, and fairness categories, have the same decision power.
The seven equivalence classes are shown on the left of Figure~\ref{fig:classes}, where \texttt{D} and \texttt{d} denote detection with and without the ability to count; \texttt{A} and \texttt{a} denote acceptance by stable consensus and by halting; and \texttt{F} and \texttt{f} denote pseudo-stochastic and adversarial fairness constraints.  So, for example, \DAf{} corresponds to the class of distributed automata in which agents can count, acceptance is by stable consensus, and selections are adversarial. (As mentioned above, the selection component is irrelevant, and one can assume for example that all classes have exclusive selection.)  Intuitively, the capital letter corresponds to the option leading to higher decision power. 

The results of \cite{ER20} only prove inclusions between classes and separations, but give no information on which properties can be decided by each class, an information available e.g.\ for multiple variants of population protocols \cite{AngluinAER07,angluin2005stably,Aspnes17,BlondinEJ19,GR09,MCS11}. In this paper, we characterise the decision power of all classes of \cite{ER20} w.r.t. \emph{labelling} properties, i.e.\ properties that depend only on the labels of the nodes.  Formally, given a labelled graph $G$ over a finite set $\Lambda$ of labels, let $L_G \colon \Lambda \rightarrow \N$ be the \emph{label count} of $G$ that assigns to each label the number of nodes carrying it. A labelling property is a set $\mathcal{L}$ of label counts. A graph $G$ satisfies $\mathcal{L}$ if $L_G \in \mathcal{L}$, and a distributed automaton decides $\mathcal{L}$ if it recognises exactly the graphs that satisfy $\mathcal{L}$. For example, the majority property is a labelling property, while the property of being a cycle is not. 

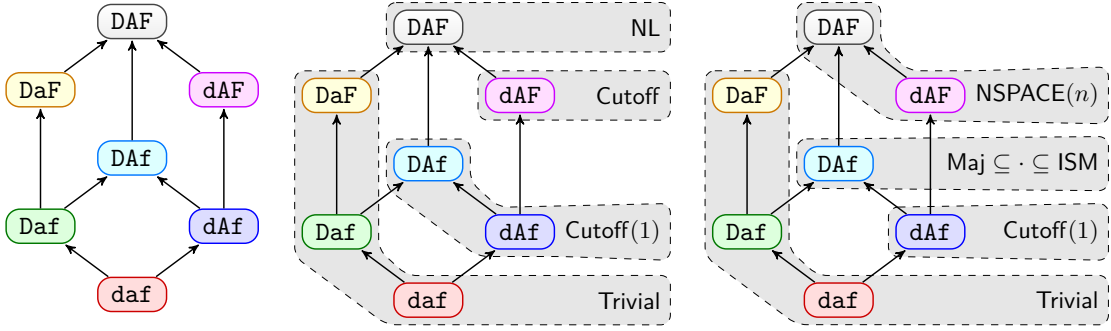
\begin{figure*}[t]
	\begin{minipage}{35mm}
		\scalebox{0.91}{\begin{tikzpicture}[semithick,>=stealth',shorten >=0.5pt,auto,on grid]
  \tikzstyle{class}=[draw,rounded corners=1.2ex,
                     inner sep=0ex,minimum height=3ex,minimum width=6ex]
  \tikzstyle{D}  =[draw=drawD, fill=fillD]
  \tikzstyle{A}  =[draw=drawA, fill=fillA]
  \tikzstyle{F}  =[draw=drawF, fill=fillF]
  \tikzstyle{DA} =[draw=drawDA,fill=fillDA]
  \tikzstyle{DF} =[draw=drawDF,fill=fillDF]
  \tikzstyle{AF} =[draw=drawAF,fill=fillAF]
  \tikzstyle{DAF}=[draw=drawDAF,shade,top color=topDAF,bottom color=botDAF]
  \def\dv{6ex}
  \def\dh{8ex}
  \def\setbordershift{1ex}

  \node [class,F]   (0000)                                   {\Type{set}{halting}{liberal}{weak}};
  \node [class,D]   (1000) [above left=\dv and \dh of 0000]  {\Type{multiset}{halting}{liberal}{weak}};
  \node [class,A]   (0010) [above right=\dv and \dh of 0000] {\Type{set}{stabilizing}{liberal}{weak}};
  \node [class,DF]  (1100) [above=2*\dv of 1000]               {\Type{multiset}{halting}{liberal}{strong}};
  \node [class,DA]  (1010) [above=2*\dv of 0000]             {\Type{multiset}{stabilizing}{liberal}{weak}};
  \node [class,AF]  (0110) [above=2*\dv of 0010]               {\Type{set}{stabilizing}{liberal}{strong}};
  \node [class,DAF] (1110) [above=2*\dv of 1010]               {\Type{multiset}{stabilizing}{liberal}{strong}};

  \draw[->,detection]
    (0000) edge (1000)
    (0010) edge (1010)
    (0110) edge (1110)
    ;

  \draw[->,acceptance]
    (0000) edge (0010)
    (1000) edge (1010)
    (1100) edge (1110)
    ;

  \draw[->,fairness]
    (1000) edge (1100)
    (0010) edge (0110)
    (1010) edge (1110)
    ;
    
  \begin{scope}[on background layer]
  \draw [draw=none] 
  ([shift={(-\setbordershift,\setbordershift)}] 0000.north west)
  --([shift={(\setbordershift,\setbordershift)}] 0000.north east)
  --([shift={(\setbordershift,-\setbordershift)}] 0000.south east)
  --([shift={(-\setbordershift,-\setbordershift)}] 0000.south west)
  --cycle;
  
  \end{scope}
\end{tikzpicture}

} 
	\end{minipage}\hspace{1.9mm}
	\begin{minipage}{51mm}
		\scalebox{0.91}{\begin{tikzpicture}[semithick,>=stealth',shorten >=0.5pt,auto,on grid]
  \tikzstyle{class}=[draw,rounded corners=1.2ex,
                     inner sep=0ex,minimum height=3ex,minimum width=6ex]
  \tikzstyle{D}  =[draw=drawD, fill=fillD]
  \tikzstyle{A}  =[draw=drawA, fill=fillA]
  \tikzstyle{F}  =[draw=drawF, fill=fillF]
  \tikzstyle{DA} =[draw=drawDA,fill=fillDA]
  \tikzstyle{DF} =[draw=drawDF,fill=fillDF]
  \tikzstyle{AF} =[draw=drawAF,fill=fillAF]
  \tikzstyle{DAF}=[draw=drawDAF,shade,top color=topDAF,bottom color=botDAF]
  \tikzstyle{set}=[rounded corners,dashed,fill=black!10]
  \tikzstyle{setlabel}=[]
  \def\dv{6ex}
  \def\dh{8ex}
  \def\setbordershift{0.6ex}
  \def\setshift{3.4}
  \def\setfontsize{\small}
  \def\movelabelr{\hspace*{-1mm}}

  \node [class,F]   (c0000)                                   {\Type{set}{halting}{liberal}{weak}};
  \node [class,D]   (c1000) [above left=\dv and \dh of c0000]  {\Type{multiset}{halting}{liberal}{weak}};
  \node [class,A]   (c0010) [above right=\dv and \dh of c0000] {\Type{set}{stabilizing}{liberal}{weak}};
  \node [class,DF]  (c1100) [above=2*\dv of c1000]               {\Type{multiset}{halting}{liberal}{strong}};
  \node [class,DA]  (c1010) [above=2*\dv of c0000]             {\Type{multiset}{stabilizing}{liberal}{weak}};
  \node [class,AF]  (c0110) [above=2*\dv of c0010]               {\Type{set}{stabilizing}{liberal}{strong}};
  \node [class,DAF] (c1110) [above=2*\dv of c1010]               {\Type{multiset}{stabilizing}{liberal}{strong}};
   
   \node[setlabel, anchor=east] (trivial) at (\setshift,0) {\setfontsize $\clsTrivial$\movelabelr};
   \node[setlabel, anchor=east] (support) at (\setshift,\dv) {\setfontsize $\clsCutoffOne$\movelabelr};   
   \node[setlabel, anchor=east] (cutoff) at (\setshift,3*\dv) {\setfontsize $\clsCutoff$\movelabelr};
   \node[setlabel, anchor=east] (NL) at (\setshift,4*\dv) {\setfontsize $\NL$\movelabelr};
   
   \draw[->,detection]
   (c0000) edge (c1000)
   (c0010) edge (c1010)
   (c0110) edge (c1110)
   ;
   
   \draw[->,acceptance]
   (c0000) edge (c0010)
   (c1000) edge (c1010)
   (c1100) edge (c1110)
   ;
   
   \draw[->,fairness]
   (c1000) edge (c1100)
   (c0010) edge (c0110)
   (c1010) edge (c1110)
   ;
   
   \begin{scope}[on background layer]
   \draw [set] 
   ([shift={(-\setbordershift,\setbordershift)}] c1010.north west)
   --([shift={(\setbordershift,\setbordershift)}] c1010.north east)
   --([shift={(\setbordershift,0)}] c1010.south east)
   --([shift={(\setbordershift+\setbordershift,\setbordershift)}] c0010.north west)
   --([shift={(\setbordershift,\setbordershift)}] support.north east)
   --([shift={(\setbordershift,-\setbordershift)}] support.south east)
   --([shift={(\setbordershift,-\setbordershift)}] c0010.south east)
   --([shift={(-\setbordershift,-\setbordershift)}] c0010.south west)
   --([shift={(-\setbordershift,-\setbordershift)}] c1010.south west)
   --cycle;
   
   \draw [set] 
   ([shift={(-\setbordershift,\setbordershift)}] c1100.north west)
   --([shift={(\setbordershift,\setbordershift)}] c1100.north east)
   --([shift={(\setbordershift,0)}] c1000.south east)
   --([shift={(\setbordershift+\setbordershift,\setbordershift)}] c0000.north west)
   --([shift={(\setbordershift,\setbordershift)}] trivial.north east)
   --([shift={(\setbordershift,-\setbordershift)}] trivial.south east)
   --([shift={(\setbordershift,-\setbordershift)}] c0000.south east)
   --([shift={(-\setbordershift,-\setbordershift)}] c0000.south west)
   --([shift={(-\setbordershift,-\setbordershift)}] c1000.south west)
   --([shift={(-\setbordershift,\setbordershift)}] c1000.north west)
   --([shift={(-\setbordershift,-\setbordershift)}] c1100.south west)
   --cycle;
   
   \draw [set] 
   ([shift={(-\setbordershift,\setbordershift)}] c0110.north west)
   --([shift={(\setbordershift,\setbordershift)}] cutoff.north east)
   --([shift={(\setbordershift,-\setbordershift)}] cutoff.south east)
   --([shift={(-\setbordershift,-\setbordershift)}] c0110.south west)
   --cycle;
   
   \draw [set] 
   ([shift={(-\setbordershift,\setbordershift)}] c1110.north west)
   --([shift={(\setbordershift,\setbordershift)}] NL.north east)
   --([shift={(\setbordershift,-\setbordershift)}] NL.south east)
   --([shift={(-\setbordershift,-\setbordershift)}] c1110.south west)
   --cycle;
   
   \end{scope}
\end{tikzpicture}

}
	\end{minipage}\hspace{1.9mm}
	\begin{minipage}{55mm}
		\scalebox{0.91}{\begin{tikzpicture}[semithick,>=stealth',shorten >=0.5pt,auto,on grid]
  \tikzstyle{class}=[draw,rounded corners=1.2ex,
                     inner sep=0ex,minimum height=3ex,minimum width=6ex]
  \tikzstyle{D}  =[draw=drawD, fill=fillD]
  \tikzstyle{A}  =[draw=drawA, fill=fillA]
  \tikzstyle{F}  =[draw=drawF, fill=fillF]
  \tikzstyle{DA} =[draw=drawDA,fill=fillDA]
  \tikzstyle{DF} =[draw=drawDF,fill=fillDF]
  \tikzstyle{AF} =[draw=drawAF,fill=fillAF]
  \tikzstyle{DAF}=[draw=drawDAF,shade,top color=topDAF,bottom color=botDAF]
  \tikzstyle{set}=[rounded corners,dashed,fill=black!10]
  \tikzstyle{setlabel}=[]
  \def\dv{6ex}
  \def\dh{8ex}
  \def\setbordershift{0.6ex}
  \def\setshift{3.8}
  \def\setfontsize{\small}
  \def\movelabelr{\hspace*{-1mm}}

  \node [class,F]   (c0000)                                   {\Type{set}{halting}{liberal}{weak}};
  \node [class,D]   (c1000) [above left=\dv and \dh of c0000]  {\Type{multiset}{halting}{liberal}{weak}};
  \node [class,A]   (c0010) [above right=\dv and \dh of c0000] {\Type{set}{stabilizing}{liberal}{weak}};
  \node [class,DF]  (c1100) [above=2*\dv of c1000]               {\Type{multiset}{halting}{liberal}{strong}};
  \node [class,DA]  (c1010) [above=2*\dv of c0000]             {\Type{multiset}{stabilizing}{liberal}{weak}};
  \node [class,AF]  (c0110) [above=2*\dv of c0010]               {\Type{set}{stabilizing}{liberal}{strong}};
  \node [class,DAF] (c1110) [above=2*\dv of c1010]               {\Type{multiset}{stabilizing}{liberal}{strong}};
   
   \node[setlabel, anchor=east] (trivial) at (\setshift,0) {\setfontsize $\clsTrivial$\movelabelr};
   \node[setlabel, anchor=east] (support) at (\setshift,\dv) {\setfontsize $\clsCutoffOne$\movelabelr}; 
   \node[setlabel, anchor=east] (InvUnderScalar) at (\setshift,2*\dv) {\setfontsize $\clsMajority \subseteq \cdot \subseteq \clsInvSM$\movelabelr};  
   \node[setlabel, anchor=east] (linspace) at (\setshift,3*\dv) {\setfontsize $\NLinSpace$\movelabelr};
   
   \draw[->,detection]
   (c0000) edge (c1000)
   (c0010) edge (c1010)
   (c0110) edge (c1110)
   ;
   
   \draw[->,acceptance]
   (c0000) edge (c0010)
   (c1000) edge (c1010)
   (c1100) edge (c1110)
   ;
   
   \draw[->,fairness]
   (c1000) edge (c1100)
   (c0010) edge (c0110)
   (c1010) edge (c1110)
   ;
   
   \begin{scope}[on background layer]
   
   \draw [set] 
   ([shift={(-\setbordershift,\setbordershift)}] c1100.north west)
   --([shift={(\setbordershift,\setbordershift)}] c1100.north east)
   --([shift={(\setbordershift,0)}] c1000.south east)
   --([shift={(\setbordershift+\setbordershift,\setbordershift)}] c0000.north west)
   --([shift={(\setbordershift,\setbordershift)}] trivial.north east)
   --([shift={(\setbordershift,-\setbordershift)}] trivial.south east)
   --([shift={(\setbordershift,-\setbordershift)}] c0000.south east)
   --([shift={(-\setbordershift,-\setbordershift)}] c0000.south west)
   --([shift={(-\setbordershift,-\setbordershift)}] c1000.south west)
   --([shift={(-\setbordershift,\setbordershift)}] c1000.north west)
   --([shift={(-\setbordershift,-\setbordershift)}] c1100.south west)
   --cycle;
   
   \draw [set] 
   ([shift={(-\setbordershift,\setbordershift)}] c0010.north west)
   --([shift={(\setbordershift,\setbordershift)}] support.north east)
   --([shift={(\setbordershift,-\setbordershift)}] support.south east)
   --([shift={(-\setbordershift,-\setbordershift)}] c0010.south west)
   --cycle;
   
   \draw [set] 
   ([shift={(-\setbordershift,\setbordershift)}] c1010.north west)
   --([shift={(\setbordershift,\setbordershift)}] InvUnderScalar.north east)
   --([shift={(\setbordershift,-\setbordershift)}] InvUnderScalar.south east)
   --([shift={(-\setbordershift,-\setbordershift)}] c1010.south west)
   --cycle;
   

	\draw [set] 
	([shift={(-\setbordershift,\setbordershift)}] c1110.north west)
	--([shift={(\setbordershift,\setbordershift)}] c1110.north east)
	--([shift={(\setbordershift,0)}] c1110.south east)
	--([shift={(\setbordershift+\setbordershift,\setbordershift)}] c0110.north west)
	--([shift={(\setbordershift,\setbordershift)}] c0110.north east)
	--([shift={(\setbordershift,\setbordershift)}] linspace.north east)
	--([shift={(\setbordershift,-\setbordershift)}] linspace.south east)
	--([shift={(-\setbordershift,-\setbordershift)}] c0110.south west)
	--([shift={(-\setbordershift,-\setbordershift)}] c1110.south west)
	--cycle;
   
   \end{scope}
\end{tikzpicture}

}
	\end{minipage}
	
	\caption{The seven distributed automata models of \cite{ER20}; their decision power w.r.t.\ labelling predicates for arbitrary networks, and for bounded-degree networks. $\clsInvSM$ stands for \emph{invariant under scalar multiplication}. The other complexity classes are defined in Section \ref{sec:unrestricted}.}
	\label{fig:classes}
\end{figure*}

Our first collection of results is shown in the middle of Figure~\ref{fig:classes}. We prove that all classes with halting acceptance can only decide the trivial labelling properties $\emptyset$ and $\N^\Lambda$. More surprisingly, we further prove that the computational power of \DAf, \dAf, and \dAF{} is very limited. Given a labelled graph $G$ and a number $K$, let $\Cutoff{L_G}{K}$ be the result of substituting $K$ for every component of $L_G$ larger than $K$.
The classes \DAf, \dAf{} can decide a property $\mathcal{L}$ if{}f membership of $L_G$ in $\mathcal{L}$ depends only on $\Cutoff{L_G}{1}$, and 
\dAF\ if{}f membership depends only on $\Cutoff{L_G}{K}$ for some $K \geq 1$. In particular, none of these classes can decide majority.  Finally, moving to the top class \DAF{} causes a large increase in expressive power: \DAF{} can decide exactly the labelling properties in the complexity class $\NL$, i.e.\ the properties $\mathcal{L}$ such that a nondeterministic Turing machine can decide membership of $L_G$ in $\mathcal{L}$ using logarithmic space in the number of nodes of $G$. In particular, \DAF-automata can decide majority, or whether the graph has a prime number of nodes.

In the last part of the paper, we obtain our second and most interesting collection of results. Molecules, cells, or  microorganisms typically have short-range communication mechanisms, which puts an upper bound on their number of communication partners. So we re-evaluate the decision power of the classes for bounded-degree networks, as also done in \cite{angluin2005stably} for population protocols on graphs. Intuitively, nodes know that they have at most $k$ neighbours for some fixed number $k$, and can exploit this fact to decide more properties.
Our results are shown on the right of Figure~\ref{fig:classes}. Both \DAF{} and \dAF{} boost their expressive power to $\NLinSpace$, where $n$ is the number of nodes of the graph. This is the theoretical upper limit since each node has a constant number of bits of memory. Further, the class \DAf{} becomes very interesting. While we are not yet able to completely characterise its expressive power, we show that it can only decide properties 
\emph{invariant under scalar multiplication} ($\clsInvSM$), i.e.\ labelling properties $\mathcal{L}$ such that $L_G \in \mathcal{L}$ if{}f $\lambda \cdot L_G \in \mathcal{L}$ for every $\lambda \in \N$, and that it can decide all properties satisfied by a graph $G$ if{}f $L_G$ is a solution to a system of homogeneous linear inequalities. In particular,  \DAf{} can decide majority, and we have the following surprising fact. If nodes have no information about the network, then they require stochastic-like selection to decide majority; however, if they know an upper bound on the number of their neighbours, they can decide majority even with adversarial selection. In particular, there is a synchronous majority algorithm for bounded-degree networks.

\subsubsection*{Related work.} Decision power questions have also been studied in~\cite{HJKLLLSV15} for a model similar to \Daf{},
and in \cite{ChatzigiannakisMNS13} for a graph version of the mediated population protocol model \cite{MCS11}.  The distinguishing feature of our work is the systematic study of the influence of a number of features on the decision power. 

Further, there exist numerous results about the decision power of different classes of population protocols. Recall that agents of population protocols are indistinguishable and communicate by rendez-vous; this is equivalent to placing the agents in a clique and selecting an edge at every step. Angluin \textit{et al.} show that standard population protocols compute exactly the semilinear predicates
~\cite{AngluinAER07}. Extensions with  absence detectors or cover-time services~\cite{MichailS15}, consensus-detectors~\cite{Aspnes17}, or broadcasts~\cite{BlondinEJ19} increase the power to $\NL$ (more precisely, in the case of~\cite{MichailS15} the power lies between $\clsL$ and $\NL$). Our result $\DAF=\NL$ shows that these features can be replaced by a counting capability. Further, giving an upper bound on the number of neighbours increases the decision power to $\NLinSpace$, a class only reachable by standard population protocols (not on graphs) if agents have identities, or channels have memory~\cite{MCS11,GR09}.

\subsubsection*{Structure of the paper.} Section \ref{sec:pre} recalls the automata models and the results of \cite{ER20}. Section \ref{sec:limitations} presents fundamental limitations of their decision power. Section \ref{sec:extensions} introduces a notion of simulating an automaton by another, and uses it to show that distributed automata with more powerful communication mechanisms can be simulated by standard automata.
Section \ref{sec:unrestricted} combines the results of Sections \ref{sec:limitations} and \ref{sec:extensions} to characterise the decision power of the models of \cite{ER20} on labelling properties (middle of Figure \ref{fig:classes}). Section \ref{sec:degree-bounded} does the same for bounded-degree networks (right of Figure \ref{fig:classes}).

Due to the nature of this research, we need to state and prove many results. For the sake of brevity, each section concentrates on the most relevant result; all others are only stated, and their proofs are given in the appendix.

\section{Preliminaries} \label{sec:pre}
Given sets $X, Y$,
we denote by $\PowerSet{X}$ the power set of $X$,
and by $Y^X$ the set of functions~${X \to Y}$.
We define a closed interval
$\Range[m]{n} \DefEq \SetBuilder{i \in \Integers}{m \leq i \leq n}$
and
$\Range{n} \DefEq \Range[0]{n}$,
for any $m, n \in \Integers$ such that $m \leq n$.

A \emph{multiset} over a set $X$ is an element of $\N^X$. Given a multiset $M \in \N^X$ and $\beta\in\N$,  we let $\Cutoff{M}{\beta}$ denote the multiset given by
$\Cutoff{M}{\beta}(x):=M(x)$ if $M(x)<\beta$ and $\Cutoff{M}{\beta}(x):=\beta$ otherwise. We say that $\Cutoff{M}{\beta}$ is the result of \emph{cutting off} $M$ at $\beta$,
and call the function that assigns $\Cutoff{M}{\beta}$ to $M$ the \emph{cutoff function} for $\beta$. 

Let~$\Alphabet$ be a finite set.
A \intro{($\Alphabet$-labelled, undirected) graph} is a triple
$G = \Tuple{V, E, \lambda}$,
where
$V$~is a finite nonempty set of \intro{nodes},
$E$~is a set of undirected \intro{edges} of the form
$e = \Set{u, v} \subseteq V$
such that $u \neq v$,
and
$\lambda \colon V \to \Alphabet$
is a \intro{labelling}. 

\subsubsection*{Convention} Throughout the paper, 
all graphs are labelled, have at least three nodes, and are connected.

\subsection{Distributed automata}

Distributed automata \cite{ER20} take a graph as input, and either accept or reject it.
We first define distributed machines.

\parag{Distributed machines}
Let~$\Alphabet$ be a finite set of \emph{labels}
and let $\beta \in \Positives$.
A \intro{(distributed) machine}
with \emph{input alphabet}~$\Alphabet$ and \emph{counting bound}~$\beta$
is a tuple
$M~=~\Tuple{Q, \delta_0, \delta, Y, N}$,
where
$Q$~is a finite set of \intro{states},
$\delta_0 \colon \Alphabet \to Q$
is an \intro{initialisation function},
$\delta \colon Q \times \Range{\beta}^{Q} \to Q$
is a \intro{transition function},
and $Y, N \subseteq Q$ are two disjoint sets
of \emph{accepting} and \emph{rejecting} states,
respectively. Intuitively, when $M$ runs on a graph, each node $v$ (or \textit{agent}) with label $\gamma$ is initially in state $\delta_0(\gamma)$ and uses~$\delta$ to update its state, 
depending on the number of neighbours it has in each state; however $v$ can only detect if it has $0, 1, \ldots, (\beta-1)$,
or at least $\beta$ neighbours in a given state. We call $\beta$ the \emph{counting bound} of $M$.

Transitions given by $\delta$ are called \emph{neighbourhood transitions}. We write $q,\mathcal{N}\mapsto q'$ for $\delta(q,\mathcal{N})=q'$. If $q=q'$ the transition is \emph{silent} and may not be explicitly specified in our constructions.
Sometimes $\delta_0,Y,N$ are also irrelevant and not specified, and we just write $M=(Q,\delta)$. 

\parag{Selections, schedules, configurations, runs, and acceptance}
A \emph{selection} of a graph $G = \Tuple{V, E, \lambda}$ is a set~$\Selection \subseteq V$. A
\emph{schedule} is an infinite sequence of selections
$\Schedule = (\Selection_0, \Selection_1, \Selection_2, \ldots) \in (2^V)^\omega$ such that for every~$v \in V$, there exist infinitely many $t \geq 0$ such that $v \in \Selection_t$. 
Intuitively, $\Selection_t$ is the set of nodes activated by the scheduler at time $t$, and schedules must activate every node infinitely often.

A \intro{configuration} of  $M = \Tuple{Q, \delta_0, \delta, Y, N}$ on~$G$ is a mapping
$\Configuration \colon V \to Q$.
We let 
$N_v^{\Configuration} \colon Q \to \Range{\beta}$
denote the \emph{neighbourhood function} that assigns to each $q \in Q$
the number of neighbours of $v$ in state~$q$ at configuration $C$,
up to threshold~$\beta$; in terms of the cutoff function,
$N_v^{\Configuration} = \Cutoff{M_v^C}{\beta}$, where $M_v^C(q) = \big|  \SetBuilder{u}{\Set{u, v} \in E \land \Configuration(u) = q} \big|$. 
The \intro{successor configuration}
of~$\Configuration$ via a selection~$\Selection$
is the configuration $\mathit{succ}_{\delta}(\Configuration, \Selection)$
obtained from~$\Configuration$ by letting all nodes in~$\Selection$ evaluate~$\delta$ simultaneously, and
keeping the remaining nodes idle.
Formally, $\mathit{succ}_{\delta}(\Configuration, \Selection)(v) = \delta
    \bigl(
    \Configuration(v), N_v^{\Configuration}
    \bigr)$ if $v \in S$ and $\mathit{succ}_{\delta}(\Configuration, \Selection)(v) = \Configuration(v)$ if $v \in V \setminus \Selection$.
We write $C\rightarrow C'$ if $C'=\mathit{succ}(C,S)$ for some selection $S$, and $\rightarrow^*$ for the reflexive and transitive closure of $\rightarrow$.
Given a schedule
$\Schedule = (\Selection_0, \Selection_1, \Selection_2, \ldots)$,
the \intro{run} of~$M$ on~$G$ scheduled by~$\Schedule$
is the infinite sequence
$\Tuple{\Configuration_0, \Configuration_1, \Configuration_2, \dots}$
of configurations defined inductively as follows: $C_0(v) = \delta_0(\lambda(v))$ for every node $v$, and  
$\Configuration_{t+1} = \mathit{succ}_{\delta}(\Configuration_t, \Selection_t)$.
We call $\Configuration_0$ the \textit{initial configuration}. A configuration $\Configuration$ is \intro{accepting} if $\Configuration(v) \in Y$ for every $v \in V$, and \intro{rejecting} if $\Configuration(v) \in N$ for every $v \in V$.  
A run~$\Run = \Tuple{\Configuration_0, \Configuration_1, \Configuration_2, \dots}$ of~$M$ on~$G$ is \intro{accepting} resp. \intro{rejecting} if there is  $t \in \Naturals$ such that $\Configuration_{t'}$ is accepting resp.\ rejecting for every $t' \geq t$. This is called acceptance by \emph{stable consensus} in \cite{AADFP06}.

\parag{Distributed automata} A \emph{scheduler} is a pair $\Scheduler = (\SelectionConstr,\FairnessConstr)$, where $\SelectionConstr$ is a \intro{selection constraint} that assigns to every graph  $G=\Tuple{V, E, \lambda}$ a set~$\SelectionConstr(G) \subseteq 2^V$ of \emph{permitted selections} such that every node $v \in V$ occurs in at least one selection $\Selection \in \SelectionConstr(G)$, and $\FairnessConstr$ is a \emph{fairness constraint} that assigns to every graph $G$ a set $\FairnessConstr(G) \subseteq \SelectionConstr(G)^\omega$ of \intro{fair schedules} of $G$.
We call the runs of a machine with schedules in $\FairnessConstr(G)$ \emph{fair runs}
(with respect to~$\Scheduler$).

A \emph{distributed automaton} is a pair $A=(M,\Scheduler)$, where $M$ is a machine 
and $\Scheduler$ is a scheduler satisfying the \emph{consistency condition}: for every graph $G$, either all fair runs of $M$ on $G$ are accepting, or all fair runs of $M$ on $G$ are rejecting.
Intuitively, whether $M$ accepts or rejects $G$ is independent of the scheduler's choices.
$A$ \intro{accepts} $G$ if some fair run of $A$ on $G$ is accepting, and \intro{rejects} $G$ otherwise. The language $L(A)$ of $A$ is the set of graphs it recognises. The \emph{property decided} by $A$ is the predicate $\varphi_A$ on graphs such that $\varphi_A(G)$ holds if{}f $G \in L(A)$.  Two automata are equivalent if they decide the same property.

\subsection{Classifying distributed automata.} 

Esparza and Reiter classify automata according to four criteria: detection capabilities, acceptance condition,  selection, and fairness. The first two concern the distributed machine, and the last two the scheduler.  

\parag{Detection} Machines with counting bound $\beta = 1$ or $\beta \geq 1$ are called \emph{non-counting} or \emph{counting}, respectively (abusing language, non-counting is considered a special case of counting).

\parag{Acceptance}  A machine is \emph{halting} if its transition function does not allow nodes to leave accepting or rejecting states, i.e.\ $\delta(q, P) = q$ for every $q \in Y \cup N$ and every 
$P \in \Range{\beta}^{Q}$. Intuitively, a node that enters an accepting/rejecting state cannot change its mind later. Halting acceptance is a special case of acceptance by stable consensus. 

\parag{Selection} A scheduler  $\Scheduler = (\SelectionConstr,\FairnessConstr)$ is \emph{synchronous} if $\SelectionConstr(G) = \Set{V}$ for every $G = \Tuple{V, E, \lambda}$ (at each step all nodes make a move); \emph{exclusive} if $\SelectionConstr(G) = \{ \{v\} \mid v \in V \}$
(at each step exactly one node makes a move); and  \emph{liberal} if $\SelectionConstr(G) = 2^V$ (at every step some set of nodes makes a move). 

\parag{Fairness} 
A schedule $\Schedule = (\Selection_0, \Selection_1, \ldots) \in \SelectionConstr(G)^\omega$ of a graph $G$ is \emph{pseudo-stochastic} if for every finite sequence $(\Selection[2]_0, \ldots, \Selection[2]_n) \in \SelectionConstr(G)^*$ there exist infinitely many $t \geq 0$ such that $(\Selection[1]_t, \ldots, \Selection[1]_{t+n}) = (\Selection[2]_0,  \ldots, \Selection[2]_n)$.  Loosely speaking, every possible finite sequence of selections is scheduled infinitely often. 
A scheduler $\Scheduler =(\SelectionConstr,\FairnessConstr)$ is \emph{adversarial}  if for every graph $G$, the set $\FairnessConstr(G)$ contains all schedules of  $\SelectionConstr(G)^\omega$ (i.e. we only require every node to be selected infinitely often), and \emph{pseudo-stochastic} if it contains precisely the pseudo-stochastic schedules. 

Whether or not a schedule $\sigma$ of a graph $G=\Tuple{V, E, \lambda}$ is pseudo-stochastic depends on $\SelectionConstr(G)$. For example,  if $\SelectionConstr(G) = \{ V \}$, i.e.\ if the only permitted selection is to select all nodes, 
then the synchronous schedule $V^\omega$ is pseudo-stochastic, but if $\SelectionConstr(G) = 2^V$, i.e.\ if all selections are permitted, then it is not. 

\smallskip This classification yields 24 classes of automata (four classes of machines and six classes of schedulers). 
It was shown in \cite{ER20} that the decision power of a class is independent of the selection type of the scheduler
(liberal, exclusive, or synchronous).  This leaves 8 classes, which we denote using the following scheme:

\begin{center}
  \begin{tabular}{l@{\hspace{5ex}}l@{\hspace{5ex}}l@{\hspace{5ex}}}
    \emph{Detection}     & \emph{Acceptance}      & \emph{Fairness} \\
    \midrule
    \DetectionType{set}: non-counting  & \AcceptanceType{halting}: halting       & \FairnessType{weak}: adversarial scheduling \\
    \DetectionType{multiset}: counting & \AcceptanceType{stabilizing}: stable consensus    & \FairnessType{strong}: pseudo-stochastic scheduling
  \end{tabular}
\end{center}

\smallskip\noindent Intuitively, the uppercase letter corresponds to the more powerful variant. Each class of automata is denoted by a string
$xyz \in
\Set{\DetectionType{set},\DetectionType{multiset}}
\times
\Set{\AcceptanceType{halting},\AcceptanceType{stabilizing}}
\times
\Set{\FairnessType{weak},\FairnessType{strong}}$. Finally, it was shown in \cite{ER20} that $\daf$ and $\daF$ have the same decision power, yielding the seven classes on the left of Figure \ref{fig:classes}.

In the rest of the paper, we generally assume that selection is exclusive (exactly one node is selected at each step). Since for synchronous automata there is only one permitted selection, adversarial and pseudo-stochastic scheduling coincide, and we therefore denote synchronous classes by strings $xy$\texttt{\$}; for example, we write \DAZf.

\section{Limitations} \label{sec:limitations}

Our lower bounds on the decision power of the seven classes follow from several lemmata 
proving limitations of their \emph{discriminating} power, i.e.\  of their ability to distinguish two graphs by accepting the one and rejecting the other.
We present four limitations. We state the first three, and prove the last one, a non-trivial limitation of \dAF-automata.
Recall that $\varphi_A$ denotes the property decided by the automaton $A$.

\subsubsection*{Automata with halting acceptance cannot discriminate cyclic graphs.}
Automata with halting acceptance necessarily accept all graphs containing a cycle, or reject all graphs containing a cycle.
Intuitively, given two graphs $G$ and $H$ with cycles, if one is accepted and the other rejected, one can construct a larger graph in which some nodes
behave as if they were in $G$, others as if they were in $H$. This makes some nodes accept and others reject, contradicting that for every graph the automaton accepts or rejects.

\begin{restatable}{lemma}{DasFLimitation}
Let $A$ be a \DasF-automaton. For all graphs $G$ and $H$ containing a cycle, $\varphi_A(G) = \varphi_A(H)$. \label{lem:trivial-on-cycles}
\end{restatable}

\subsubsection*{Automata with adversarial selection cannot discriminate a graph and its covering.}
Given two graphs $G = \Tuple{V_G, E_G, \lambda_G}$ and $H = \Tuple{V_H, E_H, \lambda_H}$, we say that $H$ \emph{covers} $G$ if there is a \emph{covering map} $f \colon V_H \rightarrow V_G$, i.e.\  a surjection that preserves labels and neighbourhoods by mapping the neighbourhood of each $v$ in $H$ bijectively onto the neighbourhood of $f(v)$ in $G$. Automata with adversarial selection cannot discriminate a graph from another one covering it. Intuitively, if $H$ covers $G$ then a node $u$ of $H$ and the node $f(u)$ of $G$ visit the same sequence of states in the synchronous runs of $A$ on $G$ and $H$. Since these runs are fair for adversarial selection, both nodes accept, or both reject.

\begin{restatable}{lemma}{DAfCovering}
Let $A$ be a \DAsf-automaton. For all graphs $G$ and $H$,  if \( H\) is a covering of \(G\), then $\varphi_A(G) = \varphi_A(H)$. \label{lem:cannot-distinguish-covering}
\end{restatable}

Let $L_G \colon \Alphabet \rightarrow \N$ assign to each label 
$\ell \in \Alphabet$ the number of nodes $v \in V$ such that $\lambda(v) = \ell$. We call $L_G$ the 
\emph{label count} of $G$. Recall that a labelling property depends only on the label count of a graph, not on its structure. Based on the existence of a $\lambda$-fold covering graph for every $G$ and $\lambda\in\N$, we immediately get the following.

\begin{restatable}{corollary}{DAfScalarMultiplication}
Let $A$ be a \DAsf-automaton deciding a labelling property. For all graphs $G$ and $H$, if $L_H =\lambda L_G$  for some  \( \lambda \in \N_{>0} \), then \( \varphi_A(G)=\varphi_A(H)\). This also holds when restricting to $k$-degree-bounded graphs.
\label{cor:closed-under-scalar-multiplication}
\end{restatable}

\subsubsection*{Automata with adversarial selection  and non-counting automata cannot discriminate beyond a cutoff.}

Our final results show that for every  \DAsf- or \dAsF-automaton deciding a labelling property there is a number $K$ such that whether the automaton accepts a graph $G$ or not depends only on $\Cutoff{L_G}{K}$, and not on the ``complete'' label count $L_G$. In such a case we say that the property admits a cutoff.  For \DAf-automata, the cutoff $K$ is simply $\beta +1$, where $\beta$ is the counting bound.

\begin{restatable}{lemma}{DAfCutoff}
Let $A$ be a \DAsf-automaton with counting bound $\beta$ that decides a labelling property. For all graphs $G$ and $H$,  if \(\Cutoff{L_G}{\beta+1}=\Cutoff{L_H}{\beta+1}\) then $\varphi_A(G) = \varphi_A(H)$, i.e. \(\varphi_A\) admits a cutoff. \label{lem:finite-cutoff}
\end{restatable}

The proof that  \dAF-automata also cannot discriminate beyond a cut-off is more involved, and the cutoff value $K$ is a complex function of the automaton. The proof technique is similar to that of Theorem 39 of \cite{AngluinAER07}.

\begin{restatable}{lemma}{dAFCutoff}
\label{lem:dAFCutoff}
Let $A$ be a \dAF-automaton that decides a labelling property. There exists $K \geq 0$ such that for every graph $G$ and $H$, if \(\Cutoff{L_G}{K}=\Cutoff{L_H}{K}\) then $\varphi_A(G) = \varphi_A(H)$, i.e. \( \varphi_A\) admits a cutoff. \label{lem:finite-cutoff-DAsf}
\end{restatable}
\begin{proof}[\Proofsketch]
\newcommand{\upcl}[1]{\lceil #1 \rceil}
\newcommand{\ctr}[1]{#1_{\text{ctr}}}
\newcommand{\sco}[1]{#1_{\text{sc}}}

Let $A$ be a \dAF-automaton, and let $Q$ be its set of states. In this proof we consider the class of \emph{star graphs}. A star is a graph in which a node called the \emph{centre} is connected to an arbitrary number of nodes called the \emph{leaves}, and no other edges exist. Importantly, for every graph \(G\), there is a star \(G'\) with the same label count. We consider labelling properties (which do not depend on the graph), so if the property has a cutoff for star graphs, then the property has a cutoff in general. A configuration of a star graph $G$ is completely determined by the state of the centre and the number of leaves in each state. So in the rest of the proof we assume that such a configuration is a pair $C= (\ctr{C}, \sco{C})$, where $\ctr{C}$ denotes the state of the centre of $G$, and $\sco{C}$ is the \emph{state count} of $C$, i.e.\  the mapping that assigns to each $q \in Q$ the number $\sco{C}(q)$ of leaves of $G$ that are in state $q$ at $C$. We denote the \emph{cutoff} of $C$ at a number \(m\) as $\Cutoff{C}{m}:=(\ctr{C}, \Cutoff{\sco{C}}{m})$.

Given a configuration $C$ of $A$, recall that $C$ is rejecting if all nodes have rejecting states. We say that $C$ is \emph{stably rejecting} if $C$ can only reach configurations which are rejecting. Given an initial configuration $C_0$, it is clear that $A$ must reject if it can reach a stably rejecting configuration $C$ from $A$. Conversely, if it cannot reach such a $C$, then $A$ will not reject $C_0$, as there is a fair run starting at $C_0$ which contains infinitely many configurations that are not rejecting.

In the appendix we now use Dickson's Lemma to show that there is a constant $m$ s.t.\ a configuration $C$ of $A$ on a star is stably rejecting iff $\Cutoff{C}{m}$ is. For this it is crucial that for stars stable rejection is \emph{downwards closed} in the following sense: if such a $C$ is stably rejecting and has at least two leaves in a state $q$, then the configuration $C'$ that results from removing one of these leaves is still stably rejecting.

Now, let $C=(\ctr{C},\sco{C})$ denote a configuration of $A$ on a star $G=(V,E)$, and let $q$ denote a state with $\sco{C}(q)\ge |Q|(m-1)+1$. We will show: if $A$ rejects $C$ then it must also reject the configuration $C'=(\ctr{C'},\sco{C'})$ which results from adding a leaf $v_\mathit{new}$ in state $q$ to $G$, i.e.\ $\ctr{C'}:=\ctr{C}$, $\sco{C'}(q):=\sco{C}(q)+1$, and $\sco{C'}(r):=\sco{C}(r)$ for states $r\ne q$.

We know that $A$ rejects $C$, so there is some stably rejecting configuration $D$ reachable from $C$. Our goal is to construct a configuration $D'$ reachable from $C'$ which fulfils $\Cutoff{D}{m}=\Cutoff{D'}{m}$, implying that $D'$ would also be stably rejecting. For this, let $S\subseteq V$ denote the leaves of $G$ which are in state $q$ in $C$. There are $|Q|$ states and $(m-1)|Q|+1$ nodes in $S$, so by the pigeonhole principle there is a state $r\in Q$ s.t. in configuration $D$ at least $m$ nodes in $S$ are in state $r$. Let $v_\mathit{old}$ denote one of these nodes.

To get $D'$, we construct a run starting from $C'$, where $v_\mathit{new}$ behaves exactly as $v_\mathit{old}$, until $D'$ is reached. 
Afterwards, the nodes may diverge because of the pseudo-stochastic scheduler. However, this does not matter as $D'$ is stably rejecting.

Let $\rho=(v_1,...,v_\ell)\in V^*$ denote a sequence of selections for $A$ to go from $C$ to $D$. We construct the sequence $\sigma\in V^*$ by inserting a selection of $v_\mathit{new}$ after every selection of $v_\mathit{old}$, and define $D'$ as the configuration which $A$ reaches after executing $\sigma$ from $C'$. We claim that $D'$ is the same as $D$, apart from having an additional leaf in the same state as $v_\mathit{old}$.

This follows from a simple induction: $v_\mathit{old}$ and $v_\mathit{new}$ start in the same state and see only the root node. As they are always selected subsequently, they will remain in the same state as each other. For the centre we use the property that $A$ cannot count: it cannot differentiate between seeing just $v_\mathit{old}$, or seeing an additional node in the same state. We remark that $G$ being a star is crucial for this argument, which does not extend to e.g.\ cliques.

To summarise, we have shown that for every rejected star $G$ and state $q$ with $L_G(q)\ge (m-1)|Q|+2$ (note the centre), the input $H$ obtained by adding a node with label $q$ to $G$ is still rejected. An analogous argument shows that the same holds for acceptance, and by induction we find that $K:=m(|Q|-1)+2$ is a valid cutoff.
\end{proof}

Since the majority property does not admit a cutoff, in particular we obtain:

\begin{corollary}
No \DAsf{}- or \dAsF-automaton can decide majority.  
\end{corollary}

\section{Extensions} \label{sec:extensions}

We introduce automata with more powerful communication mechanisms, and show that they can be simulated by standard automata with only neighbourhood transitions. We first present our notion of simulation (Definitions \ref{def:sim1}-\ref{def:sim3}), and then in Sections \ref{subsec:wb}-\ref{subsec:rendez-vous} extend automata with weak versions of broadcast (a node sends a message to all other nodes) and absence detection (a node checks globally if there exists a node occupying a given state), and with communication by rendezvous transitions (two neighbours change state simultaneously). 

\newcommand{\Runs}{\mathrm{Run}}
\newcommand{\Last}{\operatorname{\mathsf{last}}}

\begin{definition}
\label{def:sim1}
Let \(G=(V,E,\lambda)\) be a labelled graph and let $Q,Q'$ denote sets of states, with $Q\subseteq Q'$. For configurations $C_1,C_2:V\rightarrow Q'$ we define the relation $\sim_Q$ as $C_1\sim_QC_2$ iff $C_1(v)=C_2(v)$ for all $v$ with $C_1(v)\in Q$ and $C_2(v)\in Q$. Let $\pi,\pi'$ denote runs over states $Q$ and $Q'$, respectively. We say that $\pi'$ is an \emph{extension} of $\pi$ if there exists a monotonically increasing $g:\N\rightarrow\N$ with $\pi(i)=\pi'(g(i))$ for all $i\in\N$, and $\pi'(j)\sim_Q\pi'(g(i))$ or $\pi'(j)\sim_Q\pi'(g(i+1))$ for all $g(i)\le j\le g(i+1)$.
\end{definition}

To implement complicated transitions in an automaton without extensions, we decompose them into multiple standard neighbourhood transitions. Instead of performing, say, a broadcast atomically in one step, agents perform a sequence of neighbourhood transitions, moving into intermediate states in the process. 
As mentioned in Section~\ref{sec:pre}, the results of \cite{ER20} allow us to use liberal or exclusive selection without changing the decision power. We assume that selection is exclusive, unless stated otherwise.

\begin{definition}
\label{def:sim2}
Let \(G=(V,E,\lambda)\) be a labelled graph. Let $\pi,\pi'$ denote runs of an automaton induced by schedules $s,s'\in V^\omega$, respectively. Let $I,I'$ denote the set of indices where $\pi$ or $\pi'$, respectively, execute non-silent transitions, i.e.\ $I:=\{i \colon \pi_i\ne\pi_{i+1}\}$. We say that $\pi'$ is a \emph{reordering} of $\pi$ if there exists a bijection $f \colon I\rightarrow I'$ s.t.\ $s(i)=s'(f(i))$ for all $i\in\N$, and $f(i)<f(j)$ for all $i<j$ where the nodes $s(i)$ and $s(j)$ are adjacent or identical. If that is the case, we also write $\pi_f:=\pi'$ for the reordering induced by $f$.
\end{definition}

While an extension of a run can execute a single complicated transition in many steps instead of atomically, steps of different transitions, or of different phases of a transition, should not ``interfere''. Ideally all the neighbourhood transitions simulating, say, a broadcast, should be executed before any of the transitions simulating the next one. However, in distributed automata this cannot be guaranteed. This is where we make use of reorderings: We will guarantee that every run can be reordered into an equivalent run in which transitions do not ``interfere''. We will only allow reordering of nodes that are not adjacent, thus ensuring that the reordered run yields the same answer as the original one. 

Lastly, we now introduce a generic model encompassing all of our extended automata, which allows us to define our notion  of simulation for all extensions simultaneously.

\begin{definition}
\label{def:sim3}
We say that $P=(Q,\Runs,\delta_0,Y,N)$ is a \emph{generalised graph protocol}, where $Q$ are states, $\delta_0,Y,N$ are initialisation function, accepting states and rejecting states, respectively, and $\Runs$ is a function mapping every labelled graph \(G=(V,E,\lambda)\) over a given alphabet \(\Alphabet\) to a subset \(\Runs(G) \subseteq(Q^V)^\omega\) of fair runs. We define accepting/rejecting runs and the statement “$P$ decides a predicate $\varphi$” analogously to distributed automata. Further, let $\Automaton$ be an automaton with states $Q'\supseteq Q$. We say that $\Automaton$ \emph{simulates} $P$, if for every fair run $\pi'$ of $\Automaton$ there is a reordering $\pi'_f$ of $\pi'$ and a fair run $\pi\in\Runs$ of $P$, s.t.\ $\pi'_f$ is an extension of $\pi$. If $\Automaton$ simulates $P$, we refer to the states in $Q'\setminus Q$ as \emph{intermediate} states.
\end{definition}

We will apply this general definition to simulate broadcast, absence\nobreakdash-detection, and rendezvous transitions by automata with only neighbourhood transitions. In the appendix we show that if $\pi'$ is a reordering of $\pi$ and $v$ is the node satisfying $s(i)=v=s'(f(i))$, then the neighbourhood of $v$ in $\pi$ at time $i$ and the neighbourhood of $v$ in $\pi'$ at time $f(i)$ coincide. 
Furthermore, we show that an automaton $\Automaton$ that simulates $P$ can be easily transformed into an automaton $\Automaton'$ that also simulates $P$ \emph{and} is equivalent to $P$, i.e., decides the same property as $P$.  

\begin{restatable}{lemma}{restateSimulateEquivalent}\label{lem:simulate-equivalent}
Let $P=(Q,\Runs,\delta_0,Y,N)$ denote a generalised graph protocol deciding a predicate $\varphi$, and $\Automaton$ an automaton simulating $P$. Then there is an automaton $\Automaton'$ simulating $P$ which also decides $\varphi$.
\end{restatable}

The automaton \(P''\) constructed in the proof of this lemma is basically \(P'\), except that nodes remember the last state \(q \in Q\) they visited, in addition to their current state \(q' \in Q'\). This allows us to define the accepting/rejecting states of $Q''$ as the pairs $(q', q) \in Q' \times Q$ such that $q$ is an accepting/rejecting state of \(P\). 

\subsubsection*{Notation.} Because of Lemma \ref{lem:simulate-equivalent}, in simulation proofs we often leave out the  accepting and rejecting states of generalised graph protocols and automata. 

\newcommand{\IS}{\mathcal{I}}
\subsection{Weak Broadcasts}
\label{subsec:wb}

Intuitively, a broadcast transition $q\mapsto r,f$ models that an agent in state $q$, called the \emph{initiating} agent or  \emph{initiator}, sends a signal to the other agents, and moves to state $r$;  the other agents react to the signal by moving to new states, determined by their current state and by $f$, a mapping from states to states. Broadcasts are weak, meaning that multiple broadcasts can occur at the same time. When this happens, all initiators send their signals and move to their new states, and for every other agent the scheduler decides which signal it receives and reacts to. It is only guaranteed that every non-initiator receives exactly one signal, and that this signal has been sent.

\begin{definition}\label{def:weakbroadcasts}
A \emph{distributed machine with weak broadcasts} is defined as a tuple $M=(Q,\delta_0,\delta,Q_B,B,Y,N)$, where $(Q,\delta_0,\delta,Y,N)$ is a distributed machine, $Q_B\subseteq Q$ is a set of \emph{broadcast-initiating} states, and $B \colon Q_B\rightarrow Q\times Q^Q$ describes a set of \emph{weak broadcast transitions}, one for each state of $Q_B$. In particular, $B$ maps a state $q$ to a pair $(q',f)$, where $q'$ is a state and $f \colon Q\rightarrow Q$ is a \emph{response function}. We write broadcast transitions as $q\mapsto r,f$, where $f$ is usually given as a set $\{r\mapsto f(r):r\in Q\}$. (Mappings $r\mapsto r$, and silent transitions $q\mapsto q,\operatorname{id}$, $\operatorname{id}$ being the identity function, may be omitted.) Given a configuration $C$ and a selection $S\subseteq V$ of initiators such that $C(v)\in Q_B$ for every $v \in S$, the machine can move to any configuration $C'$ satisfying the following conditions:
\begin{itemize}
\item If $v \in S$, then $C'(v)=q'$, where $q'$ is the state such that $B(C(v))=(q',f)$. 
\item If $v\notin S$, then $C'(v)=f(C(v))$, where 
$B(C(u))=(q',f)$ for some $u \in S$, i.e., $f$ is the response function of an initiator $u$.
\end{itemize}
A valid selection is a nonempty independent set of nodes of $V$.
The set of valid selections is denoted $\IS$. A \emph{schedule} of $M$ is a sequence $\sigma\in(\{n,b\}\times\IS)^\omega$; intuitively, $\sigma(i) = (n, S)$ means that at time $i$ the scheduler asks the agents of $S$ to perform a neighbourhood transition, and $\sigma(i) = (b, S)$ that it asks the agents of $S$ to initiate weak broadcasts. Given a schedule $\sigma$, we generate a \emph{run} $\pi=(C_0,C_1,...)$ as follows. For each step $i\ge1$ either $\sigma(i)=(n,S)$ for $S\subseteq V$ and we execute a neighbourhood transition for $S':=S\setminus C_i^{-1}(Q_B)$, or $\sigma(i)=(b,S)$ for $S\subseteq V$ and we execute a weak broadcast transition on $S':=S\cap C_i^{-1}(Q_B)$. (In either case, if $S'$ is empty we set $C_{i+1}:=C_i$ instead.)

A schedule $\sigma$ is \emph{adversarial} if there are infinitely many $i$ with $\sigma(i)=(b,S)$ for some $S$, or for all $v\in V$ there are infinitely many $i$ with $\sigma(i)=(n,S)$ and $v\in S$. It is \emph{pseudo-stochastic}, if every finite sequence of selections $w\in(\{n,b\}\times\IS)^*$ appears infinitely often in $\sigma$. Given $xyz\in\{\texttt{d},\texttt{D}\}\times\{\texttt{a},\texttt{A}\}\times\{\texttt{f},\texttt{F}\}$, an \emph{$xyz$-automaton with weak broadcasts} is a tuple $(M,\Sigma)$ defined as for an $xyz$-automaton, except that $M$ is a distributed machine with weak broadcasts. In particular, we extend the definitions of fair runs, consensuses, and acceptance to automata with weak broadcasts.

A \emph{strong broadcast protocol} is a tuple $P=(Q,\delta_0,B,Y,N)$ that is defined analogously to a $\dAF$-automaton with weak broadcasts $(Q,\delta_0,\emptyset,Q,B,Y,N)$, except that the set of valid selections is $\IS:=\{\{v\}:v\in V\}$. In other words,  only one agent can broadcast at a given time. This model corresponds to the broadcast consensus protocols of~\cite{BlondinEJ19}.\footnote{The protocols of \cite{BlondinEJ19} also contain rendez-vous transitions, but they can be removed without affecting expressive power.}
\end{definition}


Additionally, to simplify our proofs we assume that all selections $(n,S)$ satisfy $|S|=1$, i.e.\ at each step the scheduler selects one single agent to execute a neighbourhood transition. Observe that we can assume $|S|=1$ without loss of generality. Indeed, since $S$ is an independent set by definition, it only contains non-adjacent nodes, and so after the agents of $S$ execute a neighbourhood transition, be it simultaneously or sequentially, they reach the same states.

\begin{figure*}
\def\svgwidth{147mm}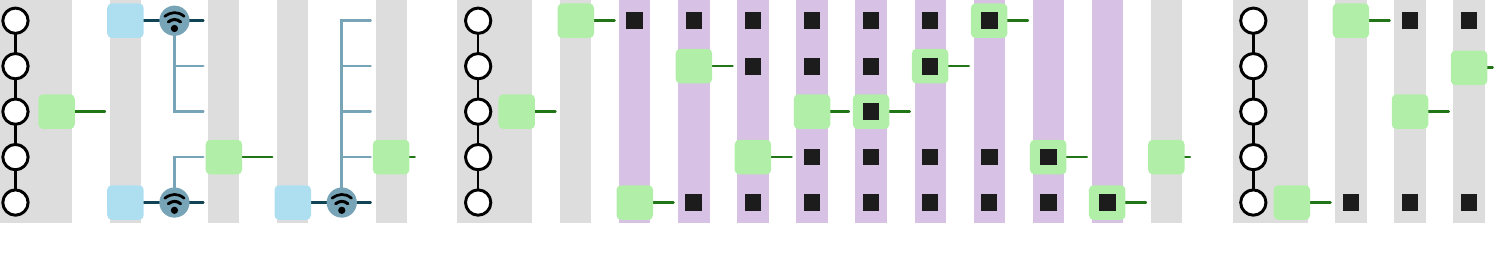
\caption{(a) Prefix of a run of the automaton of Example \ref{ex:dAF}. on a line with exactly five nodes. (b) An extension of the same run, where $\scriptstyle\blacksquare$ denotes intermediate states; only the first $12$ steps are shown. (c) A reordering of the run of (b);  only four steps are shown.}
\label{fig:sampleexecutions}
\end{figure*}
\begin{example}
\label{ex:dAF}
Consider a \dAF-automaton $\ExtendedAutomaton$ with states $\{a, b, x\}$, a neighbourhood transition $x,N\mapsto a$ for every neighbourhood $N \colon Q\rightarrow[1]$ with $N(a)>0$ (i.e.\ an agent moves from $x$ to $a$ if it has at least one neighbour in $a$), and weak-broadcast transitions 
\[a\mapsto a,\{x\mapsto a\} \qquad \text{and} \qquad b\mapsto b,\{b\mapsto a, a\mapsto x\} \ . \]
Figure~\ref{fig:sampleexecutions} shows sample runs of $\ExtendedAutomaton$ on the graph consisting of a line with five nodes. Note that the simultaneous broadcasts at both ends of the line are executed simultaneously, and are received by three and two nodes, respectively. However, the next (and last) broadcast, which is initiated by the bottom node, reaches all nodes.
The reordering depicted in (c) shows the interleaving of two different transitions: while the two ends have already initiated broadcasts, the information has not reached the middle node, and it can execute a neighbourhood transition.
\end{example}

Of course, our model of weak broadcasts would be of limited use if we were not able to simulate it. For this we use a  construction similar to the three-phase protocol of Awerbuch's alpha-synchroniser~\cite{Awerbuch85}.
Instead of simply using it to synchronise, we will propagate additional information, allowing the agents to perform the local update necessary to execute the broadcast.

\newcommand{\Empty}{\square}
\begin{restatable}{lemma}{simulationWeakBroadcast}
Every automaton with weak broadcasts is simulated by some automaton of the same class without weak broadcasts. \label{lem:simulate-weak-broadcast}
\end{restatable}
\begin{proof}[\Proofsketch]

Let $\ExtendedAutomaton=(Q,\delta_0,\delta,Q_B,B)$ denote an automaton with weak broadcasts. We define an automaton $\Automaton=(Q',\delta_0', \delta')$ simulating $\ExtendedAutomaton$. The automaton $\Automaton$ has three phases, called 0,1, and 2. A node moves to the next phase (modulo 3) only if every neighbour is in the same phase or in the next. The states of $\Automaton$ are $Q':=Q \cup Q\times\{1,2\}\times Q^Q$. Intuitively, an agent of $\Automaton$ in state $q \in Q$ is in phase 0, and simulates an agent of $\ExtendedAutomaton$ in state $q$;  an agent of $\Automaton$ in state $(q,i,f) \in Q\times\{1,2\}\times Q^Q$ is in phase $i$, and simulates an agent executing $\ExtendedAutomaton$ in state $q$, and initiating or responding to a broadcast with response function $f$. 

Let $\beta$ denote the counting bound of $\ExtendedAutomaton$. To specify the transitions, for a neighbourhood $N:Q'\rightarrow[\beta]$ we write $N[i]:=\sum_{q,f}N((q,i,f))$ for $i\in\{1,2\}$ and $N[0]:=\sum_{q\in Q}N(q)$ to denote the number of adjacent agents in a particular phase, and choose a function $g(N)\in Q^Q\cup\{\Empty\}$ s.t.\ $g(N)=f\ne\Empty$ implies $N((q,1,f))>0$, and $g(N)=\Empty$ implies $N[1]=0$. The function $g$ is used to select which broadcast to execute, if there are multiple possibilities. We define the following transitions for $\delta'$, for all states $q\in Q$ and neighbourhoods $N:Q\rightarrow[\beta]$.
\setcounter{equation}{0}
\begin{align}
q,N&\mapsto\delta(q,N)&&\text{if $q\notin Q_B$ and $N[0]=|N|$} \\
q,N&\mapsto(q',1,f)&&\text{if $q\in Q_B$ and $N[0]=|N|$, with $(q',f):=B(q)$} \\
q,N&\mapsto(f(q),1,f)&&\text{if $g(N)=f\ne\Empty$} \\
(q,1,f),N&\mapsto(q,2,f)&&\text{if $N[0]=0$} \\
(q,2,f),N&\mapsto q&&\text{if $N[1]=0$}
\end{align}

If all neighbours are in phase $0$, the agent either executes a neighbourhood transition via (1) or it initiates the broadcast in (2), depending on the state of the agent. For the latter, the agent immediately performs the local update. Once there is a phase $1$ neighbour, the agent instead executes the broadcast of one of its neighbours via (3) (if there are multiple, $g$ is used to select one). Note that (2) and (3) are indeed well-defined, as $N[0]=|N|$ holds iff $g(N)=\Empty$. Finally, transitions (4) and (5) move agents to the next phase, once all of their neighbours are in the same or the next phase.
\end{proof}

\subsection{Weak Absence Detection}
\label{subsec:wad}
Absence detection, introduced in \cite{MichailS15}, enables agents to determine the \emph{support} of the current configuration, defined as the set of states currently populated by at least one agent. More precisely, an agent that executes an absence-detection transition moves to a new state that depends on the current support. We consider a weaker mechanism where, as for weak broadcasts, multiple absence-detection transitions may occur at the same time. In this case, each agent executing an absence-detection transition moves according to the support of a \emph{subset} of the agents. However, it is ensured that every agent belongs to at least one of these subsets.

While it is possible to define and implement a more general model involving absence-detection, we limit ourselves to a special case to simplify our proofs. In particular, we define a model in which scheduling is synchronous. Further, we  implement a simulation only for graphs of bounded degree.

\begin{definition}\label{def:weak-absence-detection}
A \emph{distributed machine with weak absence-detection} is defined as a tuple $(Q,\delta_0,\delta,Q_A,\AdTrans,Y,N)$, where $(Q,\delta_0,\delta,Y,N)$ is a distributed machine, $Q_A$ is a set of \emph{initiating} states or \emph{initiators}, and $\AdTrans \colon Q_A\times 2^Q\rightarrow Q$ a set of \emph{(weak) absence-detection transitions}. Given a configuration $C$, a selection  $S\subseteq V$ of initiators such that $C(v) \in Q_A$ for every $v \in S$, and a set $S_v\subseteq V$ for every $v\in S$ satisfying $v\in S_v$ and $\bigcup_{v \in S} S_v=V$, the machine can move to any configuration $C'$ with $C'(v):=\AdTrans(v,C(S_v))$ for $v\in S$ and $C'(v):=C(v)$ for $v\notin S$. (Notice that the $S_v$ need not be pairwise disjoint.) We write $q,S\mapsto q'$ to denote  that $\AdTrans(q,S)=q'$ for $q\in Q_A$, $q'\in Q$ and $S\subseteq Q$.

We use the synchronous scheduler, so the only valid selection is $V$. A step at a configuration $C$ is performed by having each agent execute a neighbourhood transition simultaneously, moving to $C'$, followed by an absence-detection with $S:=C^{-1}(S_A)$ as set of initiators, to go from $C'$ to $C''$. If $S$ is empty, the computation hangs, and we instead set $C'':=C$. A \emph{\DAZf-automaton with (weak) absence-detection} is defined analogously to a \DAZf-automaton.
\end{definition}

As for broadcasts, absence detection is implemented using a three phase protocol. To allow the information to propagate back, we use a distance labelling that effectively embeds a rooted tree for each initiating agent.

\begin{restatable}{lemma}{simulationWeakAbsence}\label{lem:simulate-weakabsence}
Every \DAZf-automaton with weak absence detection is simulated by some \DAf-automaton, when restricted to bounded-degree graphs.
\end{restatable}

\subsection{Rendez-vous transitions} \label{subsec:rendez-vous}
In rendez-vous transitions two neighbours interact and move to new states according to a joint transition function. They are the communication mechanism of population protocols~\cite{AADFP06}. In fact, population protocols on graphs have also been studied previously \cite{angluin2005stably}, and we use exactly the same model. 

A rendez-vous transition $p,q\mapsto p',q'$ allows two neighbouring nodes $u$ and $v$ in states $p$ and $q$ to interact and change their states to $p'$ and $q'$, respectively. Like neighbourhood transitions, rendez-vous transitions are local, i.e., they only involve adjacent nodes. They are useful to model transactions such as transferring a token from one node to another. A population protocol on graphs, or graph population protocol, is a pair $(Q,\delta)$ where $q$ is a set of states and $\delta \colon Q^2 \rightarrow Q^2$ is  a set of rendez-vous transitions, and $p,q\mapsto p',q'$ denotes $\delta(p,q)=(p',q')$. The formal definition can be found in the appendix.

\begin{restatable}{lemma}{simulationRendezVous}
Every graph population protocol is simulated by some \DAF-automaton. \label{lem:simulate-rendezvous}
\end{restatable}

\section{Unrestricted Communication Graphs} \label{sec:unrestricted}

\newcommand{\Any}{\,\cdot\,}
\newcommand{\TraNs}{tra}
\newcommand{\TraName}[1]{\tag*{⟨\textsf{#1}⟩}\label{\TraNs:#1}}
\newcommand{\TraRef}[1]{\text{\ref{\TraNs:#1}}}
\newcommand{\Token}[1]{#1_\mathrm{token}}
\newcommand{\Step}[1]{#1_\mathrm{step}}
\newcommand{\Resad}[1]{#1_\mathrm{reset}}
\newcommand{\Enable}[1]{#1_\mathrm{enable}}

We prove the characterisation of the decision power of the different classes as presented in the introduction. The classes are defined as follows. For a labelling property $\varphi \colon \N^\Lambda\rightarrow\{0,1\}$ we have
\begin{itemize}
\item $\varphi\in\clsTrivial$ iff $\varphi$ is either always true or always false,
\item $\varphi\in\clsCutoffOne$ iff \( \varphi(L)=\varphi(\Cutoff{L}{1})\) for all multisets \(L\in\N^\Lambda\),
\item $\varphi\in\clsCutoff$ iff there exists a \(K\in \N\) s.t.\ \( \varphi(L)=\varphi(\Cutoff{L}{K})\) for all \(L\in\N^\Lambda\), and
\item \(\varphi\in\NL\) iff $\varphi$ is decidable by a non-deterministic log-space Turing machine.
\end{itemize}

The proof proceeds in the following steps:

\begin{enumerate}
\item \DasF\ and therefore all automata-classes with weak acceptance have an upper bound of $\clsTrivial$ and thus decide exactly $\clsTrivial$. This proof also works when restricted to degree-bounded graphs.
\item \DAsf\ and therefore also \dAsf\ can decide at most $\clsCutoffOne$.
\item \dAsf\ and therefore also \DAsf\ can decide at least $\clsCutoffOne$.
\item \dAsF\ can decide exactly $\clsCutoff$.
\item \DAsF\ can decide exactly the labelling propertis in $\NL$.
\end{enumerate}

In this section we sketch the hardest proof, the characterisation for \DAF. All other proofs can be found in the appendix. We start with some conventions and notations.

\subsubsection*{Conventions and notations.}\label{sssec:notation}
When describing automata of a given class (possibly with weak broadcasts or weak absence detection) we specify only the machine; the scheduler is given implicitly by the fairness condition and selection criteria of the class. Further, when the initialisation function and the accepting/rejecting states are straightforward, which is usually the case, we only describe the sets of states and transitions. So, for example, we speak of the automaton $(Q, \delta)$, or the automaton with weak broadcasts $(Q, \delta, Q_B, B)$. We even write $(Q,\delta)+B$; in this case $Q_B$ is implicitly given as the states of $B$ initiating non-silent broadcasts, i.e.\ $Q_B:=\{q:B(q)\ne(q,\operatorname{id})\}$.

Given an automaton $\ExtendedAutomaton$ (possibly with weak broadcasts or absence detection) with set of states $Q$ and a set $Q'$, we let $\ExtendedAutomaton \times Q'$ denote the automaton with set of states $Q \times Q'$ whose transitions leave the $Q'$ component of the state untouched. In other words, if a transition makes a node move from state $(q,q')$ to state $(p, p')$, then $q' = p'$. The definition of the transitions is straightforward, and we omit it. 

We often combine the two notations above. Given an automaton $\Automaton$, we write for example $\ExtendedAutomaton = \Automaton \times Q' + B$ to denote the automaton with weak broadcasts obtained by first constructing $\Automaton \times Q'$, and then adding the set $B$ of weak broadcast transitions.

\begin{restatable}{lemma}{restateDAFequalsNL}
\DAF-automata decide exactly the labelling properties in $\NL$. \label{lem:DAsF-compute-NL}
\end{restatable}
\begin{proof}[\Proofsketch]
First, we argue why \DAF-automata can decide only labelling properties in $\NL$. Let \(\ExtendedAutomaton\) be a \DAF-automaton deciding a labelling property $\varphi$. We exhibit a log-space Turing machine that given a labelled graph \(G=(V, E, \lambda)\) decides whether \(\ExtendedAutomaton\) accepts $G$. Since $\varphi$ is a labelling property, $\varphi(G)=\varphi(\hat{G})$ for the unique clique $\hat{G}$ with set of nodes $V$ and labelling $\lambda$. The Turing machine therefore ignores \(G\) and simulates $M$ on $\hat{G}$. A configuration of $\hat{G}$ is completely characterized up to isomorphism by the \emph{number} of agents in each state; in particular, it can be stored using logarithmic space. In \cite[Proposition~4]{BlondinEJ19} it is shown that any class of automata whose configurations have this property, and whose step relation is in $\NL$ (i.e., there is a log-space Turing machine that on input $(C, C')$ decides if the automaton can move from $C$ to $C'$), can only decide properties in $\NL$. Since the step relation of \DAF-automata on cliques is certainly in $\NL$, the result follows.

Now we show the other direction. 
It is known that strong broadcast protocols decide exactly the predicates in $\NL$~\cite[Theorem~15]{BlondinEJ19}.
Therefore, it suffices to show that for every strong broadcast protocol there is an equivalent \DAF-automaton.
By Lemma \ref{lem:simulate-weak-broadcast} \DAF-automata can simulate weak broadcasts, and so, loosely speaking, the task is to simulate strong broadcasts with weak ones.

Let $P=(Q,\delta,I,O)$ be a strong broadcast protocol. 
We start with a graph population protocol  $\Token{P}:=(\Token{Q},\Token{\delta})$, with states $\Token{Q}:=\{0,L,L',\bot\}$ and rendez-vous transitions $\Token{\delta}$ given by
\begin{equation}
(L,L)\mapsto(0,\bot),\quad(0,L)\mapsto(L,0),\quad(L,0)\mapsto(L',0) \TraName{token}
\end{equation}
\noindent Now we construct a \DAF-automaton $\Token{\Automaton}=(\Token{Q}',\Token{\delta}')$ simulating $\Token{P}$ using Lemma~\ref{lem:simulate-rendezvous}, and combine it with $P$ by setting $\Step{\ExtendedAutomaton}:=\Token{\Automaton}\times Q+\TraRef{step}$ 
, where \TraRef{step} is a weak broadcast defined as
\begin{equation}
(L',q)\mapsto(L,q'),\{(t,r)\mapsto(t,f(r)):(t,r)\in \Token{Q}'\times Q\} \TraName{step}
\end{equation}
for each broadcast $q\mapsto q',f$ in $\delta$. Finally, let $\Step{\Automaton}=(\Step{Q}',\Step{\delta}')$ be a \DAsF-automaton simulating $\Step{\ExtendedAutomaton}$, which exists by Lemma~\ref{lem:simulate-weak-broadcast}. 

Intuitively, agents in states $L,L'$ have a \emph{token}.  
If we could ensure that initially there is only one token in $L, L'$, then we would be done. Indeed, in this case at each moment only the agent with the token can move; if in $L$,  it initiates a (simulated) rendez-vous transition, and if in $L'$, a weak broadcast. Since no other agent is executing a weak broadcast at the same time, the weak broadcast is received by all agents, and has the same effect as a strong broadcast. 

We cannot ensure that initially there is only one token, but 
if the computation starts with more than one, then two tokens eventually meet using transition \TraRef{token} and an agent moves into the \emph{error state }$\bot$. We design a mechanism to restart the computation after this occurs, now with fewer agents in state $(L,\Any)$, guaranteeing that eventually the computation is restarted with only one token. For this we again add an additional component to each state and consider the protocol $\Resad{\ExtendedAutomaton}:=\Step{\Automaton}\times Q+\TraRef{reset}$, where \TraRef{reset} are the following broadcast transitions, for each $q,q_0\in Q$.
\begin{equation}
((\bot,q),q_0)\mapsto((L,q_0),q_0),\{(r,r_0)\mapsto((0,r_0),r_0):r\in\Step{Q}',r_0\in Q\} \TraName{reset}
\end{equation}
For $\Resad{\ExtendedAutomaton}$ we define the input mapping $\Resad{I}(x):=((L,I(x)),I(x))$ and the set of accepting states $\Resad{O}:=\{((r,q),q_0):q\in O,q_0\in Q,r\in\{0,L\}\}$. Using Lemma~\ref{lem:simulate-weak-broadcast} (and Lemma~\ref{lem:simulate-equivalent}) we get a \DAF-Automaton equivalent to $\Resad{\ExtendedAutomaton}$, so it suffices to show that $\Resad{\ExtendedAutomaton}$ is equivalent to $P$.


In the appendix, we show that a run of $\Resad{\ExtendedAutomaton}$ starting with more than one token will eventually reset and restart the computation with strictly fewer tokens, until only one token is left. After this moment, \TraRef{reset} is never executed again, and so we are left with a run of $\Step{\ExtendedAutomaton}$, which stabilises to a correct consensus.
\end{proof}

\section{Bounded-degree Communication Graphs} \label{sec:degree-bounded}

We characterise the decision power of the models when the degree of the input graphs is at most \(k\) for some constant \( k \in \N\). Many results for the unrestricted set of graphs continue to hold, in particular Corollary~\ref{cor:closed-under-scalar-multiplication}, proving that \DAf-automata can only compute properties invariant under scalar multiplication (called $\clsInvSM$ in Figure~\ref{fig:classes}), as well as the result that automata with halting acceptance can only decide trivial properties. The new results are:

\begin{enumerate}
\item For every $k \geq 3$ the expressive power of \dAsf{} is precisely $\clsCutoffOne$.
\item \DAsF- and \dAsF-automata decide exactly the labelling properties in $\NLinSpace$.
\item \DAsf{} can decide all homogeneous threshold predicates, in particular majority. This is a proper subset of $\clsInvSM$ (the latter contains e.g.\ the divisibility predicate $\varphi(x,y)\Leftrightarrow x\vert y$), so there is a gap between our upper and lower bounds for \DAf.

\end{enumerate}

We describe the \DAf-automata for homogeneous threshold predicates. All other proofs can be found in the appendix.

\renewcommand{\TraNs}{tra2}
\NewCommand{\BoundForMaj}{E}
\newcommand{\Cancel}[1]{#1_\mathrm{cancel}}
\newcommand{\No}{\square}
\newcommand{\Detect}[1]{#1_\mathrm{detect}}
\newcommand{\Bc}[1]{#1_\mathrm{bc}}
\newcommand{\Riset}[1]{#1_\mathrm{reset}}

\subsection{\DAf{} decides all homogeneous threshold predicates on bounded-degree graphs}
Let $\varphi:\N^l\rightarrow\{0,1\}$, $\varphi(x_1,...,x_l)\Leftrightarrow a_1x_1+...+a_lx_l\ge 0$ denote an arbitrary homogeneous threshold predicate, with $a_1,...,a_l\in\Z$, and let $k$ denote the maximum degree of the communication graph.

\parag{Local Cancellation}
We first define a protocol that performs local updates. Each agent stores a (possibly negative) integer contribution. If the absolute value of the contribution is large, then the agent will try to distribute the value among its neighbours. In particular, if a node $v$ has contribution $x$ with $x>k$, then it will “send” one unit to each of its neighbours with contribution $y\le k$. Those neighbours increment their contribution by $1$, while $v$ decrements its contribution accordingly. (This happens analogously for $x<-k$, where $-1$ units are sent.) Agents may receive multiple updates in a single step, or may simultaneously send and receive updates.

We define a \DAZf-automaton with weak absence detection $\Cancel{\ExtendedAutomaton}$ $:= (\Cancel{Q},\Cancel{\delta},\emptyset,\emptyset)$, but use only neighbourhood transitions for the moment. We use states $\Cancel{Q}:=\{-\BoundForMaj,...,\BoundForMaj\}$. Here $\BoundForMaj:=\max\{|a_1|,...,|a_l|\}\cup\{2k\}$ is the maximum contribution an agent must be able to store: any agent with contribution $x$ s.t.\ $|x|\le k$ may receive an increment or decrement from up to $k$ neighbours, so $\BoundForMaj \ge k+k$. The transitions $\Cancel{\delta}$ are
\begin{equation}
\begin{aligned}
&x,N\mapsto x-N[-\BoundForMaj,-k{-}1]+N[k{+}1,\BoundForMaj] &\quad&\text{for $x=-k,...,k$}\\
&x,N\mapsto x-N[-\BoundForMaj,k] &&\text{for $x=k+1,...,\BoundForMaj$}\\
&x,N\mapsto x+N[-k,\BoundForMaj] &&\text{for $x=-\BoundForMaj,...,-k-1$} 
\end{aligned} \TraName{cancel}
\end{equation}
Here we write $N[a,b]:=\sum_{i=a}^bN(i)$ for the total number of adjacent agents with contribution in the interval $[a:b]$. As we use the synchronous scheduler, at each step all agents make a move. It is thus easy to see that \TraRef{cancel} preserves the sum of all contributions $\sum_vC(v)$ for a configuration $C$, and that it does not increase $\sum_v|C(v)|$.

We can now show that the above protocol converges in the following sense:
\begin{restatable}{lemma}{restatelocalcancellingtwo}\label{lem:localcancelling2}
Let $\pi=(C_0,C_1,...)$ denote a run of $\Cancel{\ExtendedAutomaton}$ with $\sum_vC_0(v)<0$. Then there exists $i \geq 0$ such that either all configurations $C_i,C_{i+1},...$ only have states in $\{-\BoundForMaj,...,-1\}$, or they only have states $\{-k,...,k\}$.
\end{restatable} 

\parag{Convergence and Failure Detection}
The overall protocol waits until $\Cancel{\ExtendedAutomaton}$ converges, i.e.\ either all agents have “small” contributions, or all contributions are negative. In the latter case, we can safely reject the input, as the total sum of contributions is negative. In the former case we perform a broadcast, doubling all contributions. As we only double once all contributions are small, each agent can always store the new value. This idea of alternating cancelling and doubling phases has been used extensively in the population protocol literature~\cite{AngluinAE08a,BerenbrinkEFKKR18,BilkeCER17,KosowskiU18}.

To detect whether $\Cancel{\ExtendedAutomaton}$ has already converged, and to perform the doubling, we elect a subset of agents as leaders. A ``true'' leader election, with only one leader at the end, is impossible due to weak fairness, but we can elect a “good enough” set of leaders: whenever two leaders disagree, we can eliminate one of them and restart the computation with a non-empty, proper subset of the original set of leaders.

We use weak absence-detection transitions to determine whether $\Cancel{\ExtendedAutomaton}$ has converged. Set $Q_L:=\{0,L,L_\mathrm{double},L_\No\}$ and let $(Q,\delta):=\Cancel{\ExtendedAutomaton}\times Q_L$. (Recall the notation from Section~\ref{sssec:notation}.) We define $\Detect{\ExtendedAutomaton}:=(Q\cup\{\bot,\No\},\delta,\Cancel{Q}\times\{L\},\AdTrans)$, where $\AdTrans$ are the following absence\nobreakdash-detection transitions, for $x\in\Cancel{Q},s\subseteq Q\cup\{\bot,\No\}$.
\begin{equation}
\begin{aligned}
(x,L),s&\mapsto \bot &&\text{if $\No\in s$} \\
(x,L),s&\mapsto (x,0) &&\text{if $\bot\in s$} \\
(x,L),s&\mapsto (x,L_\mathrm{double}) &&\text{if $s\subseteq\{-k,...,k\}\times\{0\}$}\\
(x,L),s&\mapsto (x,L_\No) &&\text{if $s\subseteq\{-\BoundForMaj,...,-1\}\times\{0\}$}
\end{aligned} \TraName{detect}
\end{equation}
Intuitively, $\bot$ and $\Cancel{Q}\times\{L,L_\mathrm{double},L_\No\}$ are leader states, and $\No$ is the (only) rejecting state. State $\bot$ is an error state: an agent in that state will eventually restart the computation. Via Lemma~\ref{lem:simulate-weakabsence} we get a \DAf-automaton $\Detect{\Automaton}=(\Detect{Q}',\Detect{\delta}')$ simulating $\Detect{\ExtendedAutomaton}$.

We want our broadcasts to interrupt any (simulated) absence-detection transitions of $\Detect{\Automaton}$, by moving agents in intermediate states $\Detect{Q}'\setminus\Detect{Q}$ to their last “good” state in $\Detect{Q}$. To this end, we introduce the mapping $\Last:\Detect{Q}'\rightarrow\Detect{Q}$, which fulfils $\Last(C_{i}(v))\in\{\Last(C_{i-1}(v)),C_{i}(v)\}$ for all runs $\pi=C_0C_1...$ of $\Detect{\Automaton}$ and $i>0$, where $C_0$ has only states of $\Detect{Q}$.
It is, of course, not true that $\Last$ exists for \emph{any} simulation $\Detect{\Automaton}$ of $\Detect{\ExtendedAutomaton}$. However, one can extend any simulation which does not, by having each agent “remember” its last state in $\Detect{Q}$.

We construct a \DAsf-automaton with weak broadcasts $\Bc{\ExtendedAutomaton}$ by adding the following transitions to $\Detect{\Automaton}$. 
\begin{equation}
\begin{aligned}
(x,L_\mathrm{double})\mapsto (2x,L),\big(&\{(y,0)\mapsto(2y,0):y\in\{-k+1,...,k-1\}\} \\
&\cup\{q\mapsto\bot:q\in \Cancel{Q}\times\{L,L_\mathrm{double},L_\No\}\}\big)\circ\Last
\end{aligned} \TraName{double}
\end{equation}
\begin{equation}
\begin{aligned}
(x,L_\No)\mapsto\No,\big(&\{(y,0)\mapsto\No:y\in\{-\BoundForMaj,...,-1\}\}  \\
&\cup\{q\mapsto\bot:q\in \Cancel{Q}\times\{L,L_\mathrm{double},L_\No\}\}\big)\circ\Last
\end{aligned} \TraName{reject}
\end{equation}
These transitions are written somewhat unintuitively. Recall that we write a weak broadcast transition as $q\mapsto q',f$, where $q,q'\in\Detect{Q}'$ are states and $f:\Detect{Q}'\rightarrow\Detect{Q}'$ is the transfer function. Usually, we specify $f$ as simply a set of mappings $\{r\mapsto f(r):r\in\Detect{Q}'\}$. Here, our transition essentially is $q\mapsto q',(f\circ\Last)$, where $\circ$ denotes function composition, and $f$ is given as a set of mappings. This means that  broadcasts first move all agents to their last state in $\Detect{Q}$, and then apply the other mappings as specified.

Before extending $\Bc{\ExtendedAutomaton}$ with resets that restart the computation from an error state, we analyse the behaviour of $\Bc{\ExtendedAutomaton}$ in more detail. To talk about accepting/rejecting runs, we define the set of rejecting states as $\{\No\}$. (All other states are accepting.)
Let $\pi:=(C_0,C_1,...)$ denote a fair run of $\Bc{\ExtendedAutomaton}$ starting in a configuration $C_0$ where all agents are in states $\{-\BoundForMaj,...,\BoundForMaj\}\times\{0,L\}$, and at least one agent is in a state $(\Any,L)$. We refer to the agents starting in $(\Any,L)$ as \emph{leaders}. Note that it is not possible to enter a state in $Q\times\{L,L_\mathrm{double},L_\No\}\cup\{\bot\}$ without being a leader. We usually disregard the first component (if any) while referring to states of leaders.

To argue correctness, we state two properties of $\Bc{\ExtendedAutomaton}$. First, it is not possible for \emph{all} leaders to enter $\bot$, which ensures that a reset restarts the computation with a proper subset of the leaders. Second, $\Bc{\ExtendedAutomaton}$ works correctly if no agent enters an error state. Here, $L_G:X\rightarrow\N$ denotes the label count of the input graph, i.e.\ $L_G(x_i)=|C_0^{-1}(a_i)|$ for $i=1,...,l$.
\begin{restatable}{lemma}{restateDAfmajorityerrors}\label{lem:DAfmajorityerrors}
Assuming that no agent enters state $\bot$, $\pi$ is accepting iff $\varphi(L_G)=1$. Additionally, $\pi$ cannot reach a configuration with all leaders in state $\bot$.
\end{restatable}

\parag{Resets}
Finally, we can add resets to the protocol, to restart the computation in case of errors. We use Lemma~\ref{lem:simulate-weak-broadcast} to construct a \DAf-computation  $\Bc{\Automaton}=(\Bc{Q}',\Bc{\delta}')$ simulating $\Bc{\ExtendedAutomaton}$, and then set $\Riset{P}:=\Bc{\Automaton}\times\Cancel{Q}+\TraRef{reset}$, where the broadcasts are defined as follows, for $q_0\in\Cancel{Q}$.
\begin{equation}
(\bot,q_0)\mapsto((q_0,L),q_0), \{(r,r_0)\mapsto((r_0,0),r_0):(r,r_0)\in\Bc{Q}'\times\Cancel{Q}\}\TraName{reset}
\end{equation}

To actually compute $\varphi$, we add the initialisation function $I(x_i):=((a_i,L),a_i)$ and the set of rejecting states $N:=\{\No\}$ to $\Riset{\ExtendedAutomaton}$ (all other states are accepting).

\begin{restatable}{proposition}{restateDAfmajorityiscorrect}\label{lem:DAfmajorityiscorrect}
For every predicate $\varphi:\N^l\rightarrow\{0,1\}$ such that $\varphi(x_1,...,x_l)\Leftrightarrow a_1x_1+...+a_lx_l\ge 0$ with $a_1,...,a_l\in\Z$ there is a bounded-degree \DAsf-automaton computing $\varphi$.
\end{restatable}

\section{Conclusion} \label{sec:conclusion}

We have characterised the decision power of the weak models of computation studied in~\cite{ER20} for properties depending only on the labelling of the graph, not on its structure. For arbitrary networks,the initially twenty-four classes of automata collapse into only four; further, only \DAF{} can decide majority. For bounded-degree networks (a well-motivated restriction in a biological setting, also used in e.g.\ in~\cite{angluin2005stably,BournezL13}), the picture becomes more complex. Counting and non-counting automata become equally powerful, an interesting fact because biological models are often non-counting. Further, the class \DAf{}, which uses adversarial scheduling, substantially increases its power, and becomes able to decide majority. So, while majority algorithms require (pseudo-)random scheduling to work correctly for arbitrary networks, they can work correctly under adversarial scheduling for bounded-degree networks. In particular, there exist a synchronous deterministic algorithm for majority in bounded-degree networks.

\bibliographystyle{plainurl}
\bibliography{references}

\appendix

\section{Proofs of Section~\ref{sec:limitations}}
\label{app:limitations}

\begin{definition}
For every labelled graph \(G=(V,E,\lambda)\) over the finite set of Labels \(\mathcal{L}\) we write \(L_G\) for the multiset of labels occurring in \(G\), i.e.\ \(L_G:\mathcal{L}\rightarrow \N, L_G(x)=|\{v\in V| \lambda(v)=x\}|\) for all labels \(x\). We call \(L_G\) the \emph{label count} of \(G\).

A graph property \(\varphi\) is called a \emph{labelling property} if for all labelled graphs \(G,G'\) with \(L_G=L_{G'}\) we have \(\varphi(G)=\varphi(G')\). In such a case we also write \(\varphi(L_G)\) instead of \(\varphi(G)\).
\end{definition}

\DasFLimitation*

\begin{proof}
Assume there exist cyclic graphs \(G\) and \( H\) such that
$A$ accepts $G$ and rejects $H$. We construct a graph $GH$ and a run of $A$ on \(GH\) such that at least one node of \(GH\) halts in an accepting state, and at least one node of \(GH\) halts in a rejecting state. This contradicts the assumption that $A$ satisfies the consistency condition.

Let \(\rho_G\) and \( \rho_H\) be fair runs of $A$ on \(G\) and \(H\), and let \(g\) and \(h\) be the earliest times at which all nodes of \(G\) and \(H\) have already halted.

Fix edges \(e_G=\{u_G,v_G\}\) and \( e_H=\{u_H,v_H\}\) belonging to cycles of \(G\) and \(H\). 
We construct the graph \(GH\) in three steps. First, we put $2g+1$ copies of $G$ and $2h+1$ copies of $H$ side by side.  Let $G^i, H^i$ denote the $i$-th copy of $G$ and $H$, and let $w_G^i$ and $w_H^i$ denote the copy of a node $w_G$ in $G^i$ or $w_H$ in $H^i$. Second, we remove the edges  $\{u_G^0, v_G^0\}, \ldots, \{u_G^{2g}, v_G^{2g}\},\{u_H^0, v_H^0\}, \ldots, \{u_H^{2h}, v_H^{2h}\}$. Third, we add the edges 
$$\{v_G^0, u_G^1\}, \ldots, \{v_G^{2g-1}, u_G^{2g}\}, \{v_G^{2g}, u_H^{0}\}, \{v_H^0, u_H^1\}, \ldots, \{v_H^{2h-1}, u_H^{2h}\}$$
\noindent The construction is depicted in Figure~\ref{fig:trivial-on-cycles}. Observe that, since \(e_G\) and \( e_H\) belong to cycles of $G$ and $H$, the graph $GH$ is connected.  
\begin{figure}[htb]
	\resizebox{\textwidth}{!}{
\begin{tikzpicture}[thick, auto, on grid]
\tikzstyle{vertex}=[draw,circle,inner sep=0ex,minimum size=4ex]
\tikzstyle{vertexG}=[vertex,fill=black!20]
\tikzstyle{vertexH}=[vertex,fill=white]
\tikzstyle{phantom}=[]
\tikzstyle{graph}=[draw,semithick,cloud,cloud ignores aspect,cloud puff arc=80]
\tikzstyle{graphG}=[graph,inner xsep=-0.7ex,inner ysep=-0.3ex,label={[xshift=-5ex,yshift=-2ex]center:#1},
cloud puffs=10]
\tikzstyle{graphH}=[graph,inner xsep=-0.5ex,inner ysep=+0.5ex,label={[yshift=-\din]center:#1},fill=black!20,
cloud puffs=14]
\tikzstyle{replaced}=[dashed,dash pattern={on 2pt off 1.5pt}]
\tikzstyle{replacedG}=[replaced,black!40!white,line width=1.4pt]
\tikzstyle{replacedH}=[replaced,black!60!white,line width=1.4pt]
\tikzstyle{connection}=[niceblue,line width=1.1pt]
\tikzstyle{changed}=[draw=nicered,fill=nicered!20!white]
\tikzstyle{changed2}=[draw=nicered,fill=nicered!10!white]
\def\dh{10ex}
\def\dv{4ex}
\def\din{7ex}

\node[vertexG,changed] (uG0) {$u^G_0$};
\node[vertexG] (wG0) [right=\din of uG0] {$v^G_0$};
\node[vertexG] (aG0) [below right=0.8*\din and 0.5*\din of uG0] {};
\draw
(uG0) edge[replacedG] (wG0)
(uG0) edge (aG0)
(aG0) edge (wG0)
;

\node[vertexG] (uG1) [right=\dh of wG0] {$u^G_1$};
\node[vertexG] (wG1) [right=\din of uG1] {$v^G_1$};
\node[vertexG] (aG1) [below right=0.8*\din and 0.5*\din of uG1] {};
\draw
(uG1) edge[replacedG] (wG1)
(uG1) edge (aG1)
(aG1) edge (wG1)
;

\node[phantom,connection] (pG)  [above right=\dv and 0.65*\dh of wG1] {$\dots$};
\node[phantom] (pG2)  [below =\dv + 0.5*\din of pG] {$\dots$};
\node[vertexG] (uG2t1) [below right=\dv and 0.65*\dh of pG] {$u^G_{2g}$};
\node[vertexG,changed] (wG2t1) [right=\din of uG2t1] {$v^G_{2g}$};
\node[vertexG] (aG2t1) [below right=0.8*\din and 0.5*\din of uG2t1] {};
\draw
(uG2t1) edge[replacedG] (wG2t1)
(uG2t1) edge (aG2t1)
(aG2t1) edge (wG2t1)
;

\node[vertexH,changed2] (uH0) [right=\dh of wG2t1] {$u^H_0$};
\node[vertexH] (wH0) [right=\din of uH0] {$v^H_0$};
\node[vertexH] (aH0) [below=\din of uH0] {};
\node[vertexH] (bH0) [below=\din of wH0] {};
\draw
(uH0) edge[replacedH] (wH0)
(uH0) edge (aH0)
(aH0) edge (wH0)
(aH0) edge (bH0)
;

\node[vertexH] (uH1) [right=\dh of wH0] {$u^H_1$};
\node[vertexH] (wH1) [right=\din of uH1] {$v^H_1$};
\node[vertexH] (aH1) [below=\din of uH1] {};
\node[vertexH] (bH1) [below=\din of wH1] {};
\draw
(uH1) edge[replacedH] (wH1)
(uH1) edge (aH1)
(aH1) edge (wH1)
(aH1) edge (bH1)
;

\node[phantom,connection] (pH)  [above right=\dv and 0.65*\dh of wH1] {$\dots$};
\node[phantom] (pH2)  [below =\dv + 0.5*\din of pH] {$\dots$};
\node[vertexH] (uH2t2) [below right=\dv and 0.65*\dh of pH] {$u^H_{2h}$};
\node[vertexH,changed2] (wH2t2) [right=\din of uH2t2] {$v^H_{2h}$};
\node[vertexH] (aH2t2) [below=\din of uH2t2] {};
\node[vertexH] (bH2t2) [below=\din of wH2t2] {};
\draw
(uH2t2) edge[replacedH] (wH2t2)
(uH2t2) edge (aH2t2)
(aH2t2) edge (wH2t2)
(aH2t2) edge (bH2t2)
;

\draw
(wG0) edge[connection,bend left=60] (uG1)
(wG1) edge[connection,bend left=30] (pG)
(pG) edge[connection,bend left=30] (uG2t1)
(wG2t1) edge[connection,bend left=60] (uH0)
(wH0) edge[connection,bend left=60] (uH1)
(wH1) edge[connection,bend left=30] (pH)
(pH) edge[connection,bend left=30] (uH2t2)
;

\node[vertexG] (uG) [above=6*\dv of uG2t1] {$u^G$};
\node[vertexG] (wG) [right=\din of uG] {$v^G$};
\node[vertexG] (aG) [below right=0.8*\din and 0.5*\din of uG] {};
\draw
(uG) edge (wG)
(uG) edge (aG)
(aG) edge (wG)
;

\node[vertexH] (uH) [above=6*\dv of uH0] {$u^H$};
\node[vertexH] (wH) [right=\din of uH] {$v^H$};
\node[vertexH] (aH) [below=\din of uH] {};
\node[vertexH] (bH) [below=\din of wH] {};
\draw
(uH) edge (wH)
(uH) edge (aH)
(aH) edge (wH)
(aH) edge (bH)
;

\node (arrow) [above left=2.5*\dv and 0.5*\dh of uH0] {\Huge \color{niceblue}$\downarrow$};

\begin{scope}[on background layer]
\node[graphG=$G_0$,fit=(uG0) (wG0) (aG0)] (G0) {};
\node[graphG=$G_1$,fit=(uG1) (wG1) (aG1)] (G1) {};
\node[graphG=$G_{2g}$,fit=(uG2t1) (wG2t1) (aG2t1)] (G2t1) {};
\node[graphH=$H_0$,fit=(uH0) (wH0) (aH0) (bH0)] (H0) {};
\node[graphH=$H_1$,fit=(uH1) (wH1) (aH1) (bH1)] (H1) {};
\node[graphH=$H_{2h}$,yshift=-0.55ex,fit=(uH2t2) (wH2t2) (aH2t2) (bH2t2)] (H1) {};
\node[graphG=$G$,fit=(uG) (wG) (aG)] (G) {};
\node[graphH=$H$,fit=(uH) (wH) (aH) (bH)] (H) {};
\end{scope}

\end{tikzpicture}
} 

	\vspace{-1cm}
	\caption{Construction in proof of Lemma~\ref{lem:trivial-on-cycles}. The dashed edges are removed and replaced by the blue edges. Only the four red states can initially detect the change.}
	\label{fig:trivial-on-cycles}
\end{figure}

Let \(\rho_{GH}\) be any fair run of $A$ on $GH$ that during the first $\max\{g, h\}$ steps selects exactly the copies 
of the nodes selected at the corresponding steps of \(\rho_G\) and \( \rho_H\).  (Notice that \(\rho_{GH}\) exists, because whether a run is fair or not does not depend on any finite prefix of the run.) Initially, every node of $GH$ except $u_G^0$, $v_G^{2g}$, $u_H^{0}$, and $v_H^{2h}$ ``sees'' the same neighbourhood as its corresponding node in $G$ or $H$ (i.e.\  the same number of neighbours in the same states). Therefore, after the first step of \(\rho_{GH}\) all nodes of $GH$, except possibly these four, are in the same state as their corresponding nodes in $G$ or $H$ after one step of $\rho_G$ or $\rho_h$. 
Since the nodes of $G^g$ are at distance at least $g$  from $u_G^0$, $v_G^{2g}$, $u_H^{0}$, and $v_H^{2h}$,  during the first $g$ steps of \(\rho_{GH}\) any node $w_G^g$ of $G^g$ visits the same sequence of states as the node $w_G$ of $G$ during the first $g$ steps of  \(\rho_G\). Since all nodes of $G$ halt after at most $g$ steps by definition, all nodes of $G^g$ halt in accepting states. Similarly, after $h$ steps all nodes of $H^h$ halt in rejecting states. 
\end{proof}

\DAfCovering*

\begin{proof}
Let \(A\) be a \DAZf-automaton accepting \( \varphi\). Let \( f \colon  V_H \rightarrow V_G\) be a covering map respecting the 
labelling, i.e.\ fulfilling \( \lambda_H=\lambda_G \circ f\). We prove that \(A\) accepts \(G\) if{}f it accepts \(H\). 

Let \( \rho_G=(C_0,C_1,\dots) \) be the synchronous run of \(A\) on \(G\), and let \( \rho_H=(C_0',C_1',\dots)\) be the synchronous run of \(A\) on \(H\). Observe that, since selection is adversarial, the synchronous runs are fair runs. Since $A$ satisfies the consistency condition, it suffices to show that $\rho_G$ accepts 
\(G\) if{}f it accepts \(H\). For this we prove by induction on $t$ that \(C_t(v)=C_t(f(v))\) holds for every node $v$ of $H$ and $t \geq 0$.  For $t=0$ this follows from the fact that \(f\) respects the labelling.
For $t > 0$, assume $C_{t}(v)=C_{t}(f(v))$ we prove $C_{t+1}(v)=C_{t+1}(f(v))$. Pick an arbitrary node $u$. Since $C_{t}(u)=C_{t}(f(u))$, both $u$ and $f(u)$ occupy the same state in $C_t$. Since the run is synchronous, both are selected. 
Since the restriction of the covering $f$ to the neighbourhoods of $u$ and $f(u)$ is a bijection, and $C_{t}(v)=C_{t}(f(v))$ holds for all $v$, in particular for all neighbours of $u$, both $u$ and $f(u)$ move to the same states. 
So $C_{t+1}(u)=C_{t+1}(f(u))$
\end{proof}

\DAfScalarMultiplication*

\begin{proof}
Let \(L\) be a multiset of labels and enumerate it as \( L=(\lambda_1, \lambda_2,\dots, \lambda_{|L|})\). Enumerate \( \lambda \cdot L\) by repeating this sequence \( \lambda\) times. Since \(\varphi\) is a labelling property, the underlying graph does not influence whether the property holds. We consider the following graphs: The cycle \(G\) labelled with \(L\) in the order we established, and the cycle \(G'\) labelled with \(\lambda \cdot L\) in the order above. We have that \(G'\) covers \(G\), and therefore using Lemma~\ref{lem:cannot-distinguish-covering} obtain \( \varphi(L)=\varphi(\lambda \cdot L)\). Since the graphs \(G\) and \(G'\) are 2-degree-bounded, this statement holds also when restricting to k-bounded-degree.
\end{proof}

\DAfCutoff*

\begin{proof}
Let $A$ be a \DAZf-automaton  with counting bound $\beta$ that decides \(\varphi\). 

Since \( \varphi\) is a labelling property, $A$ accepts a graph $G$ if{}f it accepts the unique clique $G'$ (up to isomorphism) such that $L_G = L_{G'}$. Therefore, it suffices to prove $\varphi(G)=\varphi(H)$ for the case in which $G$ and $H$ are \emph{cliques} satisfying \(\Cutoff{L_G}{\beta+1}=\Cutoff{L_H}{\beta+1}\). 

Since $A$ is an automaton with adversarial selection, the synchronous runs \( \rho_G=(C_{G0},C_{G1},\dots) \)  of \(A\) on \(G\), and \( \rho_H=(C_{H0},C_{H1},\dots)\) of \(A\) on \(H\) are fair runs of $A$. Since $A$ satisfies the consistency condition, $A$ accepts $G$ if{}f $\rho_G$ is an accepting run, and similarly for $H$. So it suffices to show that $\rho_G$ is an accepting run if{}f $\rho_H$ is. 

Let $Q_{Gt}\colon Q \rightarrow \N$ be the mapping that assigns to each state $q$ of $A$
the number of nodes of $G$ that are in state $q$ at time $t$. Define $Q_{Ht}$ analogously. We claim: 
$\Cutoff{Q_{Gt}}{\beta+1} = \Cutoff{Q_{Ht}}{\beta+1}$  for every $t \geq 0$. 
The proof is by induction on $t$. For the base case $t=0$, let $q$ be a state. By definition, $Q_{G0}(q)$ is the number 
of nodes of $G$ that are initially at state $q$. Let $\Alphabet_q$ be the set of labels that are mapped to $q$
by the initialisation function of $A$. Then we have $Q_{G0}(q) = \sum_{\ell \in \Alphabet_q} L_G(\ell)$.
Since \(\Cutoff{L_G}{\beta+1}=\Cutoff{L_H}{\beta+1}\), we get  \(\Cutoff{Q_{G0}}{\beta+1}=\Cutoff{Q_{H0}}{\beta+1}\).

For the induction step, assume $\Cutoff{Q_{Gt}}{\beta+1} = \Cutoff{Q_{Ht}}{\beta+1}$. We prove $\Cutoff{Q_{G(t+1)}}{\beta} = \Cutoff{Q_{H(t+1)}}{\beta}$. It suffices to show that if two nodes $u$ and $v$ of $G \cup H$ are in the same state at time $t$, then they are also in the same state (possibly a different one) at time $t+1$. For this, observe first that, since $G$ and $H$ are cliques, all nodes of $G$ respectively $H$ are neighbours. So, since $\Cutoff{Q_{Gt}}{\beta+1} = \Cutoff{Q_{Ht}}{\beta+1}$,
the nodes $u$ and $v$ see the same neighbourhood up to $\beta$ at time $t$ (i.e.\  they see the same number of nodes in 
each state up to bound $\beta$; we go from $\beta+1$ to $\beta$ because the neighbourhood of the node does not contain the node itself). In other words, $N_u^{C_{Gt}} = N_v^{C_{Gt}}$ holds. Since the runs $\rho_G$ and $\rho_H$ are synchronous, both $u$ and $v$ are selected at time $t$ to make a move. Since $N_u^{C_{Gt}} = N_v^{C_{Gt}}$, they move to the same state, and the claim is proved.

Assume that $\rho_G$ is an accepting run of $A$. Then there is a time $t$ such that for every $j \geq 0$ all nodes of $G$ are at accepting states in $C_{G(t+j)}$. By the claim, the same holds for $C_{H(t+j)}$. So $\rho_H$ is an accepting run of $A$. The other direction is analogous.
\end{proof}

\dAFCutoff*

\begin{proof}
\newcommand{\upcl}[1]{\lceil #1 \rceil}
\newcommand{\ctr}[1]{#1^{\text{ctr}}}
\newcommand{\sco}[1]{#1^{\text{sc}}}

We start by repeating some notation of the proof sketch. Let $A$ be the \dAF-automaton, and let $Q$ be its set of states. We first gather some properties of $A$ on \emph{star graphs}. A star is a graph in which a node called the \emph{center} is connected to an arbitrary number of nodes called the \emph{leaves}, and no other edges exist. Since we consider graphs up to isomorphism, a configuration of a star graph is completely determined by the state of the center and the number of nodes in each state. So in the rest of the proof we assume that a configuration of a star graph $G$ is a pair $C= (\ctr{C}, \sco{C})$, where $\ctr{C}$ denotes the state of the center of $G$, and $\sco{C}$ is the \emph{state count} of $C$, i.e.\  the mapping that assigns to each $q \in Q$ the number $\sco{C}(q)$ of nodes of $G$ that are in state $q$ at $C$. We denote the \emph{cutoff} of $C$ as $\Cutoff{C}{m}:=(\ctr{C}, \Cutoff{\sco{C}}{m})$.

Given a configuration $C$ of $A$, recall that $C$ is rejecting if all states are rejecting. We say that $C$ is \emph{stably rejecting} if $C$ can only reach configurations which are rejecting. Given an initial configuration $C_0$, it is clear that $A$ must reject if it can reach a stably rejecting configuration $C$ from $A$. Conversely, if it cannot reach such a $C$, then $A$ will not reject $C_0$, as there is a fair run starting at $C_0$ which contains infinitely many configurations which are not rejecting.

The statement missing in the proof sketch is the following: There exists a number \(m\in \N\) such that a star configuration \(C\) is stably rejecting if and only if \(\Cutoff{C}{m}\) is stably rejecting.

To define \(m\), we have to first define some additional concepts. Given two configurations $C, D$ of $A$ on stars $G$ and $H$, we say that  $C \preceq D$ holds if (a) $\ctr{C}= \ctr{D}$, (b) $\sco{C} \geq \sco{D} $, and (c) $\sco{D}(q)=0$ implies $\sco{C}(q)=0$ for every $q \in Q$. It is easy to see that
$\preceq$ is a partial order. Further, if $C \preceq D$ then $C$ is accepting (rejecting) if{}f $D$ is accepting (rejecting). A set $\mathcal{C}$ of configurations is \emph{upward closed} if $C \in \mathcal{C}$ and $D \succeq C$ implies $D \in \mathcal{C}$.

Let $C \rightarrow^* D$ denote that $A$ can reach the configuration $D$ from $C$ in zero or more steps.
Given a set of configurations $\mathcal{C}$, let \(Pre^{\ast}(\mathcal{C})\) be the set of configurations $C$ such that 
$C \rightarrow^* D$ for some $D \in \mathcal{C}$. The following two claims will finally allow us to define \(m\): 

\begin{itemize}
\item[(1)] If $C \rightarrow^* D$ and $C' \succeq C$, there exists $D' \succeq D$ such that $C' \rightarrow^* D'$. \\
Since \(C'\succeq C\), we can obtain \(C'\) from \(C\) by adding leaves in states which already occur. Similar to the last argument given in the proof sketch, we let every one of these extra leaves copy one of the leaves from \(C\) which starts in the same state. Formally, let \(v_{\mathit{new},1},\dots,v_{\mathit{new},n}\) be the extra leaves. For every extra leaf \(v_{\mathit{new},i}\) let \(v_{\mathit{old},i}\) be some leaf starting in the same state, not necessarily distinct for \(i_1\neq i_2\). Now let $\rho=(v_1,...,v_\ell)\in V^*$ denote a sequence of selections for $A$ to go from $C$ to $D$. We construct a sequence $\sigma\in V^*$ by inserting a selection of $v_{\mathit{new},i}$ after every selection of $v_{\mathit{old},i}$, if multiple \(i\) have the same \(v_{\mathit{old},i}\), insert all of them after \(v_{\mathit{old},i}\) in some order. Define $D'$ as the configuration which $A$ reaches after executing $\sigma$ from $C'$. We claim that $D'$ is the same as $D$, apart from having additional leaves in the same states as $v_\mathit{old}$. This follows from a simple induction: $v_{\mathit{old},i}$ and $v_{\mathit{new},i}$ start in the same state and see only the root node. As they are always selected without the root being selected in between, they will remain in the same state as each other. For the centre we use the property that $A$ cannot count: it cannot differentiate between seeing just $v_{\mathit{old},i}$, or seeing one (or maybe more) additional nodes in the same state.
\item[(2)] For every upward-closed set of configurations $\mathcal{C}$, the set \(Pre^{\ast}(\mathcal{C})\) has finitely many minimal configurations w.r.t. $\preceq$. We denote this finite set by \(MinPre^{\ast}(\mathcal{C})\). \\
Assume for contradiction that \(Pre^{\ast}(\mathcal{C})\) has infinitely many minimal configurations. We can enumerate its elements to obtain an infinite sequence \(C_1, C_2,\dots \) of configurations of star graphs. By the pigeonhole principle,
there exists an infinite subsequence \(C_{i_1}, C_{i_2},\dots \) such that $\ctr{C_{i_j}} = q$ for some state $q$ and every $j \geq 0$, and $\sco{C_{i_j}}(q)=0$ if{}f  $\sco{C_{i_k}}(q)=0$ for every 
$j , k \geq 1$ and every state $q$. By Dickson's Lemma (for every infinite sequence of vectors \(v_1, v_2, \dots \in \N^k\), there exist two indices \(i<j\) such that \(v_i \leq v_j\) with respect to the pointwise partial order),
there exist $j, k$ such that $C_{i_j} \leq C_{i_k}$. But then $C_{i_j}$ and $C_{i_k}$ satisfy conditions (a)-(c) of the definition of $\preceq$, and so $C_{i_j} \preceq C_{i_k}$, which is a contradiction.
\end{itemize}

Now we can define \(m\). Consider the set $\mathcal{C}$ of non-rejecting configurations of $A$ on all star graphs, with any number of leaves. It is easy to see that $\mathcal{C}$ is upward closed. By the claim the set \(MinPre^{\ast}(\mathcal{C})\), i.e.\  the set of smallest configurations from which it is possible to reach a non-rejecting configuration, is finite. Let $m$ be the number of nodes of the largest star such that some configuration of it belongs to 
\(MinPre^{\ast}(\mathcal{C})\). In other words: for every star with more than $m$ nodes, and for every configuration $C$ of this star that can reach a non-rejecting configuration, i.e.\ is not stably rejecting, there is a configuration $C' \prec C$ of another star that can also reach a non-rejecting configuration, i.e.\ is not stably rejecting. Combining this with the fact that \(MinPre^{\ast}(\mathcal{C})\) is upward-closed, we obtain that for every configuration \(C\), \(C\) is not stably rejecting if and only if \( \Cutoff{C}{m}\) is not stably rejecting. By contraposition, this implies the statement we wanted to prove.
\end{proof}

\section{Proofs of Section~\ref{sec:extensions}} \label{appendix:extensions}
The main goal of this section as a whole is to prove that the different models with weak broadcasts, weak absence detection as well as rendezvous transitions can be simulated. We will proceed as follows.

\begin{enumerate}
\item We start by proving Lemmata~\ref{lem:reordering-neighbourhood} and~\ref{lem:simulate-equivalent}, which are general properties of reorderings.
\item In subsection \ref{appendix:3-phase-automata} we show a general lemma concerning reorderings of three-phase protocols, which are used for both weak broadcasts and weak absence detection. In particular, we prove that there is a reordering where all nodes move in lock step.
\item This will dramatically shorten the proofs that weak broadcasts and weak absence detection can be simulated, which make up the next two subsections.
\item At last, we prove that rendezvous transitions can be simulated by \DAF-automata.
\end{enumerate}

\begin{lemma}
Let \(G=(V,E,\lambda)\) be a labelled graph. Let $\pi=(C_0,C_1,...)$ denote a run on \(G\), $\pi_f=(C_0',C_1',...)$ a reordering of $\pi$, and $v$ a node. For all $t\in\N_0$ where $v$ is selected for a non-silent transition and all nodes $u$ adjacent or identical to $v$ we have $C_t(u)=C_{f(t)}'(u)$. \label{lem:reordering-neighbourhood}
\end{lemma}
\begin{proof}
Write the node sequence \(\pi\) as \( (v_0,v_1,v_2,\dots) \). Write the neighbourhood of a node \(v\) as \(N(v)\). Write the reordered run as \(\pi'=(C_0'=C_0,C_1',C_2',\dots) \). We assume wlog that all silent transitions were removed, unless no non-silent transition is enabled. The proof will proceed by induction on $t$.

For the induction basis \(t=0\) we have to prove that before time \(f(0)\), no neighbour of \(v_0\) has changed state yet in the reordered run. Assume for contradiction that some neighbour has changed state already, i.e.\ we have \(f(i)<f(0)\) for some $i$ with \( v_i \in N(v_0)\). Since \(i>0\) and \( \{v_i,v_0\}\in E\), we would have \( f(i)>f(0)\), contradicting the definition of a reordering.

For the induction step, let \(v \in N(v_t)\). We have to prove \( C_t(v)=C_{f(t)}'(v)\). Consider the latest time \(s<t\) where \(v\) has been selected. We obtain \( C_t(v)=C_{s+1}(v)\). By induction hypothesis, we have \( C_s(N(v))=C_{f(s)}'(N(v))\). This implies \(C_{s+1}(v)=C_{f(s)+1}'(v)\). Since \(v\) is a neighbour of \(v_t\), we have \( f(s)<f(t)\). We claim that \(v\) has not been selected between time \(f(s)\) and \(f(t)\) in the reordered run. Assume for contradiction that \(v\) has been selected at time \( f(s) < t' < f(t) \). Let \(m\) be such that \( f(m)=t'\), which exists because we removed silent transitions. Since \(f(s) < f(m)< f(t) \) and \( \{v_m, v_t\} \in E\), we have \( m < t\). We similarly obtain \( s < m\). Therefore \(s\) would not have been the latest time before \(t\) where \(v\) moved, yielding a contradiction. Therefore we have \( C_{f(s)+1}'(v)=C_{f(t)}'(v)\).
\end{proof}

\restateSimulateEquivalent*
\begin{proof}
For $P''$ we reuse $\delta_0$ as initialisation function. We change $P'$ so that each agent remembers its last non-intermediate state. We define $Y'$ as the set of states where the last non-intermediate state is in $Y$, and define $N'$ analogously. Then we use $Y',N'$ as accepting/rejecting states for $P''$. Any run $\pi$ of $P''$ starting in an initial configuration has a reordering $\pi_f$ which is an extension of a run $\tau$ of $P$. If $P$ accepts, then every node $v$ is only finitely often in a state $Q\setminus Y$ in $\tau$, which then also holds for $\pi_f$ and $\pi$. Thus, $v$ will eventually remain in $Y'$, and $\pi$ accepts as well. Similarly, $P''$ will reject if $P$ does.
\end{proof}

This lemma also explains why we treat silent transitions separately: To simulate weak broadcasts, we use a three-phase protocol and want to prove that we can reorder the run such that all nodes move to phase 1, then all to phase 2 and so on. However, Lemma~\ref{lem:reordering-neighbourhood} shows that if some node $v$ observes a neighbour who is behind a phase and a neighbour who is ahead, then this would have to be reflected at some point in time in the reordered run, making it impossible for all nodes to be at most one phase apart. However, in this case our protocols have $v$ do nothing, so removing silent transitions resolves the issue.

\subsection{Reorderings in three-phase automata}\label{appendix:3-phase-automata}
As we want to make a general statement about three-phase protocols, we start by formally introducing the notion. The idea is that each state belongs to one of three phases, that agents may not move directly to the previous phase, and that an agent does nothing, unless all its neighbours are in the same or the next phase. Further, we require that an agent either moves to the next phase or does nothing, if it has a neighbour in the next phase. The last condition is rather technical, it will later allow us construct a reordering which executes the transitions that do not move agents into the next phase first, before executing the other transitions.

\begin{definition}\label{def:threephase}
Let $P=(Q,\delta_0,\delta,Y,N)$ denote an automaton with counting bound $\beta$. We say that $P$ is a \emph{three-phase automaton} if $Q=Q_0\cup Q_1\cup Q_2$, for some pairwise disjoint $Q_0,Q_1,Q_2$, and for all states $q\in Q_i$ and neighbourhoods $N:Q\rightarrow[\beta]$ we have
\begin{enumerate}
\item $\delta(q,N)=q$ if $N(r)>0$ for some $r\in Q_{i-1}$,
\item $\delta(q,N)\in Q_i\cup Q_{i+1}$, and
\item $\delta(q,N)\in \{q\}\cup Q_{i+1}$ if $N(r)>0$ for some $r\in Q_{i+1}$.
\end{enumerate}
Here, we set $Q_3:=Q_0$ and $Q_{-1}:=Q_2$ for convenience. We refer to states in $Q_i$ as \emph{phase-$i$} states.
\end{definition}

As defined above, it is possible for a three-phase automaton to get “stuck” when some agents move to the next phase, but others cannot make progress. Our protocols will not have this problem, so we define the following semantic constraint.

\begin{definition}\label{def:nonblocking}
A three-phase automaton $P=(Q,\delta_0,\delta,Y,N)$ is \emph{nonblocking} if every reachable configuration $C:V\rightarrow Q$ which has agents from at least two phases will eventually execute a transition where an agent changes phase.
\end{definition}

With these definitions we can now state the main proposition of this section.

\begin{proposition}\label{prop:phases-special-reordering}
Let $P=(Q,\delta)$ denote a nonblocking three-phase automaton and $C_0:V\rightarrow Q_0$ an initial configuration. Then every fair run \(\pi=(C_0,C_1,\dots)\) of \(P\) has a reordering $\pi_f=(C_0',C_1',\dots)$ which fulfils, for each step $i$,
\begin{enumerate}
\item $C_i'(V)\subseteq Q_j\cup Q_{j+1}$ for some $j$, and at step $i$ an agent moves to its next phase, or
\item $C_i'(V)\subseteq Q_j$ for some $j$.
\end{enumerate}
\end{proposition}

\newcommand{\Pc}{\mathsf{pc}}

The proof will take up the remainder of this section. We now fix such a $P=(Q,\delta)$ and \(\pi=(C_0,C_1,\dots)\).

It will be convenient to count the total number of phases changes of a node, so we define the \emph{phase count} $\Pc(v,i)\in\N$ as the smallest function which is non-decreasing w.r.t.\ $i$ and has $C_i(v)\in Q_j$ for all $i,v$ and $j:=(\Pc(i,v)\bmod3)$. Intuitively, this means that we increment $\Pc$ whenever a node moves to the next phase. Observe that $\Pc(v,i+1)-\Pc(v,i)\le1$, as a node can move at most one phase per transition.

\begin{lemma}
For all adjacent nodes \(u,v \in V\) we have \( |\Pc(u,i)-\Pc(v,i)|\leq 1\) for all $i$. \label{lem:weak-broadcast-claim1}
\end{lemma}
\begin{proof}
Assume for contradiction that the statement does not hold and prick appropriate $u,v,i$ where $i$ is minimal and $\Pc(u,i)=\Pc(v,i)-1$. Then at step $i-1$ node $v$ must move to the next phase, but this is prohibited by condition (2) of Definition~\ref{def:threephase}.
\end{proof}

Lemma~\ref{lem:weak-broadcast-claim1} implies that if one node has infinitely many phase changes, then all nodes do. We will now show the stronger statement that if a node has $m$ phase changes, for $m\in\N\cup\{\infty\}$, then all other nodes have as well. Here we use that $P$ is nonblocking, so it is not possible for nodes to become stuck in prior phases.

\begin{lemma}
If $\Pc$ is bounded, i.e.\ $\Pc(v,i)\le M$ for some $M\in\N$ and all $v,i$, then there are $m,i\in\N$ with $\Pc(v,j)=m$ for all nodes $v$ and $j\ge i$.
\label{lem:weak-broadcast-claim2}
\end{lemma}
\begin{proof}
Assume that $\Pc$ is bounded. As $\Pc$ is non-decreasing we can thus find a $i$ s.t.\ no node moves to another phase after step $i$. If $C_i$ has two nodes $u,v$ with different phase counts, i.e.\ $\Pc(u,i)<\Pc(v,i)$, then must also be such $u,v$ which are adjacent. Due to Lemma~\ref{lem:weak-broadcast-claim1}, the phase counts of $u$ and $v$ differ only by $1$, therefore we know that $C_i(u)\in Q_j$ and $C_i(v)\in Q_{j+1}$ for some $j$. As $P$ is nonblocking, eventually an agent will move to its next phase, contradicting our choice of $i$. So at step $i$ all phase counts must be pairwise equal.
\end{proof}

Lemma~\ref{lem:weak-broadcast-claim2} is crucial for the reordering, since in the new run of \(A\) all nodes are supposed to perform the same number of phase changes. Now we can define the reordering \(f\). Let \( (v_0,v_1,v_2,\dots)\in V^\omega\) be the sequence of selections inducing run \(\pi\).

We define a new ordering on natural numbers by 
\[i\le_f j\Leftrightarrow (\Pc(v_i,i),\Pc(v_i,i+1),i)\le_{\mathrm{lex}}(\Pc(v_j,j),\Pc(v_j,j+1),j)\]
where $\le_{\mathrm{lex}}$ denotes the lexicographical ordering. The intuition is that we always execute a transition from a node with the lower phase count. Amongst those, we pick one that will not move the node to its next phase, if possible. Finally, from the remaining choices we pick the one that occurred first in the original run $\pi$. Let $I:=\{i\in\N:C_i\ne C_{i+1}\}$ denote the indices of non-silent steps of $\pi$. The function $f:I\rightarrow\N$ can now be defined as $f(i):=|\{j\in I:j\le_f i\}|-1$. (We will see shortly that $f$ is indeed well-defined.)

\begin{lemma}
$\pi_f$ is a reordering.
\end{lemma}
\begin{proof}
We first check that \(f\) is well-defined. For this not to be the case, we would have to have an $i$ with $j\le_f i$ for infinitely many $j$. In particular, there must be a node $u$ s.t.\ there are infinitely many $j\le_fi$ with $v_j=u$, which implies that $\Pc(u,j)<\Pc(v_i,i+1)$ for infinitely many $j$. Therefore $\Pc$ must be bounded by Lemma~\ref{lem:weak-broadcast-claim1}, but then Lemma~\ref{lem:weak-broadcast-claim2} implies that $\Pc(v,j)<\Pc(v_i,i+1)$ can only occur for finitely many $j$, a contradiction.

The function \(f\) is clearly a bijection. In order to show that it induces a reordering, let $i,j\in I$ with $i<j$ and $v_i,v_j$ being adjacent. (Note that the transitions at steps $i$ and $j$ are not silent, by definition of $I$.) We need to show that $f(i)<f(j)$. Assuming that $f(i)\ge f(j)$ holds, i.e.\ $i\ge_f j$, there are two possible cases.

Case 1: $\Pc(v_i,i)>\Pc(v_j,j)$. As $\Pc(\Any,t)$ is non-decreasing in $t$ and $i<j$, we get $\Pc(v_i,i)>\Pc(v_j,i)$. Further, $v_i$ and $v_j$ are adjacent, so Lemma~\ref{lem:weak-broadcast-claim1} implies that $\Pc(v_i,i)=\Pc(v_j,i)+1$. But then, by condition~(1) of Definition~\ref{def:threephase}, the transition at step $i$ must be silent, contradicting our assumption. 

Case 2: $\Pc(v_i,i)=\Pc(v_j,j)$ and $\Pc(v_i,i+1)>\Pc(v_j,j+1)$. Again, $\Pc(\Any,t)$ is non-decreasing in $t$, so $\Pc(v_i,j)\ge\Pc(v_i,i+1)>\Pc(v_j,j+1)\ge\Pc(v_j,j)$. Using Lemma~\ref{lem:weak-broadcast-claim1} we now get $\Pc(v_i,j)=\Pc(v_j,j)+1$ and thus $\Pc(v_j,j+1)=\Pc(v_j,j)$. So at step $j$ node $v_i$ is one phase ahead of its neighbour $v_j$ and $v_j$ moves neither to the next phase, nor does it perform a silent transition. This contradicts condition (3) of Definition~\ref{def:threephase}.
\end{proof}

Let $\pi':=\pi_f$ denote the reordered execution and define $\Pc'$ analogously to $\Pc$. To complete the proof of Proposition~\ref{prop:phases-special-reordering}, it suffices to show that at each step of $\pi'$ an agent with the smallest phase count will be selected, and amongst those the agents that do not move to the next phase are preferred. Intuitively, this sounds reasonable, as we have defined the reordering $f$ in precisely this manner. We do, however, need to argue briefly that the phase counts $\Pc'$ of the reordered execution correspond directly to the original phase counts $\Pc$.

\begin{lemma}
Let $v$ denote a node and let $i\in I$ with $v_i=v$. Then $\Pc'(v,f(i))=\Pc(v,i)$.
\end{lemma}
\begin{proof}
This follows from a simple induction on $i$ combined with Lemma~\ref{lem:reordering-neighbourhood}.
\end{proof}

This concludes the proof of Proposition~\ref{prop:phases-special-reordering}.

\subsection{Simulating Weak Broadcasts} \label{appendix:weak-broadcast}

\simulationWeakBroadcast*

(We repeat the construction from the proof sketch for clarity.) 
Let $\ExtendedAutomaton=(Q,\delta_0,\delta,Q_B,B)$ denote an automaton with weak broadcasts. We define an automaton $\Automaton=(Q',\delta_0', \delta')$ simulating $\ExtendedAutomaton$. The automaton $\Automaton$ has three phases, called 0,1, and 2. A node moves to the next phase (modulo 3) only if every neighbour is in the same phase or in the next. The states of $\Automaton$ are $Q':=Q \cup Q\times\{1,2\}\times Q^Q$. Intuitively, an agent of $\Automaton$ in state $q \in Q$ is in phase 0, and simulates an agent of $\ExtendedAutomaton$ in state $q$;  an agent of $\Automaton$ in state $(q,i,f) \in Q\times\{1,2\}\times Q^Q$ is in phase $i$, and simulates an agent executing $\ExtendedAutomaton$ in state $q$, and initiating or responding to a broadcast with response function $f$. 

Let $\beta$ denote the counting bound of $\ExtendedAutomaton$. To specify the transitions, for a neighbourhood $N:Q'\rightarrow[\beta]$ we write $N[i]:=\sum_{q,f}N((q,i,f))$ for $i\in\{1,2\}$ and $N[0]:=\sum_{q\in Q}N(q)$ to denote the number of adjacent agents in a particular phase, and choose a function $g(N)\in Q^Q\cup\{\Empty\}$ s.t.\ $g(N)=f\ne\Empty$ implies $N((q,1,f))>0$, and $g(N)=\Empty$ implies $N[1]=0$. The function $g$ is used to select which broadcast to execute, if there are multiple possibilities. We define the following transitions for $\delta'$, for all states $q\in Q$ and neighbourhoods $N:Q\rightarrow[\beta]$.
\setcounter{equation}{0}
\begin{align}
q,N&\mapsto\delta(q,N)&&\text{if $q\notin Q_B$ and $N[0]=|N|$} \\
q,N&\mapsto(q',1,f)&&\text{if $q\in Q_B$ and $N[0]=|N|$, with $(q',f):=B(q)$} \\
q,N&\mapsto(f(q),1,f)&&\text{if $g(N)=f\ne\Empty$} \\
(q,1,f),N&\mapsto(q,2,f)&&\text{if $N[0]=0$} \\
(q,2,f),N&\mapsto q&&\text{if $N[1]=0$}
\end{align}

We will use the definitions and results from the previous section, Appendix~\ref{appendix:3-phase-automata}, where we have constructed a general reordering for three-phase protocols.

\begin{lemma}\label{lem:wb-isthreephase}
The automaton $P'$ is a nonblocking three-phase automaton.
\end{lemma}
\begin{proof}
Using $Q_i$ to denote the phase $i$ states as defined above, it is easy to check that $P'$ is a three-phase automaton. It remains to show that $P'$ is nonblocking, so let $\pi=(C_0,C_1,...)$ denote a fair run of $P'$ and $C_i:V\rightarrow Q'$ a configuration where two agents are in different phases. We define the phase count $\Pc$ as for the proof of Proposition~\ref{prop:phases-special-reordering}, meaning that $\Pc(v,j)$ is the number of phase changes of node $v$ until step $j$.

Let $U:=\{u\in V:\Pc(u,i)=\min_v\Pc(v,i)\}$ denote the set of nodes which have a minimal number of phase changes at step $i$. We know that $C_i$ has nodes of two different phases, so $U$ is a proper subset of $V$ and we can pick adjacent nodes $u,v$ with $u\in U$ and $v\notin U$. We claim that the next step selecting $u$ will move it to its next phase. This can be seen by a simple case distinction: if $u$ is in phase $0$, $1$, or $2$, then transition (3), (4), or (5) will move it to the next phase, respectively. Selecting $u\in U$ ensures that $u$ has no neighbours in the previous phase, which already suffices to enable (4) and (5), while having $u$ adjacent to $v$ ensures that (3) can be executed.
\end{proof}

Again, let $\pi$ denote a fair run of $P'$. Using Proposition~\ref{prop:phases-special-reordering} we find a specific reordering $\pi_f=(C_0,C_1,...)$ of $\pi$. In particular, every configuration $C_i$ either has all agents in the same phase, or it has agents in at most two phases and at step $i$ one agent moves to the its next phase. We are now going to show that $\pi_f$ is an extension of a fair run $\tau$ of $P$, which will complete the proof of Lemma~\ref{lem:simulate-weak-broadcast}.

First, note that in $\pi_f$ there are infinitely many configurations $C_i$ where all agents are in the same phase. Due to transitions (4) and (5) it is not possible to perform a silent transition if all agents are in phase $1$ or $2$, respectively. So the set $I:=\{i\in\N:C_i(V)\subseteq Q\}$ of indices $i$ where $C_i$ has only phase $0$ agents has infinitely many elements. We define the mapping $g:\N\rightarrow\N$ as the unique bijection with $g(\N)=I$ which is strictly increasing, and set $\tau:=(K_0,K_1,...)$ where $K_i:=C_{g(i)}$ for all $i$.

\begin{lemma}\label{lem:bc-tauisextension}
$\pi_f$ is an extension of $\tau$.
\end{lemma}
\begin{proof}
Fix any $i,j\in\N$ with $g(i)<j<g(i+1)$. Due to the definition of $g$ we know that $C_{g(i)+1}$ does not contain only phase $0$ agents, while $C_{g(i)}$ does. So an agent has moved to the next phase at step $g(i)$ in $\pi_f$ and the properties of $\pi_f$ guarantee us that there are $t_1,t_2$ with $g(i)<t_1,t_2<g(i+1)$ s.t.\ $C_{t_1}$ ($C_{t_2}$) has only agents in phase $1$ (phase $2$). Moreover, we know that in $\pi_f$ every step $t$ with $g(i)\le t<t_1$ or $t_2\le t<g(i+1)$ moves an agent to its next phase or is silent. As phases $1$ and $2$ consist of only intermediate states, this implies that $C_{g(i)}\sim_QC_j$ if $j<t_1$, $C_{g(i+1)}\sim_QC_j$ if $j>t_2$, and $C_{g(i)}\sim_QC_j\sim_QC_{g(i+1)}$ if $t_1\le j\le t_2$.
\end{proof}

The next lemma is mostly a matter of looking carefully at the definition of our transitions~(1)-(5).

\begin{lemma}
$\tau$ is a run of $P$.
\end{lemma}
\begin{proof}
Fix any $i\in\N$. If $g(i+1)=g(i)+1$, then step $g(i)$ of $\pi_f$ has simply executed transition~(1), which correctly performs a neighbourhood transitions for an agent not in a broadcast-initiating state. Otherwise, as we argued for Lemma~\ref{lem:bc-tauisextension}, there is a $g(i)<t_1<g(i+1)$ s.t.\ $C_{t_1}$ has only agents in phase $1$. Let $S$ denote the set of agents executing transition (2) between steps $g(i)$ and $t_1$ in $\pi_f$. As agent can only move from phase $0$ to phase $1$ via transitions (2) and (3), and transition (3) is enabled iff a neighbour is already in phase $1$, we find that $S$ is both nonempty and an independent set. Additionally, the definition of (2) ensures that $S$ contains only agents in broadcast-initiating states (i.e.\ $K_i(v)\in Q_B$ for $v\in S$). 

Now we simply note that $K_{i+1}$ is the result of executing a weak broadcast transition on $K_i$ on the selection $S$. Transition~(2) correctly perform the local update, while (3) moves the node according to some response function.
\end{proof}

Finally, we have to show that $\tau$ is fair.

\begin{lemma}
$\tau$ is fair.
\end{lemma}
\begin{proof}
There are two cases, depending on whether $P$ uses adversarial or pseudo-stochastic scheduling.

We start with the former. Here, $\pi$ either contains infinitely many transitions where an agent moves to the next phase, in which case $\tau$ executes infinitely many broadcasts and is fair by definition, or there is some $i$ with $K_i+j=C_{g(i)+j}$ for all $j\in\N$. The transitions in $\pi_f$ after step $i$ do not move to new phases, so they are not affected by the reordering $f$. In particular, as $\pi$ is fair, so is $\tau$.

Now we consider the case of pseudo-stochastic scheduling. It is well-known that a pseudo-stochastic schedule (i.e.\ every finite sequence of selections appears infinitely often) implies that every configuration $C$ which can be reach infinitely often will be reached infinitely often. (This follows from there being only finitely many distinct configurations in a run.) That argument can be strengthened to show that, for any finite sequence $\sigma$ of selections, $\sigma$ will be executed starting from $C$ infinitely often.

Further note that any configuration $C$ which appears infinitely often in $\tau$ and thus in $\pi_f$ can be reached infinitely often in $\pi$. This can be seen by choosing a specific reordering which executes the transitions of $\pi$ faithfully up to a point, and then only executes only the transitions necessary to reach $C$ in $\pi_f$. It is easy to see that this a valid reordering (using that $\pi_f$ is reordering).

Now let $\sigma$ denote any finite sequence of selections of $P$, the automaton with weak broadcasts. We want to show that $\sigma$ is infinitely often in $\tau$. So we pick any configuration $K$ appearing infinitely often in $\tau$, and therefore can be reached infinitely often $\pi$. It would now suffice to show that, in $\tau$, $K$ is followed infinitely often by the selections of $\sigma$. As we know that for any sequence $\sigma'$ of selections of $P'$ we have that $K$ appears infinitely often in $\pi$ followed by the selections of $\sigma'$, it now suffices to show that we can pick an appropriate $\sigma'$ that would lead to $\sigma$ being executed in $\tau$.

To execute a selection $(n,v)$ at configuration $K:V\rightarrow Q$, i.e.\ a neighbourhood transition of node $v\in V$, we either have $K(v)\notin Q_B$ and select $v$, thus executing transition (1), or we do nothing. (To be precise, in the latter case we would have to append a copy of $K$ to $\tau$, and modify $g$ s.t.\ $\tau$ is still an extension of $\pi_f$.) For a selection $(b,S)$ with $S\subseteq V$ an independent set of broadcast initiating nodes, we either have $S'=\emptyset$ with $S':=S\cap K^{-1}(Q_B)$ and again do nothing, or we select all agents in $S'$ to move them to phase $1$ via transition (2), then move all other nodes to phase $1$ via transition (3), and then use transitions (4) and (5) to move all nodes to phase $2$ and then back to phase $0$.
\end{proof}

\subsection{Simulating Weak Absence Detection}\label{appendix:weak-absence}
\simulationWeakAbsence*

\newcommand{\Child}{\operatorname{child}}
\newcommand{\Old}{\operatorname{old}}
\newcommand{\Union}{\operatorname{union}}
\newcommand{\Root}{\mathrm{root}}

The proof will take up the remainder of this section.

As mentioned in the main paper, we combine a three-phase protocol with a distance-labelling, the latter allowing us to propagate information about the states that have been seen back to the agents initiating the absence detection. Before we define the necessary transitions, we briefly characterise the distance labelling we are going to use.

\begin{definition}
Let $k\in\N$. We use $D$ to denote a set of \emph{(distance) labels}, where $D:=\Z_{2k+1}\cup\{\Root\}$. We define increment on $D$ by using the usual arithmetic (modulo $2k+1$) for elements in $\Z_{2k+1}$, and setting $\Root+1:=1\in\Z_{2k+1}$. We refer to $\Root$ as \emph{root label}. For all $d\in D$ we say that $d+1$ is the \emph{child label} of $d$.
\end{definition}

Nodes initiating the absence detection use the root label. Each other node will pick a child label $d$ of one of its neighbours, taking care that no neighbour holds a child label of $d$. At this point we will use the bound on the maximum degree of the graph, which makes it easy to see that this is always possible.

\begin{lemma}\label{lem:ad-youhaveachild}
Let $S\subset D$ with $0<|S|\le k$. Then there is a label $d\in D$ s.t.\ $d+1\notin S$ and there is some label $d'\in S$ with $d'+1=d$.
\end{lemma}
\begin{proof}
Note that the statement can be simplified to there being a $d\in S$ with $d+2\notin S$. Pick any $d\in S$. As $2k+1$ is odd, the sequence $d,d+2,d+4,...,d+2k$ is pairwise distinct. Additionally, it contains $2k+1>|S|$ elements, and thus at least one element not in $S$. Moreover, we can thus find two subsequent elements $d',d'+2$ with $d'\in S$, $d'+2\notin S$ in that sequence.
\end{proof}

We will now formally define our construction. Let $P=(Q,\delta,Q_A,A)$ denote the automaton we want to simulate, $k$ the maximum degree of our graph. We will construct a \DAf-automation $P'=(Q',\delta')$ simulating $P$. As states we use $Q':=Q_0\cup Q_1\cup Q_2$, where $Q_i$ contains the phase $i$ states. In particular, we set $Q_0:=Q$, $Q_1:=Q^2\times D$ and $Q_2:=Q\times2^Q$. For $(q,r,i)\in Q_1$, an agent $v$ carries its phase $0$ state $r$ and a distance label $i\in D$. In phase $2$ an agent stores the set of states that it has seen so far, which will be propagated to its parents.

To define the transitions $\delta'$, we introduce some notation. For any neighbourhood $N:Q'\rightarrow\N$ we write $N(S):=\sum_{q\in S}N(q)$ for $S\subseteq Q'$. We set $\Old(N):=N'$, where $N'(q):=N(q)+N(Q\times\{q\}\times\{0,...,2k\})$ is the number of agents that were in $q\in Q$ in phase $0$. We will use $\Old(N)$ to determine which neighbourhood transition of $P$ to execute.

We also define a unique child label $\Child(N)$ for each $N$ with $N(Q_1)>0$. For this, let $S:=\{d\in D:N(Q^2\times\{d\})>0\}$ denote the set of distance labels appearing in $N$. As we have at most $k$ neighbours, Lemma~\ref{lem:ad-youhaveachild} yields a suitable choice $d=:\Child(N)$. Intuitively, this means that $d$ is the child label of a neighbour, but no neighbour is a child of $d$, which ensures that we never create cycles. 

Finally, we write $\Union(N):=\bigcup\{S':(q',S')\in Q_2,N((q',S'))>0\}$ for the union of all states indicated by phase $2$ neighbours. Our transitions $\delta'$ are now defined as follows, for all $q,N$.
\setcounter{equation}{0}
\begin{align}
q,N&\mapsto(q',q,\Root)&&\text{if $N(Q_2)=0$ and $q':=\delta(q,\Old(N))\in Q_A$}\\
q,N&\mapsto(q',q,\Child(N))&&
\begin{aligned}
\text{if }N(Q_2)=0 \text{ and }q':=\delta(q,\Old(N))\notin Q_A \\
\text{and } N(Q_1)>0
\end{aligned}\\
(q,r,i),N&\mapsto(q,\Union(N)\cup\{q\})&&\text{if $N(Q_0)=0$ and $N(Q^2\times\{i+1\})=0$}\\
(q,S),N&\mapsto A(q,S)&&\text{if $N(Q_1)=0$ and $q\in Q_A$}\\
(q,S),N&\mapsto q&&\text{if $N(Q_1)=0$ and $q\notin Q_A$}
\end{align}
Transitions (1) and (2) move the agents from phase $0$ to phase $1$, executing a neighbourhood transition of $\delta$ in the process (synchronously). The move is initiated by agents in $Q_A$, which pick $\Root$ as distance label, while the others wait for a neighbour to enter phase $1$, at which point they become a child of that neighbour. In phase $1$, each node waits until all children have entered phase $2$ (and thus indicate the set of states they have observed), and then executes (3) to move to phase $2$, indicating the union of all sets of its children. Finally, the absence-detection initiating nodes move to phase $0$ by executing the absence-detection via (4), moving into the appropriate state, while (5) simply moves the other agents to phase $0$ without changing their states.

It is crucial that the distance-labels assigned by transition~(2) never form a cycle; else we would get a deadlock. Our choice of $\Child$ ensures that this is the case.
\begin{lemma}\label{lem:ad-nocycle}
The automaton $P'$ cannot reach a configuration $C$ with a cycle of nodes $(v_0,...,v_l)$, i.e.\ $v_0=v_l$ and $v_i$ is adjacent to $v_{i+1}$ for all $i$, where each $v_i$ has label $(i\bmod2k+1)\in\Z_{2k+1}$.
\end{lemma}
\begin{proof}
It is only possible for a node to receive a label $d\in\Z_{2k+1}$ via transition~(2). (Note that (1) only assigns label $\Root$, which cannot be part of a cycle.) However, transition~(2) will never close such a cycle, due to the definition of $\Child$.
\end{proof}

We will proceed in a similar manner as in the previous section. We want to use Proposition~\ref{prop:phases-special-reordering} to construct our reordering, will show that $P'$ is a nonblocking three-phase automaton. Afterwards, we argue that the given reordering is an extension of a run of $P$.

\begin{lemma} \label{lem:ad-nonblocking}
The automaton $P'$ is a nonblocking three-phase automaton.
\end{lemma}
\begin{proof}
Again, it is easy to check that $P'$ is a three-phase automaton by inspecting transitions (2)-(5). To show that $P$ is nonblocking the proof is similar to the proof of Lemma~\ref{lem:wb-isthreephase}.

Let $\pi=(C_0,C_1,...)$ denote a fair run of $P'$ and $C_i:V\rightarrow Q'$ a configuration where two agents are in different phases. We define the phase count $\Pc$ as for the proof of Proposition~\ref{prop:phases-special-reordering}, meaning that $\Pc(v,j)$ is the number of phase changes of node $v$ until step $j$.

Let $U:=\{u\in V:\Pc(u,i)=\min_v\Pc(v,i)\}$ denote the set of nodes which have a minimal number of phase changes at step $i$. If all nodes in $U$ are in phase $0$ or phase $2$, then selecting any node in $U$ will move it to the next phase via transitions (1), (2) or (4), (5), respectively. Otherwise, we pick any node $u\in U$ and write $d$ for the distance label of $u$. As $u$ is in phase $1$, the only non-silent transition it could perform is (3). If (3) is enabled, then executing it moves $u$ to the next phase, and we are done. If that is not the case, there must be a node $u$ adjacent to $v$, s.t.\ $v$ is also in phase $1$ and has label $d+1$. We now set $u:=v$ and repeat this process. There are only finitely many nodes and the distance labels form no cycles (due to Lemma~\ref{lem:ad-nocycle}), so this must terminate.
\end{proof}

As for weak broadcasts, let $\pi$ denote a fair run of $P'$. Using Proposition~\ref{prop:phases-special-reordering} we find a specific reordering $\pi_f=(C_0,C_1,...)$ of $\pi$. In particular, every configuration $C_i$ either has all agents in the same phase, or it has agents in at most two phases and at step $i$ one agent moves to the its next phase. We are now going to show that $\pi_f$ is an extension of a run $\tau$ of $P$.

Again, note that in $\pi_f$ there are infinitely many configurations $C_i$ where all agents are in the same phase. However, to construct $g$ and the extension $\tau$ we will now have to argue slightly differently. If all agents are in phase $1$ then there must be at least one agent where transition (3) is enabled. This follows from the same argument as used in the proof of Lemma~\ref{lem:ad-nonblocking}. If all agents are in phase $2$, then either transition (4) or (5) will be enabled for each agent.

The set $I:=\{i\in\N:C_i(V)\subseteq Q\}$ of indices $i$ where $C_i$ has only phase $0$ agents thus has infinitely many elements, as before. We then define $I':=I\setminus\{i:C_i=C_{i-1},\;\exists j>i: C_j\neq C_i\}$ by removing all steps which perform a silent transition and are followed by a non-silent transition from $I$.

We define the mapping $g:\N\rightarrow\N$ as the unique bijection with $g(\N)=I'$ which is strictly increasing, and set $\tau:=(K_0,K_1,...)$ where $K_i:=C_{g(i)}$ for all $i$.

\begin{lemma}\label{lem:ad-tauisextension}
$\pi_f$ is an extension of $\tau$.
\end{lemma}
\begin{proof}
The proof is analogous to Lemma~\ref{lem:bc-tauisextension}, as the removal of silent transitions does not affect the notion of extension.
\end{proof}

Finally, we argue that $\tau$ is a run of $P$. There is no need to argue that $\tau$ is fair, as all runs of $P$ are fair.

\begin{lemma}
$\tau$ is a run of $P$.
\end{lemma}
\begin{proof}
Fix any $i\in\N$. If $g(i+1)=g(i)+1$ then must have executed a silent transition at step $g(i)$ of $\pi_f$, as else it is impossible to remain in phase $0$. However, we have removed silent transitions followed by non-silent transition from $I'$ and therefore $g$, so step $i$ is followed by an infinite sequence of silent transitions $C_{g(i)}=C_{g(i)+1}=...$ . We know that $\pi$ and thus $\pi_f$ are fair (w.r.t.\ an adversarial scheduler), so this means that transition~(1) is not enabled for any node. Therefore $K_i(V)\cap Q_A=\emptyset$ and Definition~\ref{def:weak-absence-detection} states that $P$ hangs in this case, so $K_{i+1}=K_i$ should hold, which is what we have.

Otherwise, there is a $g(i)<t_1,t_2<g(i+1)$ s.t.\ $C_{t_1}$ ($C_{t_2}$) has only agents in phase $1$ (phase $2$). First, every node executes either (1) or (2), effectively moving to configuration $C'$ with $(C'(v),\Any,\Any)=C_{t_2}(v)$ for $v\in V$. In particular, transitions~(1) and (2) use the phase $0$ state of each agent, so this executes a synchronous neighbourhood transition (i.e.\ one with selection $V$).

Let $S$ denote the set of agents executing transition (1) between steps $g(i)$ and $t_1$ in $\pi_f$. A brief look at transition~(1) reveals that $S=(C')^{-1}(Q_A)$ contains precisely the agents which are in absence-detection initiating states after executing the synchronous neighbourhood transition. 

Now we simply note that $K_{i+1}$ is the result of executing a weak absence-detection transition on $C'$ using selection $S$. Here, we observe that every node $v\notin S$ must pick a child label of a neighbour $u$ in transition~(2). Agent $u$ will only execute transition~(3) once $v$ is in phase $2$, so the information of $v$ will be propagated to $u$, then to a parent of $u$, and so on, until it reaches an agent in $S$. The agents in $S$ then perform transition~(4) and move according to the weak absence-detection transition, while all other nodes execute~(5) and remain in their original state.
\end{proof}

\subsection{Simulating Rendez-vous Transitions} \label{appendix:rendez-vous}

\begin{definition}
A \textit{graph population protocol} is a tuple $(Q,\delta)$, where $Q$ is a finite set of \textit{states}, and $\delta:Q^2\rightarrow Q^2$ is a set of \textit{transitions} that describes the rendez-vous interactions between two adjacent nodes. 
In particular, if $\delta(p,q) = (p', q')$, then we write $p, q \trans p', q'$. 
Further, let $\delta_1(p,q) = p'$ and $\delta_2(p,q) = q'$ be functions for the first and second component of $\delta$.
The definitions of configurations and runs are equivalent to the ones of distributed machines. 
Selections are ordered pairs of adjacent nodes, i.e.\ the set of possible selections is $\{(u,v):\{u,v\}\in E\}$.
If $(u,v) \in V^2$ is the selection in some configuration $C$, then the successor configuration is $C'$ with $C'(u) := \delta_1(C(u), C(v))$, $C'(v) := \delta_2(C(u), C(v))$ and $C'(x) := C(x)$ for all $x \in V \setminus\{u,v\}$. 
We require the schedules the be pseudo-stochastic, so every finite sequence of selections has to appear infinitely often in a schedule.
\end{definition}

\simulationRendezVous*

\begin{proof}
Let $P = (Q, \delta)$ be a population protocol on graphs.
We define a \DAF-automaton $M = (Q', \delta')$ that simulates $P$. We set the counting bound $\beta:=2$.
Let $Q_\RVwait = Q$, $Q_\RVsearch = Q \times \left\{\RVsearch\right\}$, $Q_\RVanswer = Q \times \left\{\RVanswer\right\}$ and $Q_\RVconfirm = Q \times \left\{\RVconfirm\right\} \times Q$.
We define $Q' := Q_\RVwait \cup Q_\RVsearch \cup Q_\RVanswer  \cup Q_\RVconfirm$.
Intuitively, each node stores its state in the original protocol in the first component and has a status that helps to simulate rendez-vous transitions. 
The status can be ``waiting'' (\RVwait), ``searching'' (\RVsearch), ``answering'' (\RVanswer) or ``confirming'' (\RVconfirm) and initially every node is waiting. 
Additionally, a confirming node stores the state it would have after the rendez-vous transition is completed. 
This is necessary, because the other node will performs its part of the rendez-vous interaction first.

Now will define the transition function $\delta'$ of $M$, for states $q,q',q''\in Q$ and neighbourhood $N \in [\beta]^{Q'}$. 
Let $N(\RVwait) := \sum_{q \in Q} N(q)$ be the number of detectable waiting neighbours. 
Further, we use the auxiliary function $f(N)$ to denote the unique non-waiting neighbour, if any, i.e.\ we set $f(N):=x$ if $N(\RVwait) = |N|-1$ and $N(x)=1$. If no such neighbour exists, we set $f(N):=\RVwait$ if $N[\RVwait] = |N|$ (all neighbours are waiting), and $f(N):=\bot$ otherwise. 
Figure~\ref{fig:rendez-vous} contains the formal definition of $\delta$ and a diagram that visualises how nodes change their status.
\begin{figure}[h]
    \hspace*{-10mm}
	\begin{minipage}{102mm}
		\begin{align*}
		q,N &\mapsto (q,\RVsearch) &&\text{ for } f(N)=\RVwait \\
		q,N &\mapsto (q,\RVanswer) &&\text{ for } f(N)=(q',\RVsearch) \\
		(q,\RVsearch),N &\mapsto (q,\RVconfirm, \delta_1(q,q')) &&\text{ for } f(N)=(q',\RVanswer) \\
		(q,\RVanswer),N &\mapsto \delta_2(q',q) &&\text{ for } f(N)=(q',\RVconfirm,q'') \\
		(q,\RVconfirm,q'),N &\mapsto q' &&\text{ for } f(N)=\RVwait
		\end{align*}
	\end{minipage}%
	\hfill\vrule\hfill
	\begin{minipage}{50mm}
		\begin{tikzpicture}[thick,auto,->,inner sep=1pt]
\tikzstyle{status}=[draw,circle]
\tikzstyle{statechange}=[dashed]
\tikzstyle{tlabel}=[rectangle]
\def\d{10ex}
\def\dd{8ex}

\node[status] (wait) {\RVwait};
\path (wait.west) +(-\d,0) -- +(0,\d) node[status,anchor=east,pos=0.5] (search) {\RVsearch};
\node[status] (answer) [right=\dd of wait] {\RVanswer};
\node[status] (confirm) [above =0.7*\d of wait] {\RVconfirm};

\draw
(wait) edge node[tlabel,pos=0.3] {all \RVwait} (search)
(wait) edge[bend right] node[tlabel,swap,below=3pt,pos=0.7] {1$\times$\RVsearch, rest \RVwait} (answer)
(search) edge node[tlabel,pos=0.8] {1$\times$\RVanswer, rest \RVwait} (confirm)
(answer) edge[statechange, bend right] node[tlabel,swap,above=2pt,pos=0.3] {1$\times$\RVconfirm, rest \RVwait} (wait)
(confirm) edge[statechange] node[tlabel,swap] {all \RVwait} (wait)
;
\end{tikzpicture}\vspace{-1cm}
	\end{minipage}

	\caption{
		Neighbourhood transitions that simulate rendez-vous interactions. 
		The left side show the formal definition of the transition function $\delta'$. 
		For all inputs $x,N$ where $\delta'(x,N)$ is undefined, the status is set to waiting (\RVwait) by changing to the original state saved in the first component of $x$. 
		The right side visualises the neighbourhood transitions as a graph. 
		States in the diagram only show the status of a node. 
		A node only change its status by following an edges in the diagram, if its neighbourhood satisfies the condition on the edge. 
		If a node is selected and no edge can be followed, it instead changes its status to waiting (\RVwait). 
		The edges that apply the rendez-vous transition $\delta$ are drawn dashed.
	}
	\label{fig:rendez-vous}
\end{figure}

Intuitively, the simulation of a rendez-vous transition $p,q \trans p',q'$ starts with a waiting (\RVwait) agent with original state $p$ that only sees waiting nodes. 
This agent searches for a  partner by changing its status to \RVsearch. 
Then, its waiting neighbours can answer by changing to \RVanswer{} if they detect exactly one search. 
If the searching agent detects exactly one answer, it confirms by changing to \RVconfirm{} while remembering the state it would have after the rendez-vous with the answering node.
If the answering node with original state $q$ sees exactly one confirmation, it applies the state change ($q$ to $q'$) and waits.
Then, the confirming node detects that the answering node is now waiting and applies the state change it remembered ($p$ to $p'$). 
However, once a node detects an irregularity in the simulation (e.g. more than one non-waiting neighbour) it cancels the interaction by changing its state to \RVwait.

We still have to show that $M$ simulates $P$.
We call a change to the original state of a node a \textit{state change}. 
In other words, neighbourhood transitions that change the first component of a nodes state perform a state change.
We will now argue, that state changes only occur in pairs and that they simulate the rendez-vous transitions in $\delta$.
First note, that from a configuration $C$ where two nodes $u,v$ and their neighbours are waiting, scheduling the sequence $u,v,u,v,u$ correctly applies the state changes $\delta_1(C(u),C(v))$ for $u$ and $\delta_2(C(u),C(v))$ for $v$.
For a node $u$ to enter the confirming state, it must have exactly one answering neighbour $v$ and all other neighbours must be waiting.
$u$ cannot perform its state change before $v$ because it needs to wait until all nodes are waiting.
Once the answering agent $v$ performs its state change, all of $u$'s neighbours are waiting and cannot change their status because they see that $u$ is confirming.
Thus, the next time $u$ is scheduled, it must perform its state change. 
Further, $v$ can only perform the state change if it sees exactly one confirming state and all other of $v$'s neighbours are waiting.
Thus, once one of nodes in the rendez-vous interaction performs the state change, the full rendez-vous interaction will be performed.
Further, it is impossible for more than two nodes to interact simultaneously, because selection is exclusive and whenever a node detects more than one non-waiting neighbour, it cancels the interaction and waits.

Next, we need to reorder a given run $\pi'$ of $M$ such that it is an extension of some run $\pi$ in $P$. 
For this, we make sure that after an answering node performs the state change, the corresponding confirming node is scheduled immediately so that it can perform its state change. 
Intuitively, the reordering makes sure that the state changes in the simulation of two different rendez-vous transitions are not executed in an interleaving manner. 
Thus, the reordered run is indeed an extension of a run in $P$ where the state changes of rendez-vous interactions happen atomically.
The reordering is valid, because after the answering node performs the state change, the nodes in the neighbourhood of the confirming node are all waiting and they cannot change their status because they see a confirming node. 
Thus, scheduling the confirming agent earlier in the reordered run does not interfere with the neighbourhood transitions that were executed between the two state changes in $\pi'$.

Lastly, we need to argue about fairness. 
Let $\pi$ be the simulated run of $P$ for some fair run $\pi'$ of $M$. 
Let $C$ be some configuration that is visited infinitely often in $\pi$. 
Further, let $S = (u_1,v_1),\cdots,(u_k,v_k) \in (V \times V)^*$ be a finite sequence of selections such that scheduling $S$ in $C$ leads to come configuration $C_f$.
As $C$ is visited infinitely often, there are infinitely many configurations $C'_1,C'_2,\cdots \in \pi'$ with $C \sim_Q C'_i$ for all $i > 0$. 
Because there are only finitely many different configurations for a given graph, there is at least one configuration $C'$ that is visited infinitely often in $\pi'$ such that $C \sim_Q C'$. 
$C'$ can reach a configuration $C'' \sim_Q C'$ where all nodes are waiting by scheduling all confirming nodes, then all answering nodes and lastly all searching nodes. 
$C''$ can reach a configuration $C'_f \sim_q C_f$ by simulating all selections $(u_i,v_i)$ of $S$ one after the other by scheduling $u_i,v_i,u_i,v_i,u_i$.
Because $\pi'$ is fair, $C'_f$ is visited infinitely often in $\pi'$. 
Therefore, $C_f$ is visited infinitely often in $\pi$ and $\pi$ is fair.
\end{proof}

\section{Proofs of Section~\ref{sec:unrestricted}}
As mentioned in the introduction, we are interested in labelling properties. 

\begin{definition}
For every labelled graph \(G=(V,E,\lambda)\) over the finite set of Labels \(\mathcal{L}\) we write \(L_G\) for the multiset of labels occurring in \(G\), i.e.\ \(L_G:\mathcal{L}\rightarrow \N, L_G(x)=|\{v\in V: \lambda(v)=x\}|\) for all labels \(x\). We call \(L_G\) the label count of \(G\).

A graph property \(\varphi\) is called a labelling property if for all labelled graphs \(G,G'\) with \(L_G=L_{G'}\) we have \(\varphi(G)=\varphi(G')\). In such a case we also write \(\varphi(L_G)\) instead of \(\varphi(G)\).
\end{definition}

In this section, we will use \(L\) for multisets of labels.
\subsection{\DaF\ only decides trivial properties}\label{Appendix:Weak-Acceptance}

\begin{proposition}
Let \( \varphi\) be a labelling property decided by a \DasF-automaton in the unrestricted set of graphs or in the set of $k$-degree-bounded graphs. Then \( \varphi\) is trivial, i.e.\ either always false or always true. \label{prop:weak-acceptance-trivial}
\end{proposition}

\begin{proof}
Assume that \( \varphi\) is not always false, i.e.\ \( \varphi(L)=1\) for some \(L\). We have to prove that \( \varphi \) is always true, i.e.\ \( \varphi(L')=1\) for all labelling multisets \(L'\). Let \( L'\) be any labelling multiset. By our general assumption, network graphs have at least 3 nodes, i.e.\ \( |L|, |L'|\geq 3\). Since \(\varphi\) is a labelling property, we can choose the underlying graph. Let \(G\) be the cycle with \( |L| \) nodes labelled with \(L\), and let \(G'\) be the cycle with \(|L'|\) nodes labelled with \(L'\). By Lemma~\ref{lem:trivial-on-cycles}, we cannot distinguish \(G\) and \(G'\) and therefore have \( \varphi(L')=\varphi(L)=1\) as claimed. This also holds in the $k$-degree-bounded case since the graph constructed in the proof of Lemma~\ref{lem:trivial-on-cycles} is $k$-degree-bounded, if both \(G\) and \(G'\) were.
\end{proof}

\subsection{\DAf\ decides at most $\clsCutoffOne$} \label{Appendix:Weak-fairness}

\begin{proposition}
Let \( \varphi\) be a labelling property decided by a \DAsf-automaton. Then \(\varphi \in \clsCutoffOne\).
\end{proposition}

\begin{proof}
Let \(A\) be a \DAf-automaton with \(\varphi_A=\varphi\). Let \(K=\beta+1\) be as in Lemma~\ref{lem:finite-cutoff}, i.e.\ the natural number such that \( \varphi(L)=\varphi(\Cutoff{L}{K})\) for all labelling multisets \(L\). In addition, we know that \(\varphi\) is closed under scalar multiplication by Corollary~\ref{cor:closed-under-scalar-multiplication}. We use this corollary with \(\lambda=K\) to scale up \(L\) with the factor \(K\), then cut if off at \(K\) and scale down again. 

Formally, we start by proving \( \Cutoff{\lambda \cdot L}{\lambda}=\lambda \cdot \Cutoff{L}{1}=\Cutoff{\lambda \cdot \Cutoff{L}{1}}{\lambda}\) for all \( \lambda \in \N\) by case distinction, namely if some label \(x\) occurs, then all those three functions equal \(\lambda\), and otherwise they all equal 0.

We use this to obtain the following chain of equalities:
\newcommand{\StackRefEq}[1]{\stackrel{\text{#1}}{=}}
\begin{align*}
\varphi(L)&\StackRefEq{C\ref{cor:closed-under-scalar-multiplication}}\varphi(K \cdot L)\StackRefEq{L\ref{lem:finite-cutoff}}\varphi(\Cutoff{K \cdot L}{K})=\varphi(K \cdot \Cutoff{L}{1})=\varphi(\Cutoff{K \cdot \Cutoff{L}{1}}{K})\\&\StackRefEq{L\ref{lem:finite-cutoff}}\varphi(K \cdot \Cutoff{L}{1})\StackRefEq{C\ref{cor:closed-under-scalar-multiplication}}\varphi(\Cutoff{L}{1})
\end{align*}
\end{proof}

\subsection{\dAf\ can decide $\clsCutoffOne$} \label{Appendix:1-Cutoff}

\begin{proposition}
\dAsf-automata can decide all labelling properties \( \varphi \in \clsCutoffOne\).
\end{proposition}

\begin{proof}
Let \( \varphi \in \clsCutoffOne\). Let \( x_1,\dots,x_n \) be the variables occurring in \( \varphi \). Then \( \varphi \) corresponds to a subset \(M\) in \( \{0,1\}^{ \{1,\dots, n\}} \), describing whether we accept if exactly the variables with indices \( i \mapsto 1\) occur. By \cite[Proposition~12]{ER20}, \dAsf-automata can decide the language \(B\) of graphs with a black node, i.e.\ the labelling predicate \( \varphi(x,y)\Leftrightarrow x \geq 1. \) On the level of subsets \(M\), this means the set \[M_i:=\{f:\{1,\dots,n\}\rightarrow \{0,1\}:f(i)=1\}\] of all functions with \( i \mapsto 1\) can be decided. We can write every subset \(M\) via unions, complements and intersections of sets \(M_i\). This corresponds to writing \( \varphi \) as a boolean combination of \( x_i \geq 1\), which can be decided.
\end{proof}

\subsection{\dAF\ can decide exactly $\clsCutoff$} \label{Appendix:Cutoff}

\begin{lemma}
For every property $\varphi:\N^l\rightarrow\{0,1\}$ with $\varphi(x,y_1,...,y_{l-1})\Leftrightarrow x\ge k$ for some $k\in\N$ there is a \dAsF-automaton deciding $\varphi$.
\label{lem:dAsF-greater-than-k}
\end{lemma}
\begin{proof}
We construct a \dAsF-automaton $P=(Q,\emptyset,I,O)$, which we will augment with weak broadcast transitions. As states we use $Q:=\{0,1,...,k\}$, the input mapping is given by $I(x):=1$ and $I(y_1)=...=I(y_{l-1})=0$, and the set of accepting states is $O:=\{k\}$. We add the following broadcasts, with $i=1,...,k-1$.
\begin{gather}
i\mapsto i,\{i\mapsto i+1\} \TraName{level} \\
k\mapsto k,\{q\mapsto k:q\in Q\} \TraName{accept}
\end{gather}
Using Lemma~\ref{lem:simulate-weak-broadcast} we get an equivalent \dAsF-automaton.

Let $C_0$ denote an initial configuration with $c:=|C_0^{-1}(1)|$ set to the number of agents starting in state $1$. It is easy to see that $C_0$ is accepting iff it can reach a configuration $C$ with $k\in C(V)$, i.e.\ at least one agent in state $k$. Of course, \TraRef{accept} cannot be used to reach $k$ as the initiator is already in state $k$, so we now consider only configurations $C$ reachable by \TraRef{level}.

It is only possible for an agent to go from state $i$ to $i+1$ by receiving broadcast \TraRef{level} initiated by an agent in state $i$, for $i=1,...,k-1$. The initiator remains in state $i$, so we have that $i+1\in C(V)$ implies $i\in C(V)$. Therefore $k\in C(V)$ implies $\{1,2,...,k\}\subseteq C(V)$ and thus at least $k$ agents have started in state $1$, i.e.\ $c\ge k$, as it is not possible to leave state $0$ via \TraRef{level}.

To summarise, the protocol accepts only initial configuration which should be accepted. It remains to show that the converse holds as well, so we require $c\ge k$ and set $C$ to an arbitrary configuration reachable from $C_0$ with only \TraRef{level}. We have pseudo-stochastic fairness, so it is enough to show that $C$ can reach a configuration $C'$ with $k\in C'(V)$.

Let $m_i:=|C^{-1}(i)|$ denote the number of agents in state $i$, for $i=1,...,k$. We define an ordering on the set of configurations by ordering the tuples $(m_k,m_{k-1},...,m_1)$ lexicographically. If $C$ does not have a node in state $k$, then it has $c\ge k$ occupying states $1$ to $k-1$, i.e.\ $m_1+...+m_{k-1}=k$ and, by pigeonhole principle, there is some $1\le j<k$ with $m_j\ge 2$. By executing transition \TraRef{level} on one of those agents exclusively, at least one agent moves to state $j+1$. Hence the resulting configuration is strictly larger w.r.t.\ our ordering. There are only finitely many configurations with $n$ agents, so we can repeat this procedure until at least one agent has state $k$, thereby proving that $C$ reaches some accepting configuration.
\end{proof}

\begin{proposition}
The set of labelling properties decided by \dAsF-automata is precisely $\clsCutoff$.
\end{proposition}

\begin{proof}
By Lemma~\ref{lem:dAFCutoff}, the expressive power is contained in $\clsCutoff$.

Now let \( \varphi \in \clsCutoff\). Let \(K \in \N\) be as in the definition of $\clsCutoff$. Let \(x_1,\dots, x_n\) be the variables occurring in \( \varphi\). Then \( \varphi \) corresponds to a \(M\in [K]^{ \{1,\dots,n\} } \) of accepted cutoffs. If we can decide all formulas corresponding to 1-element subsets of \( \{0,1,\dots, K\}^{ \{1,\dots,n\} } \), then we can decide \( \varphi\), since \( \varphi \) can be written as a disjunction of such formulas. Let \(M\) be such a 1-element subset, write this element as \( f: \{1,\dots, n\} \rightarrow \{0,1,\dots, K\} \). Let \(S\subseteq \{1,\dots, n\} \) be the set of indices \(i\) with \(f(i)=K\). The formula corresponding to \(M\) is \[ \bigwedge_{i\notin S} (x_i \geq f(i) \wedge \neg (x_i \geq f(i)+1)) \wedge \bigwedge_{i \in S} (x_i \geq f(i)), \] which can be decided since by Lemma~\ref{lem:dAsF-greater-than-k}, we can compute \( x_i \geq f(i)\), and the set of decidable properties is closed under boolean combinations.
\end{proof}

\subsection{\DAF\ can decide exactly the labelling properties in \(\NL\)}
\restateDAFequalsNL*
\begin{proof}
As argued in the main paper, \DAF-automata can decide at most the predicates in $\NL$. We now restate the construction from the proof sketch for clarity.

\renewcommand{\TraNs}{tra3}

As a technical aide to state the proof, we introduce another graph population protocol $\Token{P}^*:=(\Token{Q},\Token{\delta}^*)$, where $\Token{\delta}^*:=\{(L,L)\mapsto(0,\bot),(0,L)\mapsto(L,0)\}$. Essentially, we want to ignore the difference between $L$ and $L'$. For this, we consider the mapping $g:\Token{Q}'\rightarrow\Token{Q}'$, which maps $g(L'):=L$ and all other states to themselves. We extend $g$ to $\Step{Q}$ by applying it only to the first component, i.e.\ $g((q,r)):=g(q)$ for $(q,r)\in\Step{Q}$, and then to configurations $C:V\rightarrow\Step{Q}$ and runs $\pi$ of $\Step{P}$ in the obvious manner.

Let $\pi$ denote a fair run of $\Step{P}$. If we only consider the first component of the states in $\pi$, this is essentially a run of $\Token{P}'$, except that at some steps an agent transitions from $L'$ to $L$ using \TraRef{step}. It is not guaranteed that a simulation continues to work under these conditions, but the construction of Lemma~\ref{lem:simulate-rendezvous} does not rely on the non-intermediate states remaining unchanged between transitions. This means that $g(\pi)$ has a reordering which is an extension of a fair run of $\Token{P}^*$.

Let $C_0^*:V\rightarrow Q$ denote an initial configuration of $P$, and let $\pi:=(C_0,C_1,...)$ denote a run of $\Step{P}$, starting in a configuration $C_0$ with $C_0(v)=(0,C_0^*(v))$ or $C_0(v)=(L,C_0^*(v))$ for all nodes $v$. 

If $C_0$ has $k>1$ tokens, then, we claim, $\pi$ will reach a configuration with an agent in an error state, and the set $S:=\bigcup_i C_i^{-1}(Q\times\bot)$ of agents to ever reach an error state has size at most $k-1$. We will refer to this property as $A(\pi)$. Crucially, if $A(\pi)$ holds for any run $\pi$, then it also holds for any reordering of $\pi$, and it also holds for any extension of $\pi$.

As we argued before, $g(\pi)$ is a reordering of an extension of a fair run of $\Token{P}^*$. Moreover, this projection does not affect $A$. So it is sufficient to show $A(\tau)$ for all fair runs $\tau$ of $\Token{P}^*$, which follows immediately from its definitions.

Let $\pi':=(C_0',C_1',...)$ denote a fair run of $\Resad{P}$ with initial configuration $C_0'$ defined similar to $C_0$, so $C_0'(v)=((0,C_0^*(v)),C_0^*(v))$ or $C_0'(v)=((L,C_0^*(v)),C_0^*(v))$ for all nodes $v$, and the latter holds for exactly $k>1$ nodes. (Note that initial configurations of $\Resad{P}$ always have $k=n$.) As $A(\pi)$ holds for all fair runs $\pi$ of $\Step{P}$ as defined above, we find that a fair run $\pi'$ where \TraRef{reset} is never executed will necessarily enable it once. But, as per Definition~\ref{def:weakbroadcasts}, all states $((\bot,\Any),\Any)$ are broadcast-initiating, so they cannot execute a neighbourhood transition and can only change their state via \TraRef{reset}. This has to occur eventually, so let step $i$ denote the first execution of \TraRef{reset}. 

Again, due to $A$, before step $i$ the number of agents in a state in $(\{\bot\}\times Q)\times Q$ is at most $k-1$. So we get $C_i'=((0,C_0(v)),C_0(v))$ or $C_i'(v)=((L,C_0(v)),C_0(v))$ for all nodes $v$, and the latter holds for at most $k-1$ (but not zero) nodes. (Recall that a weak broadcast might be executed simultaneously by multiple agents, so it is possible to end up with more than one token after executing \TraRef{reset}.) By induction, we eventually find a suffix of the run starting in a configuration $C_0'$ as defined above with $k=1$, i.e.\ exactly one agent is holding a token.

As we argued for $A$, no agent can ever reach an error state from such an initial configuration, so transition \TraRef{reset} will never be executed and it suffices to show that any fair run $\pi=(C_0,C_1,...)$ of $\Step{P}$ stabilises to the correct consensus, where $C_0(v)=(L,C_0^*(v))$ for some node $v$ and $C_0(u)=(0,C_0^*(u))$ for all other nodes $u\ne v$.

First, we argue that there is always at most one agent holding a token. Again, applying $g$ yields a reordering of an extension of a run of $\Token{P}^*$, and it is clear that the property holds for any run of $\Token{P}^*$ and any extension $\tau=(K_0,K_1,...)$ of such a run (with initial configuration $K_0$ defined s.t.\ $C_0(v)\in\{K_0(v)\}\times Q$ for all $v$). However, we still need to show that any reordering $\tau_f$ of $\tau$ also fulfils the property. It is only possible for a node $v$ to receive a token by moving from state $0$ or an intermediate state to $L$. If this happens, say, at step $i$ in $\tau$, then $v$ or an adjacent node $u$ must have left $\{L,L'\}$ at a step $j$ directly before $i$, i.e.\ a step $j<i$ s.t.\ in configurations $K_{j+1},K_{j+2},...,K_{i-1}$ there are no agents with a token. For the reordering we then have $f(j)<f(j)$, so the token leaves $u$ before entering $v$ as well in $\tau_f$.

This means that transition \TraRef{step} cannot be executed by multiple agents simultaneously and it thus updates the states in the same manner as in $P$. Finally, it remains to show that $\pi$ does so in a pseudo-stochastic manner, for which it is sufficient to prove that any configuration $C_i$ after executing $\TraRef{step}$ can, for any node $v$, reach a configuration $C'$ where $v$ has the token without executing $\TraRef{step}$.

Intuitively, this clearly holds, based on transitions \TraRef{token}. To make this formally precise, we reference the specific construction of Lemma~\ref{lem:simulate-rendezvous}. Starting with $C_i$ we repeatedly select agents in intermediate states (and execute the corresponding transition) until none are left. This will never select the (unique) node $v$ with $C_i(v)=(L,\Any)$, and it will terminate, due to the transitions of Lemma~\ref{lem:simulate-rendezvous}. Thus we reach a configuration $C'$ with $C'(v)=(L,\Any)$ and all other agents $u\ne v$ have $C'(u)=(0,\Any)$. (We have already argued that it is not possible to reach a state with more than one token.) As $C'$ contains no intermediate states, it is easy to see that there is a sequence of neighbourhood transitions to move the token to any node.
\end{proof}

\section{Proofs of Section~\ref{sec:degree-bounded}}
\subsection{\dAf\ can only decide $\clsCutoffOne$} \label{Appendix:dAf-bounded-degree}

\begin{proposition}
The set of labelling properties decided by \dAsf-automata in the $k$-degree-bounded case for \( k\geq 3\) is precisely $\clsCutoffOne$.
\end{proposition}

\begin{proof}
We know that $\clsCutoffOne$ is contained in the expressive power of \dAsF\ for $k$-degree-bounded graphs, since we can compute predicates in $\clsCutoffOne$ even in the unrestricted set of graphs.

Now let \( \varphi\) be a property decided by some \dAsf-automaton \(M\). We claim that for every multiset \(L\) and every label \(x\) with \( L(x) \geq 1 \) we have \( \varphi(L)=\varphi(L+x) \).

Proof of claim: since \(\varphi\) is a labelling property, we can choose the underlying graph. Let \(G=(V,E,\lambda)\) be a line labelled with the set \(L\) and the label \( x \) on the first end. Define the graph \(G'=(V',E',\lambda')\) by copying \(G\) and adding a extra node, which is labelled with \( x\) and connected to the second node only. Since \(M\) is consistent, \(M\) accepts \(G\) if and only if the synchronous run \( \rho \) on \(G\) is accepting and it accepts \(G'\) if and only if the synchronous run \( \rho'\) on \(G'\) is accepting. It follows by induction that every node of the graph \(G\) is always in the same state in both runs, and that the extra node is always in the same state as the first end. 

This shows that \(\rho\) is accepting if and only if \(\rho'\) is accepting. Therefore \(G\) is accepted if and only if \(G'\) is accepted, proving the claim.

Now we use the claim to prove the proposition. Let \(L\) be some multiset. We have to prove that \( \varphi(L)=\varphi(\Cutoff{L}{1})\). For this, we write \( L=\Cutoff{L}{1}+x_1+\dots+x_n \) with \( x_i(\Cutoff{L}{1})\geq 1\) and use the claim a total of \(n\) times.
\end{proof}

\subsection{\dAF\ and \DAF\ decide exactly the labelling properties in \(\NLinSpace\)} \label{Appendix:DAF-bounded-degree}

\begin{proposition}
A labelling property \(\varphi\) can be decided by a \dAsF-automaton in the $k$-degree-bounded case if and only if \(\varphi \in \NLinSpace\).

A labelling property \(\varphi\) can be decided by a \DAsF-automaton in the $k$-degree-bounded case if and only if \(\varphi \in \NLinSpace\).
\end{proposition}

\begin{proof}
By \cite[Proposition~22]{ER20}, the expressive power of \dAF-automata is equal to the expressive power of \DAF-automata in the $k$-degree-bounded case. It is therefore enough to consider \DAF.

We start by proving that labelling properties \(\varphi \in \NLinSpace\) can be decided. By \cite{BournezL13}, when restricting to $k$-degree-bounded graphs, graph population protocols can decide all symmetric properties \( \varphi \in \NLinSpace\), in particular all labelling properties \(\varphi \in \NLinSpace\), since they are by definition invariant under rearranging the labels. By Lemma~\ref{lem:simulate-rendezvous}, all properties decidable by graph population protocols can also be decided by \DAF-automata. 

Now let \( \varphi\) be a labelling property decided by a \DAsF-automaton \(M\) with counting bound \(\beta\). We have to prove that \( \varphi\) can be decided by a non-deterministic Turing machine with linear space. Since every node uses constant space and we have a linear number of nodes, a Turing machine with linear space can save configurations of our automaton \(M\). We claim that checking whether two configurations \(C,C'\) fulfil \(C \rightarrow C'\) can be checked in \(\NLinSpace\). For this, the Turing machine guesses for each node whether it has to be selected or not, and then checks for every node \(v\) whether 
\begin{align*}
C(v)&=C'(v) &&\text{if }v\notin S \\
\delta(C(v),\Cutoff{C(N(v))}{\beta})&=C'(v) &&\text{if }v\in S,
\end{align*}
i.e.\ the definition of the semantics. Since \(C \rightarrow C'\) can be checked in \(\NLinSpace\), \(C \rightarrow^{\ast} C'\) also can. For this, the Turing Machine does the following \((|Q|)^{|V|}\) times (upper bound on number of configurations): guess a configuration \(C''\) and check \(C \rightarrow C''\). Overwrite \(C\) with \(C''\). If \(C''=C'\) accept, if we finish the loop without this occurring reject. 

Now we use Immerman–Szelepcsényi theorem in the general version to obtain that \(C \not\rightarrow^{\ast} C'\) can also be checked in \(\NLinSpace\). Due to the automaton \(M\) using pseudo-stochastic fairness, we accept from some initial configuration \(C_0\) if and only if there exists a configuration \(C\) fulfilling the following three conditions:

\begin{enumerate}
\item \(C_0 \rightarrow^{\ast} C\).
\item \(C\) is accepting.
\item For all non-accepting configurations \(C'\), we have \(C \not\rightarrow^{\ast} C'\).
\end{enumerate}

We can check this in \(\NLinSpace\) by guessing the configuration \(C\) and checking the reachability conditions as described above.
\end{proof}

\subsection{\DAf\ can decide majority}

The proof of Lemma~\ref{lem:localcancelling2} will use the following lemma, which encapsulates the main argument.

\begin{lemma}\label{lem:localcancelling1}
Let $\pi=(C_0,C_1,...)$ denote a fair run of $\Cancel{P}$. There are only finitely many $C_i$ with $C_i(V)\cap\{k+1,...,E\}\ne\emptyset$ and $C_i(V)\cap\{-E,...,0\}\ne\emptyset$.
\end{lemma}
\begin{proof}
First, we note $C_i(V)\subseteq S\Rightarrow C_{i+1}(V)\subseteq S$ for $S=\{0,...,E\}$, $S=\{1,...,E\}$, and $S=\{-E,...,k\}$. In particular, the latter two imply that it suffices to show that there exists an $i$ with $C_i(V)\cap\{k+1,...,E\}=\emptyset$ or $C_i(V)\cap\{-E,...,0\}=\emptyset$.

Our proof will proceed by first showing that $C_i(V)\cap\{k+1,...,E\}\ne\emptyset$ and $C_i(V)\cap\{-E,...,-1\}\ne\emptyset$ cannot both hold for all $i$. Afterwards, we will argue that $C_i(V)\cap\{k+1,...,E\}\ne\emptyset$ and $0\in C_i(V)$ also cannot always hold, thus completing the proof.

Assume $C_i(V)\cap\{k+1,...,E\}\ne\emptyset$ and $C_i(V)\cap\{-E,...,-1\}\ne\emptyset$ for all $i$. We fix an $i$, and let $S_0(C_i):=C_i^{-1}(\{-E,...,0\})$ denote the set of agents with nonpositive contribution in $C_i$. Due to our assumption, $S_0(C_i)$ is nonempty. We write $S_d(C_i)$ for the set of nodes with distance $d$ to $S_0(C_i)$, for $d=1,...,n$, and define $\lambda_d(C_i):=\sum_{v\in S_d(C_i)}C_i(v)$ as the sum of contributions of $S_d$. Finally, we set $\lambda(C_i):=(\lambda_0(C_i),...,\lambda_n(C_i))$.

We now claim that $\lambda(C_i)<\lambda(C_{i+1})$ for each $i$, using lexicographical ordering, which is a contradiction, as there are only finitely many different configurations. To show the claim, we split the transition from $C_i$ to $C_{i+1}$ into a set of pairwise transactions $U\subseteq V\times V$, s.t.\ $C_{i+1}(v)=C_i(v)-|U\cap\{v\}\times V|+|U\cap V\times\{v\}|$ for each node $v$, and $C_i(u)>k\wedge C_i(v)\le k$ or $C_i(u)\ge-k\wedge C_i(v)<-k$ for all $(u,v)\in U$. Intuitively, $(u,v)\in U$ means that $u$ sends one unit to $v$.

We always have $C_i(u)>C_i(v)$ for $(u,v)\in U$, so $\lambda_0(C_i)\le\lambda_0(C_{i+1})$. If there exist adjacent nodes $u,v\in V$ with $C_i(u)<0<C_i(v)$ and $(v,u)\in U$ then we have $\lambda_0(C_i)<\lambda_0(C_{i+1})$ and our claim follows. Hence we will now exclude this case. In particular, we thereby exclude the possibility of a node $v$ leaving $S_0$, i.e.\ $v\in S_0(C_i)\setminus S_0(C_{i+1})$.

Let $U^+:=\{(u,v)\in U: u,v\notin S_0(C_i)\}$ denote the set of transitions where neither node has negative contribution. We pick $d,u$ where $u\in S_d$ and $C_i(u)>k$ s.t.\ $d$ is minimal. It is clear that $d>1$ holds, as else there would be a transaction from $u$ to a node in $S_0$. There is some node $v$ adjacent to $u$ in $S_{d-1}$ which, by choice of $u$, fulfils $C_i(u)\le k$. Therefore $(u,v)\in U^+$. In particular, $u$ sends one unit to $v$, thereby increasing $S_{d-1}$.

Moreover, all transactions $(u',v')\in U^+$ have $C_i(u')>k$, so it is not possible for any such transaction to decrease any $S_{d'}$ with $d'<d$. Neither can such a transaction change $S_0$. Therefore we find that the transactions in $U^+$ strictly increase $\lambda$, without affecting $S_0$.

Let $C$ denote the configuration where the transaction in $U^+$ have been executed, i.e.\ $C(v):=C_i(v)-|U^+\cap\{v\}\times V|+|U^+\cap V\times\{v\}|$. From the above considerations we get $\lambda(C)>\lambda(C_i)$ and $S_0(C)=S_0(C_i)$. The transactions in $U\cap S_0(C)\times S_0(C)$ do not change $\lambda_0$ and do not affect $S_0$ (a node could go from $-k-1$ to $-1$, but not further), so we now set $C'$ to the configuration after executing those.

Finally, consider a transaction $(u,v)\in U$ with $v\in S_0(C)\subseteq S_0(C')$ and $u\notin S_0(C)$. If $C'(u)>0$, then the transactions would  strictly increase $\lambda_0$, without changing $S_0$. If $C'(u)<0$, then neither $\lambda$ nor $S_0$ would change. Otherwise, the contribution of $u$ becomes $-1$, in which case $u$ would enter $S_0$, and $\lambda_0$ would remain unchanged. As a consequence of $u$ entering $S_0$, the distance between some other nodes and $S_0$ might decrease, but that can only increase $\lambda$, as all nodes outside of $S_0$ have nonnegative contribution. We can now proceed inductively, by updating $C'$ corresponding to $(u,v)$.

This concludes the first part of the proof. It remains to argue that $C_i(V)\cap\{k+1,...,E\}\ne\emptyset$ and $0\in C_i(V)$ cannot hold for all $i$. We argue analogously to before and assume the contrary. Then we set $S_0(C_i):=C_i^{-1}(0)$ and define $S_d$, $\lambda_d$ and $\lambda$ as before. It is not possible for a node to enter $S_0$, so a node can leave $S_0$ only finitely often. Choosing an $i$ large enough, the set $S_0$ thus does not change. Finally, we again pick $d,u$ where $u\in S_d$ and $C_i(u)>k$ s.t.\ $d$ is minimal, and see that $\lambda_{d-1}$ and thus $\lambda$ must increase at each step, which is a contradiction.
\end{proof}

\restatelocalcancellingtwo*
\begin{proof}
First, note that it suffices to show the claim for a single $i$, as $C_i(V)\subseteq\{-E,...,-1\}$ implies $C_{i+1}(V)\subseteq\{-E,...,-1\}$, and $C_i(V)\subseteq\{-k,...,k\}$ even implies $C_{i+1}=C_i$.

This then follows from Lemma~\ref{lem:localcancelling1} together with the following observation: \TraRef{cancel} is symmetric w.r.t.\ negation of all contributions, hence we could flip all signs, apply Lemma~\ref{lem:localcancelling1}, and derive the statement that there are only finitely many $C_i$ with $C_i(V)\cap\{-E,...,-k-1\}\ne\emptyset$ and $C_i(V)\cap\{0,...,E\}\ne\emptyset$.

As $0>\sum_vC_0(v)=\sum_vC_1(v)=...$, it is impossible that $C_i(V)\cap\{-E,...,0\}$ is empty, for any $i$. Hence Lemma~\ref{lem:localcancelling1} yields that we eventually have $C_i(V)\cap\{k+1,...,E\}=\emptyset$ for all sufficiently large $i$. Combining this with the above observation we get the desired statement.
\end{proof}

\restateDAfmajorityerrors*

We split the proof into two parts, Lemmata~\ref{lem:dafmajorityworksifnoerror} and \ref{lem:dafmajoritynotallperp}.

\begin{lemma}\label{lem:dafmajorityworksifnoerror}
Assuming that no agents enters state $\bot$, $\pi$ is accepting iff $\varphi(L_G)=1$.
\end{lemma}
\begin{proof}
If no weak broadcast is executed in $\pi$, then the computation is necessarily accepting ($\No$ is only reachable via \TraRef{reject}), so we can assume that $\varphi(C_\varphi)=0$. Additionally, we know that $\pi$ is a run of $\Detect{P}'$ as well, which simulates $\Detect{P}$, so there is a run $\tau$ of $\Detect{P}$ s.t.\ $\pi$ is a reordering of an extension of $\tau$. As \TraRef{detect} does not affect the first component, we get a run $\sigma=(K_0,K_1,...)$ of $\Cancel{P}$ by projecting $\tau$ onto the first component. Due to $\varphi(C_\varphi)=0$, Lemma~\ref{lem:localcancelling2} implies that any run of $\Cancel{P}$ starting at $K_0$ would eventually have only states in $\{-k,...,k\}$, or only states in $\{-E,...,-1\}$. In both cases, executing \TraRef{detect} would move a leader from $L$ to $L_\mathrm{double}$ or $L_\No$. 

In run $\pi$, it is not possible for a leader to leave state $L_\mathrm{double}$ or $L_\No$, as these states are broadcast initiating. This contradicts the weak fairness condition, as then either \TraRef{double} or \TraRef{reject} must be executed eventually.

Therefore, let $i$ denote the first step at which a weak broadcast is executed in $\pi$ (i.e.\ \TraRef{double} and/or \TraRef{reject}), and $M\subseteq V$ the set of its initiators. If there is a leader $v\notin M$, then it cannot be in state $\bot$, due to our assumption, nor can it be in $\No$, as \TraRef{reject} has not been executed before step $i$. But then $v$ would move to state $\bot$ in step $i$, which cannot happen by assumption. Hence $M$ is precisely the set of leaders.

If both \TraRef{double} and \TraRef{reject} are executed at step $i$, i.e.\ $(\Any,L_\mathrm{double}),(\Any,L_\No)\in C_i(M)$, then $C_{i+1}$ has all leaders in state $L$ or $\No$, with at least one in each. Additionally, $C_{i+1}$ is a valid input configuration of $\Detect{P}$ (it does not contain any intermediate states added in $\Detect{P}'$). Any fair run $\tau$ of $\Detect{P}$ starting in $C_{i+1}$ has one leader $v$ which starts in state $L$ and moves to $\bot$ upon the first execution of $\TraRef{detect}$ as there is an agent in $\No$. So $v$ enters neither $L_\mathrm{double}$ nor $L_\No$ in $\tau$. Therefore, until the second broadcast is executed at step $j>i$ in $\pi$, we have $C(v)\notin Q\times\{L_\mathrm{double},L_\No\}$ for any configuration $C\in\{C_{i+1},...,C_j\}$. If $j=\infty$, then $v$ moves eventually to $\bot$ in $\pi$, as it does in $\tau$, otherwise either $\TraRef{double}$ or $\TraRef{reject}$ move $v$ immediately to $\bot$. In both cases, our assumption is violated, so at step $i$ we cannot execute both \TraRef{double} and \TraRef{reject}.

Now there are two cases. If we execute only \TraRef{reject} at step $i$ of $\pi$, we know that $C_i(v)=(\Any,L_\No)$ for any leader $v$. This is only possible if \TraRef{detect} moves all leaders to $L_\No$ at once, so at some point a configuration in $\pi$ had only states in $\Last^{-1}(\{-E,...,-1\}\times Q_L)$, which neither \TraRef{detect} nor a transition of $\Cancel{P}$ can change.\footnote{It is clear that this holds for some reordering of $\pi$. To be entirely precise we would have to argue that it is impossible to reorder the steps at which the leaders enter $L_\No$ to before the steps where the other agents enter $\{-E,...,-1\}\times Q_L$.}
In particular, this means $C_i(V)\subseteq\Last^{-1}(\{-E,...,-1\}\times Q_L)$, so \TraRef{reject} would move all agents (including the leaders) to $\No$. At that point, no further transitions can be performed and the protocol moves into a stable $0$-consensus. This is correct, as it is only possible for $\Cancel{P}$ to move all agents to states $\{-E,...,-1\}$ if the sum of all contributions in $C_0$ is negative.

The second case is executing only \TraRef{double} at step $i$ of $\pi$. Similarly, this is only possible if all leaders move to $L_\mathrm{double}$ at once using \TraRef{detect}. For that to happen, all agents must be in states $\{-k,...,k\}\times\{0,L\}$ before executing \TraRef{detect}, moving the leaders to $L_\mathrm{double}$. It is not possible to execute \TraRef{detect} or any transition of $\Detect{P}$ with only these states, so we get $C_i(V)\subseteq\{-k,...,k\}\times\{0,L_\mathrm{double}\}$ as well and \TraRef{double} moves the agents back to states $Q\times\{0,L\}$ by doubling their contributions. Doubling every contribution does not change whether the sum is negative, so our claim follows inductively in this case, by considering the suffix $C_{i+1},C_{i+2},...$ . (Note that $C_0,...,C_i$ do not contain state $\No$. So if this case happens infinitely often, which occurs only if the sum of contributions is zero, $\pi$ is accepting.)
\end{proof}

\begin{lemma}\label{lem:dafmajoritynotallperp}
The run $\pi$ cannot reach a configuration with all leaders in state $\bot$.
\end{lemma}
\begin{proof}
If \TraRef{reject} is ever executed, then a leader (its initiator) enters state $\No$, from which it cannot enter $\bot$. Otherwise, it is not possible for any agent to enter $\No$, thus \TraRef{detect} cannot move an agent to $\bot$. Only \TraRef{double} remains, but it also leaves the leader initiating the broadcast in state $L$.
\end{proof}

\restateDAfmajorityiscorrect*
\begin{proof}
Let $\pi=C_0C_1...$ denote a fair run of $\Riset{P}$ starting in a configuration $C_0$ with $C_0(V)\subseteq\Cancel{Q}\times\{0,L\}$. Note that all valid initial configurations have this form. As before, we refer to agents starting in $((\Any,L),\Any)$ as \emph{leaders}. If no agent ever enters a state $((\bot,\Any),\Any)$, then \TraRef{reset} is never executed and Lemma~\ref{lem:dafmajorityworksifnoerror} implies that we reach a correct consensus.
If \TraRef{reset} is executed at some step $i$, we move to a configuration $C_{i+1}$ with only states $C_{i+1}(V)\subseteq\Cancel{Q}\times\{0,L\}$, i.e.\ a valid choice for $C_0$. Let $\pi':=C_{i+1}C_{i+2}...$ denote the suffix of $\pi$ starting at $i+1$.
Due to Lemma~\ref{lem:dafmajoritynotallperp} we know that at least one leader is not in state $\bot$ when executing \TraRef{reset}, so $\pi'$ has strictly fewer leaders than $\pi$, but at least one (the latter follows directly from the definition of \TraRef{reset}). Hence, we conclude that \TraRef{reset} is executed only finitely often.

It still remains to show that an agent entering $((\Any,\bot),\Any)$ at some point implies that \TraRef{reset} will be executed. This follows immediately, as all such states are broadcast-initiating and thus can only execute \TraRef{reset}.
\end{proof}

\end{document}